\documentclass[11pt]{article}
\usepackage[margin=1in]{geometry}
\usepackage[T1]{fontenc}
\usepackage{longtable}
\usepackage{lmodern}
\usepackage{amsmath,amssymb,amsthm,mathtools,booktabs,microtype,mathrsfs}
\usepackage[colorlinks=true,linkcolor=blue,citecolor=blue,urlcolor=blue]{hyperref}
\newtheorem{theorem}{Theorem}[section]
\newtheorem{lemma}[theorem]{Lemma}
\newtheorem{proposition}[theorem]{Proposition}
\newtheorem{corollary}[theorem]{Corollary}
\theoremstyle{remark}

\newcommand{\Z}{\mathbb Z}
\newcommand{\ssupp}{\mathcal S}
\newcommand{\MomentU}{\mathsf U}
\newcommand{\MomentT}{\mathsf T}
\title{The Golomb--Welch conjecture for lattice tilings}
\author{Bocong Chen$^1$,  Jing Huang$^2$ and Gaojun Luo$^3$\footnote{E-mail addresses: {\it
bocongchen@foxmail.com (B. Chen),
jhuangmath@foxmail.com (J. Huang),
gaojun\_luo@nuaa.edu.cn (G. Luo).
}}
}

\date{\small
$1.$ School of Mathematics, South China University of Technology, Guangzhou 510641, China\\
$2.$
School of Mathematics and Information Science,
Guangzhou University, Guangzhou 510006, China\\
$3.$ School of Mathematics, Nanjing University of Aeronautics and Astronautics, Nanjing 210016, China\\
}
\begin{document}
\maketitle
\begin{abstract}
We prove the Golomb--Welch conjecture for lattice tilings: for every $n\ge3$ and $r\ge2$, the integer lattice $\Z^n$ admits no lattice tiling by Lee balls $B_n(r)$ of radius $r$. Together with the classical constructions for $n\le2$ and for $r=1$, this determines all parameters for which a lattice Lee tiling exists.
Equivalently, we show that every abelian Cayley graph of even degree $2n\ge6$ and diameter $r\ge2$ has strictly fewer than $|B_n(r)|$ vertices.

\vspace{0.3cm}

\end{abstract}

\section{Introduction}\label{sec:introduction}

For $\mathbf x,\mathbf y\in\Z^n$, the \emph{Lee distance} is
$d_L(\mathbf x,\mathbf y)=\sum_{i=1}^n|x_i-y_i|$.
The Lee ball of radius $r$ centered at the origin is
\[
 B_n(r)=\{\mathbf x\in\Z^n:\|\mathbf x\|_1\le r\},
 \qquad
 |B_n(r)|=\sum_{j=0}^{\min(n,r)}2^j\binom nj\binom rj.
\]
A set $C\subseteq\Z^n$ is a \emph{perfect $r$-error-correcting Lee
code} if every point of $\Z^n$ lies at distance at most $r$ from
exactly one element of $C$. Equivalently,
$\Z^n=\bigsqcup_{\mathbf c\in C}(\mathbf c+B_n(r))$.
When $C=L$ is a subgroup of $\Z^n$, this partition is a
\emph{lattice tiling}, and the corresponding code is called linear.
The finite ball is then a complete set of representatives for the
finite quotient $\Z^n/L$.

The Lee metric arose in the study of nonbinary error-correcting
codes~\cite{Lee}. Its perfect codes ask when disjoint error balls can
fill the ambient space without leaving unused points. Golomb and
Welch connected this coding problem with tilings by unions of unit
cubes in their work of 1968 and 1970~\cite{GW1968,GW}. They constructed
lattice Lee tilings for every radius in dimensions one and two, and
for radius one in every dimension. They conjectured that these are
the only possibilities: for $n\ge3$ and $r\ge2$, there should be no
tiling of $\Z^n$ by translates of $B_n(r)$, even with an arbitrary
set of centers. 
In coding and graph-theoretic terms, the lattice case is of independent interest: it characterizes linear perfect Lee codes and is directly equivalent to the attainability of the Moore bound for even-degree abelian Cayley graphs in the degree--diameter problem. Moreover, the finite abelian quotient $\Z^n/L$ allows the tiling condition to be studied through group homomorphisms, character sums, and group-ring identities.

Early progress came from geometry. Golomb and Welch~\cite{GW}
proved nonexistence for all sufficiently large radii in each fixed
dimension $n=3$ or $n\ge5$, without specifying a threshold.
The same argument extends to $n=4$ when only translations are allowed;
see Horak and Kim~\cite[Remark~8]{HK}. Thus every fixed $n\ge3$
has a radius beyond which no Lee tiling exists. For
periodic tilings, Post~\cite{Post} obtained explicit bounds by
counting low-dimensional sections; in dimensions $n\ge6$, his
condition is
$r\ge n/\sqrt2-3/(2\sqrt2)-1/2$.
Lepist\"o~\cite{Lepisto} subsequently excluded periodic tilings when
$r\ge285$ and $n<(r+2)^2/2.1$.
Horak and Kim~\cite{HK} extended both bounds to arbitrary center
sets and developed a discrete linear-programming obstruction,
including nonexistence for
$r\ge18$ and $3r+21\le n\le r^2/2-20$.
These results control substantial regions in which the radius is
large relative to the dimension.

A complementary line of work treated small dimensions for all
radii. Gravier, Mollard, and Payan~\cite{GMP} settled dimension
three, and \v Spacapan~\cite{Spacapan} proved nonexistence of
face-to-face four-dimensional Lee tilings in a setting that allows
different radii. Horak~\cite{HorakTilings} established the conjecture
for $3\le n\le5$ and all $r\ge2$.
These geometric results apply beyond lattice tilings. Their scope
and the development of the conjecture are discussed in detail
in~\cite{HK}.

For lattice tilings, the finite quotient provides another route.
The standard homomorphism criterion states that a finite set
$B\subseteq\Z^n$ tiles by a lattice if and only if there are a
finite abelian group $A$ of order $|B|$ and a homomorphism
$\pi:\Z^n\to A$ whose restriction to $B$ is bijective;
see~\cite[Theorem~6]{HA}. Horak and Gro\v sek~\cite{HG} used an invariant of
finite abelian groups to reformulate the problem and to exclude
linear perfect Lee codes of radius two in dimensions
$3\le n\le12$. Kim~\cite{Kim} introduced a symmetric-polynomial
argument for radius two in certain dimensions with prime ball
cardinality. Zhang and Ge~\cite{ZG} obtained arithmetic criteria
excluding lattice tilings in infinitely many dimensions at radii
three and four. Qureshi~\cite{Qureshi} proved nonexistence in
infinitely many dimensions for each radius in a set of natural
density one, and developed a more general criterion based on
$Q$-polynomials.

Group-ring methods made the quotient structure more explicit.
Zhang and Zhou~\cite{ZZ} derived group-ring equations for radii
two and three and used characters and algebraic number theory to
exclude further infinite families of dimensions. Leung and
Zhou~\cite{LZ} then settled the lattice case at radius two for
every $n\ge3$. Their result also established nonattainment of the
even-degree abelian Cayley Moore bound at diameter two.
For general radii, Xiao and Zhou~\cite{XZ} determined all
coefficients of the associated $Q$-polynomials and obtained
necessary and sufficient group-ring conditions for lattice Lee
tilings. In particular, they proved positive-density
nonexistence results in the dimension for each $3\le r\le50$.
Their computation treats the $1{,}358$ pairs satisfying
$n\ge3$, $3\le r\le12$, and $|B_n(r)|<2^{30}$, and excludes
$1{,}356$ of them. Combined with the known results outside this
range, their Theorem~5.1 excludes every pair with $n\ge3$,
$r\ge2$, and $|B_n(r)|<2^{30}$ except $(35,3)$ and $(805,3)$. 
Bhattacharjee and Bhattacharya~\cite[Theorem~1.1]{BB} subsequently
proved nonexistence for those two pairs. We identify the finite cases already covered by Xiao and Zhou and
include direct proofs of all such cases used below.
Within this algebraic approach, the remaining difficulty is to rule out all radii $r\ge3$ and dimensions $n\ge3$ uniformly, 
without imposing arithmetic hypotheses on $|B_n(r)|$. We resolve this problem as follows.

\begin{theorem}\label{thm:main}
For all integers $n\ge3$ and $r\ge2$, there is no lattice tiling
of $\Z^n$ by translates of $B_n(r)$.
\end{theorem}

Combined with the classical
constructions~\cite[Theorems~1 and~3]{GW}, Theorem~\ref{thm:main} gives the exact
existence criterion for positive integers $n,r$: a lattice tiling
by $B_n(r)$ exists if and only if $n\le2$ or $r=1$.
It also excludes linear perfect Lee codes over every alphabet
$\Z/m\Z$ with $m\ge2r+1$ in the conjectured range: the inverse
image of such a code under $\Z^n\to(\Z/m\Z)^n$ would be a lattice
tiling by $B_n(r)$.

The graph-theoretic consequence can be stated precisely.
Let $AC(2n,r)$ denote the largest order of a finite simple
undirected abelian Cayley graph of degree $2n$ and diameter at
most $r$. The abelian Cayley Moore bound of Dougherty and
Faber~\cite{DF} is $AC(2n,r)\le |B_n(r)|$.
Its equality case is equivalent to a lattice Lee tiling;
see~\cite[Section~2.1]{ZZ} and the discussion preceding
\cite[Corollary~1.3]{LZ}. For the broader degree--diameter problem,
see~\cite{MS}.

\begin{corollary}\label{cor:abelian-moore}
For all integers $n\ge3$ and $r\ge2$,
\[
 AC(2n,r)\le
 \sum_{j=0}^{\min(n,r)}2^j\binom nj\binom rj-1.
\]
For positive integers $n$ and $r$, the even-degree abelian Cayley
Moore bound is attained if and only if $n\le2$ or $r=1$.
\end{corollary}

Corollary~\ref{cor:abelian-moore} determines when the even-degree abelian Cayley Moore bound is attained. Determining the exact maximum orders below this bound is a separate problem.
We briefly outline the proof of Theorem~\ref{thm:main}. 
After applying the known results for $n\le5$, $r=2$, and large
radii relative to $n$, we use a weighted convolution inequality on
four signed-permutation orbits to exclude $n\le N_r$, where
$N_3=7$ and $N_r=r(r+1)/2+2$ for $r\ge4$.
For $n>N_r$, we work with a hypothetical quotient homomorphism
$\pi:\Z^n\to A$ that is bijective on $B_n(r)$.
For each boundary representative
$\mathbf b\in B_n(r)\setminus B_n(r-1)$, its preimages in the
next shell induce an exact partition of the signed coordinate
directions outside the signed support of $\mathbf b$.
The resulting identities retain both repeated coordinates and
cancellation. They yield integral moment bounds and weighted
collision estimates for the kernels of scalar multiplication
on $A$.
Dividing $2n$ by $r+1$, we treat the resulting remainder classes---including the zero and terminal remainders---by combining these boundary and collision estimates with character bounds for scalar images and with prime-quotient reductions at small radii. The finite cases are treated by explicit polynomial identities, recurrences, and modular calculations; at radius three, we also use a known nonexistence criterion.

The paper is organized as follows.
Section~\ref{sec:preliminaries} records the known inputs and fixes
notation. Section~\ref{sec:boundary-methods} establishes the
boundary identities, moment inequalities, collision bounds, and
character estimates. Section~\ref{sec:uniform-exclusion} gives
the signed-orbit obstruction and the exclusions for the intermediate,
zero, and terminal remainders. Section~\ref{sec:small-radii} treats
the remaining zero-remainder case at radius three and completes the proof.

\section{Preliminaries}\label{sec:preliminaries}
We recall the known inputs and fix notation for a hypothetical lattice tiling.
Horak~\cite[Theorem~6]{HorakTilings} excludes dimensions three through five.
\begin{lemma}\label{lem:known-small}
For $3\le n\le5$ and $r\ge2$, there is no tiling of $\Z^n$
by translates of $B_n(r)$.
\end{lemma}

Leung and Zhou~\cite[Theorem~1.1]{LZ} settle the lattice case at radius two.
\begin{lemma}\label{lem:known-radius-two}
For every $n\ge3$, the ball $B_n(2)$ admits no lattice tiling
of $\Z^n$.
\end{lemma}

For larger radii, we use the following extension of Post's bound~\cite{Post} from periodic tilings to arbitrary center sets, established by Horak and Kim~\cite[Theorem~13]{HK}.
\begin{lemma}\label{lem:known-post}
If $n\ge6$ and
$r\ge\frac n{\sqrt2}-\frac{3}{2\sqrt2}-\frac12$,
then there is no tiling of $\Z^n$ by translates of $B_n(r)$.
\end{lemma}

Bold lowercase letters denote coordinate vectors; elements of the quotient group are not bold.
For a hypothetical lattice tiling, the quotient map $\pi:\Z^n\longrightarrow A=\Z^n/L$
restricts to a bijection on $B_n(r)$. Thus $B_n(r)$ is a complete
set of coset representatives, also called a \emph{transversal}. Conversely,
any such homomorphism defines a lattice tiling by its kernel. 
Write $g_i=\pi(\mathbf e_i)$, where
$\mathbf e_1,\ldots,\mathbf e_n$ are the standard basis vectors.
Every nonzero vector in $L$ has norm at least $2r+1$, since
$B_n(r)-B_n(r)=B_n(2r)$.
For a real number $t$, write $t_+=\max\{t,0\}$.

For nonnegative integers $j$, write $B_j=B_n(j)$ and
$S_j=B_j\setminus B_{j-1}$. Their cardinalities are
$V_j(n)=|B_n(j)|$ and $\sigma_j(n)=|S_j|$.
For a formal series $F(z)$, write $[z^j]F$ for its coefficient of $z^j$.
We use $\binom ab=0$ for integers $a\ge0$ and $b<0$ or $b>a$.
We omit $n$ from $V_j(n)$ and $\sigma_j(n)$ when the dimension is
fixed. Balls and shells of negative index are empty, their cardinalities are
zero, and $\sigma_0=1$. Choosing support, signs, and positive
coordinate magnitudes gives
\begin{align*}
 V_j(n)&=\sum_{a=0}^{\min(n,j)}2^a\binom na\binom ja,
 &\sum_{j\ge0}V_j(n)z^j&=\frac{(1+z)^n}{(1-z)^{n+1}},
 \\
 \sigma_j(n)&=\sum_{a=1}^{\min(n,j)}2^a\binom na\binom{j-1}{a-1}
 \quad(j\ge1),
 &\sum_{j\ge0}\sigma_j(n)z^j&=\left(\frac{1+z}{1-z}\right)^n.
\end{align*}
In particular $|A|=V_r(n)$ is odd. We use $V_j(0)=1$ for $j\ge0$.
Logarithmic differentiation yields, for $j\ge1$,
$
 j\sigma_j=2n\sigma_{j-1}+(j-2)\sigma_{j-2}
$
and 
$
 jV_j=(2n+1)V_{j-1}+(j-1)V_{j-2}.
$
Throughout the boundary and exclusion arguments, $q$ and $c$ are the
quotient and remainder determined by
\[
 2n=(r+1)q+c,\qquad q\ge0,\qquad 0\le c\le r.
\]
For $\mathbf x\in\Z^n$, its signed support is
$\ssupp(\mathbf x)=\{\operatorname{sgn}(x_i)\mathbf e_i:x_i\ne0\}$,
with size $s(\mathbf x)=|\ssupp(\mathbf x)|$. We write
$\boldsymbol\varepsilon(\mathbf x)
=\sum_{\mathbf a\in\ssupp(\mathbf x)}\mathbf a$ for the sum of its
signed support directions.
Thus $\ssupp(\mathbf x)\subseteq S_1=\{\pm\mathbf e_i:1\le i\le n\}$.
For $\mathbf a=\pm\mathbf e_i\in S_1$, write $x_{\mathbf a}=\pm x_i$
for the corresponding signed coordinate.

We also recall the Fourier conventions on $A$.
We write $\mathbf1_E$ for the indicator of a set or condition $E$.
For an integer $j$ and a character $\chi$, write
$\chi^j(a)=\chi(ja)$ for $a\in A$.
For $X=\sum_{a\in A}c_a[a]\in\mathbb C[A]$, put
$X^{(k)}=\sum_a c_a[ka]$ and denote its identity coefficient by $[0]X$.
We identify $X$ with its coefficient function $a\mapsto c_a$ on $A$.
In group-algebra expressions and character evaluations, a finite
centrally symmetric set $E\subseteq\Z^n$ denotes its formal image sum
$\sum_{\mathbf x\in E}[\pi(\mathbf x)]$. Multiplicities are retained
under scalar multiplication.
If $\widehat A$ is the character group and $\mathbf1$ its trivial
character, orthogonality gives
\[
 [0]X=\frac{1}{V_r(n)}\sum_{\chi\in\widehat A}\chi(X),\qquad
 \sum_{a\in A}|c_a|^2=\frac{1}{V_r(n)}\sum_{\chi\in\widehat A}|\chi(X)|^2.
\]
For a function $u:A\to\mathbb C$, define its Fourier transform by
$\widehat u(\chi)=\sum_{a\in A}u(a)\chi(a)$.
For functions $u,v:A\to\mathbb C$, use the unnormalized inner product
$\langle u,v\rangle=\sum_{a\in A}u(a)\overline{v(a)}$ and write
$\|u\|_2^2=\langle u,u\rangle$.
The same conventions apply to subgroups and quotient groups of $A$.
For functions on an additive group, convolution is $(u*v)(a)=\sum_bu(b)v(a-b)$ whenever the sum is defined.
For the character values of the radius-one ball, put
$\lambda_\chi=\chi(B_1)=1+\sum_i(\chi(g_i)+\chi(-g_i))$ and
$\Lambda=\max_{\chi\ne\mathbf1}|\lambda_\chi|$.
For an integer $a$ and $E\subseteq A$, put $aE=\{ax:x\in E\}$ and
$A[a]=\{x\in A:ax=0\}$. Then $|A[a]|\,|aA|=V_r(n)$.

For completeness, we record the elementary prime-characteristic
identities used in the scalar arguments.
For a prime $p$ and a nonzero integer $a$, $v_p(a)$ denotes its
$p$-adic valuation. If $a=\sum_i a_i p^i$ and $b=\sum_i b_i p^i$
are nonnegative integers in base $p$, then Lucas's congruence is
$\binom ab\equiv\prod_i\binom{a_i}{b_i}\pmod p$.
Lucas's congruence also factors a ball volume into its digit contributions.
If $p$ is prime and $u=\sum_i u_ip^i$, $v=\sum_i v_ip^i$ are their
base-$p$ expansions, then
$V_v(u)\equiv\prod_iV_{v_i}(u_i)\pmod p$.
Indeed, $2^{p^i}\equiv2\pmod p$ and Lucas's congruence give
\[
 V_v(u)\equiv\sum_{(j_i)}\prod_i
       2^{j_i}\binom{u_i}{j_i}\binom{v_i}{j_i}
 =\prod_i\left(\sum_{j_i=0}^{p-1}
       2^{j_i}\binom{u_i}{j_i}\binom{v_i}{j_i}\right)
 =\prod_iV_{v_i}(u_i)\pmod p.
\]
For a finite set $Y$ in a real vector space, write
$tY=\{\mathbf y_1+\cdots+\mathbf y_t:\mathbf y_i\in Y\}$.
The following bounds are consequences of
Matolcsi and Ruzsa~\cite[Theorem~1.5]{MR}.
\begin{lemma}\label{lem:planar-sumset}
If $Y$ has affine dimension at least two, then
$|2Y|\ge3(|Y|-1)$ and $|3Y|\ge4(|Y|-1)$.
\end{lemma}
Indeed, apply the cited theorem with $A=B$ and $k=1,2$
after a generic planar projection that is injective on $Y$
and has a noncollinear image. It gives $|2Y|\ge3|Y|-3$ and
$|3Y|\ge6|Y|-8\ge4|Y|-4$, since $|Y|\ge3$.

\section{Boundary fibers, moment inequalities, and spectral bounds}
\label{sec:boundary-methods}
Assume $r\ge3$. We establish the intersection, boundary, and character estimates used in the remainder exclusions.

\subsection{Mixed-radius intersections}
For integer radii $a,b$ and $\mathbf v\in\Z^n$,  put $C_{a,b}(\mathbf v)=|B_n(a)\cap(\mathbf v+B_n(b))|$, with $B_n(a)=\varnothing$ for $a<0$. Reflection about $\mathbf v/2$ gives $C_{a,b}(\mathbf v)=C_{b,a}(\mathbf v)$. The next lemma compares these intersections at adjacent radii.

For $\|\mathbf v\|_1\ge3$, put
\[
 \kappa(\mathbf v)=
 \begin{cases}3/2,&3\le\|\mathbf v\|_1\le4,\\4/3,&\|\mathbf v\|_1\ge5.
 \end{cases}
\]
\begin{lemma}\label{thm:lens-three-halves}

For $\|\mathbf v\|_1\ge3$ and $k\ge0$,
$C_{k,k}(\mathbf v)\le\kappa(\mathbf v)C_{k-1,k+1}(\mathbf v)$.
If $\mathbf v$ has at least two nonzero coordinates, then
\begin{equation}\label{eq:lens-curvature}
 2C_{a,b}(\mathbf v)\le\kappa(\mathbf v)
       \bigl(C_{a-1,b+1}(\mathbf v)+C_{a+1,b-1}(\mathbf v)\bigr)
       \qquad(a,b\ge0).
\end{equation}
If $\mathbf v$ has one nonzero coordinate, the two intersections in the
first assertion have equal sizes.
\end{lemma}
\begin{proof}
Sign changes permit $v_i\ge0$. If only $v_1=m\ge3$ is nonzero,
the one-dimensional intersection counts satisfy, for every integer $j$,
$C_{j,j}(m)=C_{j-1,j+1}(m)=\max\{0,2j-m+1\}$.
For each fixed $\mathbf y\in\Z^{n-1}$, the residual radii are obtained
by setting $j=k-\|\mathbf y\|_1$; when $n=1$, the set
$\Z^0$ consists of its single zero vector. Summing this equality over
$\mathbf y$ proves the assertion for support one.
In dimension two write $\mathbf v=(\alpha,\beta)$, $\alpha\ge \beta\ge1$, and put
\[
 M=a+b,\quad m=\alpha+\beta\ge3,\quad
 P=M-m+1,\quad Q=M-(\alpha-\beta)+1=P+2\beta.
\]
If $M<m$, all three intersections in~\eqref{eq:lens-curvature}
are empty. Suppose $M\ge m$. The coordinates
$u=x_1+x_2$, $w=x_1-x_2$ identify the integer lattice with
$\{(u,w)\in\Z^2:u\equiv w\pmod2\}$ and transform a Lee ball
into a square. By symmetry assume $a\le b$. The two intersection
intervals have lower endpoints $\max\{-a,m-b\}$ and
$\max\{-a,\alpha-\beta-b\}$, upper endpoint $a$, and respective integer lengths
$\min\{2a+1,P\}$ and $\min\{2a+1,Q\}$.
If one length is even, exactly half the rectangle is admissible.
If both lengths are odd, their lower endpoints have the same parity:
this follows from $m\equiv\alpha-\beta\pmod2$ when the same kind of endpoint
is selected; in the mixed case the endpoints are $m-b,-a$, and
oddness of $P=a+b-m+1$ gives $m-b\equiv-a\pmod2$.
Thus, with $j=\min(a,M-a)$,
\begin{equation*}
 C_{a,M-a}(\alpha,\beta)=f_j,
 \qquad f_j=\left\lceil
       \frac{\min\{2j+1,P\}\min\{2j+1,Q\}}{2}\right\rceil.
\end{equation*}

We first prove the stronger inequality
$3f_j\le2(f_{j-1}+f_{j+1})$ away from a single even central
index; at an odd central pair the right neighbor means $f_j$.
For $M=3$ the entire sequence is $(1,2,2,1)$, which satisfies
this inequality, with outside values zero. Let $M\ge4$.
At an endpoint, $f_0=1$ and $f_1\ge2$. At a noncentral index
$j=1$, the cases $P=1,2,\ge3$ give
\[
 (f_1,f_2)=(2,\ge2),\ (3,\ge4),\ (5,\ge8),
 \qquad 3f_1\le2(1+f_2).
\]
For $j\ge2$, if $P\le2j-1$, the first
factor is capped, and
\[
 \frac{f_{j-1}}{f_j}\ge\frac{2j-1}{2j+1}\ge\frac{3}{5},
 \qquad f_{j+1}\ge f_j,
 \qquad 2(f_{j-1}+f_{j+1})\ge\frac{16}{5}f_j>3f_j.
\]
For the first ratio, if $Q\le 2j-1$ it equals one; if $Q=2j$,
then $P\equiv Q\pmod2$ is even and the ratio is
$(2j-1)/(2j)$; if $Q\ge2j+1$, the products $P(2j-1)$ and $P(2j+1)$
have the same parity. For even products the ratio is $(2j-1)/(2j+1)$;
for odd products it is
$(P(2j-1)+1)/(P(2j+1)+1)\ge(2j-1)/(2j+1)$.

For $P\ge2j$, one has $2f_{j-1}=4j^2-4j+2$ and the
following exact values or lower bounds:
\[
\begin{array}{c|c|c|c}
 P&2f_j&2f_{j+1}\text{ at least}
       &2(f_{j-1}+f_{j+1})-3f_j\text{ at least}\\\hline
 2j&2j(2j+1)&4j(j+1)&2j^2-3j+2\\
 2j+1&4j^2+4j+2&4(j+1)^2&2j^2-2j+3\\
 2j+2&4j^2+4j+2&(2j+2)(2j+3)&2j^2+5\\
 P\ge2j+3&4j^2+4j+2&4j^2+12j+10&2j^2+2j+9
\end{array}
\]
The last column is obtained by adding the preceding two neighbor
values and subtracting $3f_j$. Its entries are positive for $j\ge2$:
$2j^2-3j+2=j(2j-3)+2>0$, $2j^2-2j+3=2j(j-1)+3>0$,
and the last two expressions have positive coefficients.
At either odd central index, $2j+1=M$ and
$P=M-m+1\le 2j-1$, so the earlier capped case applies.
Reflection handles the descending indices.

At the single central index $M=2\ell$, one has $P\le M-2$,
so the first factor is capped for both $f_{\ell-1}$ and $f_\ell$.
If $\alpha-\beta\ge2$, both factors are capped and their ratio is one.
If $\alpha-\beta=1$, $P$ is even and
$\frac{f_\ell}{f_{\ell-1}}=\frac{M}{M-1}\le\frac{4}{3}$
for $M\ge4$.
If $\alpha=\beta$, then $m\ge4$ is even, $P$ is odd, and
$\frac{f_\ell}{f_{\ell-1}}=\frac{P(M+1)+1}{P(M-1)+1}$.
For $m=4$, the bound $f_\ell/f_{\ell-1}\le3/2$ is equivalent
to $P(M-5)+1\ge0$: at $M=4$, $P=1$ gives zero, and for
$M\ge6$ both summands are positive. For $m\ge6$, the bound
$f_\ell/f_{\ell-1}\le4/3$ is equivalent to $P(M-7)+1\ge0$:
at $M=6$, $P=1$ gives zero, and for $M\ge8$ it is positive.
These are exactly the central inequalities
$2f_\ell\le2\kappa(\mathbf v)f_{\ell-1}$.
The noncentral inequality already proved has constant $4/3$
and hence also holds with $3/2$. This establishes
\eqref{eq:lens-curvature} in dimension two.
To pass to arbitrary support, put
\begin{align*}
 &H_j(s,t)=\sum_{x\in\Z}s^{|x|}t^{|x-j|},\qquad
 G_{\mathbf v}(s,t)=\frac{\prod_iH_{v_i}(s,t)}{(1-s)(1-t)},\\
 &K_3=3(s^2+t^2)-4st,\qquad ~~K_5=2(s^2+t^2)-3st.
\end{align*}
Then  $[s^at^b]G_{\mathbf v}=C_{a,b}(\mathbf v)$, and
\begin{align*}
&[s^{a+1}t^{b+1}]K_3G_{\mathbf v}
=3(C_{a-1,b+1}+C_{a+1,b-1})-4C_{a,b},\\
& [s^{a+1}t^{b+1}]K_5G_{\mathbf v}
=2(C_{a-1,b+1}+C_{a+1,b-1})-3C_{a,b}.
\end{align*}
Coefficients with either exponent zero are nonnegative directly.
All $H_j$ have nonnegative coefficients, so a two-coordinate
inequality remains valid after multiplication by the other $H_j$.

For the constant $3/2$, any positive coordinate at least two
and any other positive coordinate supply the required pair.
Otherwise there are at least three coordinates equal to one, and
\[
 H_1=\frac{s+t}{1-st},\qquad
 K_3(s+t)^3=3s^5+5s^4t+5st^4+3t^5
\]
proves coefficientwise nonnegativity after multiplication by the
remaining nonnegative factors.
For the constant $4/3$, a pair of positive coordinates of sum at
least five supplies the two-dimensional inequality. In the remaining
case, the positive coordinates contain one of the five patterns
\[
 (3,1,1),\quad(2,2,2),\quad(2,2,1),\quad(2,1,1,1),\quad(1,1,1,1,1).
\]
Indeed, without a pair of positive coordinates of sum at least five,
the maximum coordinate is at most three. A coordinate three has at
least two accompanying ones; three coordinates two give $(2,2,2)$;
two coordinates two require an accompanying one; one coordinate two
requires at least three ones; and otherwise at least five ones occur.
We give explicit algebraic certificates for all five.
Only $H_1,H_2,H_3$ occur in these patterns. Write $H_j=Q_j/(1-st)$, where
\[
 \begin{aligned}
 Q_1&=s+t,\qquad Q_2=s^2+st+t^2-s^2t^2,\\
 Q_3&=s^3+s^2t+st^2+t^3-s^2t^2(s+t).
 \end{aligned}
\]
For a pattern $\mathbf w$, put $R_{\mathbf w}=K_5\prod_iQ_{w_i}$.
The following identities suffice, since each quotient
$(1-s^at^b)/((1-s)(1-t))$ is coefficientwise nonnegative for
nonnegative $a,b$: its coefficient at $s^it^j$ is
$1-\mathbf1_{i\ge a,\,j\ge b}\ge0$.
\[
\begin{aligned}
	R_{311}={}&s^6(2s+3t)(1-t^2)+t^6(3s+2t)(1-s^2)
	+s^2t^2(s+t)(s^2+st+t^2+s^2t^2),\\
	R_{2111}={}&2s^7(1-t^2)+2t^7(1-s^2)
	+s^6t(5-3t^2)+st^6(5-3s^2)\\
	&+4s^2t^2(s^3+t^3)+s^3t^3(s+t)(1+st),\\
	R_{221}={}&2s^7(1-t^2)+2t^7(1-s^2)
	+st(s+t)(s^4+t^4)(3-2st)\\
	&+s^3t^3(s+t)(1-s^2t^2)+2s^4t^4(s^3+t^3),\\
	R_{11111}={}&2(s^7+t^7)+7st(s+t)(s^4+t^4).
\end{aligned}
\]
For $\mathbf w=(2,2,2)$, use the following nonnegative matrix,
with rows and columns indexed by $0,\ldots,8$:
\[
(c_{ij})_{0\le i,j\le8}=\begin{pmatrix}
0&0&0&0&0&0&0&0&2\\
0&0&0&0&0&0&0&3&5\\
0&0&0&0&0&0&5&8&4\\
0&0&0&0&0&2&7&7&3\\
0&0&0&0&3&5&4&4&6\\
0&0&0&2&5&10&9&6&8\\
0&0&5&7&4&9&11&8&8\\
0&3&8&7&4&6&8&8&8\\
2&5&4&3&6&8&8&8&8
\end{pmatrix}.
\]
Direct multiplication by $(1-s)(1-t)$ verifies the identity
\[
\begin{split}
\frac{K_5Q_2^3}{(1-s)(1-t)}={}&
 \sum_{i,j=0}^7 c_{ij}s^it^j
 +\frac{t^8}{1-t}\sum_{i=0}^7 c_{i8}s^i+\frac{s^8}{1-s}\sum_{j=0}^7 c_{8j}t^j
 +\frac{8s^8t^8}{(1-s)(1-t)}.
\end{split}
\]
Thus $R_{\mathbf w}/((1-s)(1-t))$ is coefficientwise nonnegative for all five
patterns. Multiplying by one factor $(1-st)^{-1}$ for each
coordinate of the pattern, and by the remaining $H_{v_i}$,
establishes the general curvature inequality.

Both constants now hold for every displacement of the corresponding
norm. Setting $a=b=k$ in~\eqref{eq:lens-curvature} and using
$C_{k+1,k-1}=C_{k-1,k+1}$ proves the equal-radius assertions.
\end{proof}

\subsection{The boundary partition and its scalar representative}
For $\mathbf b\in S_r$, let
$F_{\mathbf b}=\{\mathbf x\in S_{r+1}:\pi(\mathbf x)=\pi(\mathbf b)\}$
and $f_{\mathbf b}=|F_{\mathbf b}|$. Define the defect and residual sum by
\[
 D_{\mathbf b}=(r+1)f_{\mathbf b}-2n+s(\mathbf b),\qquad
 \boldsymbol\delta_{\mathbf b}
 =\sum_{\mathbf x\in F_{\mathbf b}}
       (\mathbf x-\boldsymbol\varepsilon(\mathbf x)).
\]
Uniqueness of representatives in $B_r$ gives the following partition.
\begin{lemma}\label{lem:exact-fiber}
The signed supports of the points of 
$F_{\mathbf b}$ partition
$S_1\setminus\ssupp(\mathbf b)$. Consequently,
$
 D_{\mathbf b}
=\sum_{\mathbf x\in F_{\mathbf b}}(r+1-s(\mathbf x))\ge0,
$
$\|\boldsymbol\delta_{\mathbf b}\|_1\le D_{\mathbf b}$
and 
$
 (f_{\mathbf b}+1)\pi(\mathbf b)
 =\pi\bigl(\mathbf b-\boldsymbol\varepsilon(\mathbf b)
                         +\boldsymbol\delta_{\mathbf b}\bigr).
$
Every point of $S_{r+1}$ has its representative in $S_r$.
\end{lemma}
\begin{proof}
If $\mathbf x\in F_{\mathbf b}$ and $\mathbf a\in\ssupp(\mathbf x)$,
then $\mathbf x-\mathbf a\in B_r$ represents $\pi(\mathbf b-\mathbf a)$.
For fixed $\mathbf a$, uniqueness in $B_r$ permits at most one such
$\mathbf x$. If also $\mathbf a\in\ssupp(\mathbf b)$, then
$\mathbf b-\mathbf a\in B_r$ has the same image, forcing
$\mathbf x=\mathbf b$, a contradiction.
Conversely, let $\mathbf a\notin\ssupp(\mathbf b)$ and choose
$\mathbf w\in B_r$ with $\pi(\mathbf w)=\pi(\mathbf b-\mathbf a)$.
If $\mathbf w+\mathbf a\in B_r$, it equals $\mathbf b$, contrary to
$\|\mathbf b-\mathbf a\|_1=r+1$. Thus
$\|\mathbf w+\mathbf a\|_1=r+1$, which also forces
$\mathbf a\in\ssupp(\mathbf w+\mathbf a)$.
This proves the partition and gives
\[
 \sum_{\mathbf x\in F_{\mathbf b}}s(\mathbf x)=2n-s(\mathbf b),
 \qquad
 \sum_{\mathbf x\in F_{\mathbf b}}\boldsymbol\varepsilon(\mathbf x)
 =-\boldsymbol\varepsilon(\mathbf b).
\]
Subtracting support sizes from norms proves $
 D_{\mathbf b}
=\sum_{\mathbf x\in F_{\mathbf b}}(r+1-s(\mathbf x))\ge0
$
and 
$\|\boldsymbol\delta_{\mathbf b}\|_1\le D_{\mathbf b}$;
summing images proves $
 (f_{\mathbf b}+1)\pi(\mathbf b)
 =\pi\bigl(\mathbf b-\boldsymbol\varepsilon(\mathbf b)
                         +\boldsymbol\delta_{\mathbf b}\bigr)
$.
Finally, an exterior point with representative in $B_{r-1}$ would
give a nonzero vector of $L$ of norm at most $2r$.
\end{proof}

The sums needed below follow from the support distribution of a Lee shell.
For $k\ge0$, with negative-index shells empty, one has
\begin{equation}\label{eq:support-first-moment}
	\begin{split}
		&\sum_{\mathbf x\in S_k}(k-s(\mathbf x))=(k-1)\sigma_{k-1} \\
		&\sum_{\mathbf x\in S_k}(k-s(\mathbf x))(k-s(\mathbf x)-1) =(k-1)(k-2)\sigma_{k-2}.
	\end{split}
\end{equation}
For $k=0,1$, both moment identities hold directly. The second
identity also holds directly for $k=2$, because
$(2-s(\mathbf x))(1-s(\mathbf x))=0$ for every $\mathbf x\in S_2$.
For $k\ge2$, the number of vectors of support size $a$ is
$2^a\binom na\binom{k-1}{a-1}$, and
\[
 (k-a)\binom{k-1}{a-1}=(k-1)\binom{k-2}{a-1}.
\]
For $k\ge3$, one also has
\[
 (k-a)(k-a-1)\binom{k-1}{a-1}
 =(k-1)(k-2)\binom{k-3}{a-1}.
\]
Multiplying by $2^a\binom na$ and summing over $a$ proves the
asserted identities.
Since the fibers partition $S_{r+1}$, these identities and the shell
recurrence give
\begin{equation}\label{eq:boundary-total-defect}
 \sum_{\mathbf b\in S_r}D_{\mathbf b}
 =(r+1)\sigma_{r+1}-2n\sigma_r+\sum_{\mathbf b\in S_r}s(\mathbf b)
 =r\sigma_r.
\end{equation}

To control the fiber sizes, use $2n=(r+1)q+c$ with $0\le c\le r$.
We call fibers of size $q$ minimum fibers and write
$E=|\{\mathbf b\in S_r:f_{\mathbf b}=q\}|$ for their number.
For each such fiber put
$\boldsymbol\eta_{\mathbf b}
=\mathbf b-\boldsymbol\varepsilon(\mathbf b)+\boldsymbol\delta_{\mathbf b}$.
For every nonnegative integer $k$, define
$\mathcal P_k=\{\mathbf y\in B_k:\|\mathbf y\|_1\equiv k\pmod2\}$.
The following bounds include $c=0$ and $c=r$.
\begin{lemma}\label{lem:minimal-fibers}
One has $f_{\mathbf b}\ge q$ and
$(r+1)E\ge(r-c+1)\sigma_r-(r-1)\sigma_{r-1}$.
For every minimum fiber, $\boldsymbol\eta_{\mathbf b}$ is a sum
of exactly $r-c$ signed units and satisfies
\begin{equation}\label{eq:minimum-scalar-representative}
 \boldsymbol\eta_{\mathbf b}\in B_{r-c},
 \qquad \pi(\boldsymbol\eta_{\mathbf b})=(q+1)\pi(\mathbf b).
\end{equation}
Consequently,
\begin{equation}\label{eq:minimum-parity-capacity}
 E\le |A[q+1]|\,|(q+1)A\cap\pi(\mathcal P_{r-c})|.
\end{equation}
For $c=r$, one has $E\le |A[q+1]|-1$.
\end{lemma}
\begin{proof}
Lemma~\ref{lem:exact-fiber} gives
$(r+1)f_{\mathbf b}\ge 2n-s(\mathbf b)\ge (r+1)q+c-r>(r+1)(q-1)$.
Thus $f_{\mathbf b}-q\ge0$, and
\[
 D_{\mathbf b}=s(\mathbf b)-c+(r+1)(f_{\mathbf b}-q),
 \qquad
 (r+1)\sum_{\mathbf b}(f_{\mathbf b}-q)
 =c\sigma_r+(r-1)\sigma_{r-1}.
\]
Every nonminimum fiber contributes at least one to the sum of
$f_{\mathbf b}-q$, proving $
 (r+1)E\ge(r-c+1)\sigma_r-(r-1)\sigma_{r-1}.
$
For a minimum fiber, $\mathbf b-\boldsymbol\varepsilon(\mathbf b)$
contains $r-s(\mathbf b)$ signed units and
$\boldsymbol\delta_{\mathbf b}$ contains $s(\mathbf b)-c$ signed
units, for a total of $r-c$.
Lemma~\ref{lem:exact-fiber}  proves the scalar identity.
Multiplication by $q+1$ has fibers of size $|A[q+1]|$, giving
\eqref{eq:minimum-parity-capacity}. If $c=r$, the boundary images
are distinct nonzero elements of $A[q+1]$.
\end{proof}

The volume and shell recurrences give $V_{k-1}/V_k\le k/(2n+1)$
and $\sigma_{k-1}/\sigma_k\le k/(2n)$. Iteration yields
\begin{equation}\label{eq:successive-ratios}
 \frac{V_{r-j}}{V_r}\le\frac{(r)_j}{(2n+1)^j},\qquad
 \frac{\sigma_{r-j}}{\sigma_r}\le\frac{(r)_j}{(2n)^j}
 \qquad(0\le j\le r),
\end{equation}
where $(r)_j=r(r-1)\cdots(r-j+1)$ and $(r)_0=1$.
The next bound applies to any subgroup, independently of the kernel size.
\begin{lemma}\label{lem:odd-tail-occupancy}
If $G\le A$ and $k>0$ is odd with $3k<2r+1$, then
\[
 \frac{|G\cap\pi(\mathcal P_k)|}{|G|}\le\frac{2}{\min\{j\ge1:j\text{ is odd},\ jk\ge2r+1\}}
 \le\frac{2k}{2r+1}.
\]
\end{lemma}
\begin{proof}
The assertion is immediate if $G\cap\pi(\mathcal P_k)$ is empty.
The set $G\cap\pi(\mathcal P_k)$ is symmetric and does not contain zero, since $k<r$ and $\pi$ is
injective on $B_r$. The Cayley graph on $G$ with step set $G\cap\pi(\mathcal P_k)$
contains an odd cycle: any nonzero step has odd order, and its
successive multiples form such a cycle.
Let $\ell$ be the shortest odd-cycle length. Lifting its edge labels
to $\mathcal P_k$ gives a sum in $L$ of odd coordinate parity and norm at
most $\ell k$. Thus the sum is nonzero, $\ell k\ge2r+1$,
and $\ell\ge5$.
Fix a shortest odd cycle $C$. It has no chord, since a chord would
split it into an even cycle and a shorter odd cycle.
Every vertex on $C$ therefore has exactly two neighbors on $C$.
A vertex outside $C$ cannot have two consecutive neighbors on $C$,
because they would form a triangle. If it had at least three
neighbors, the gaps between successive neighbors would all be at
least two, and one gap would be odd and at most $\ell-4$.
That gap together with the two edges to the outside vertex would
give an odd cycle of length at most $\ell-2$, a contradiction.
Consequently every vertex has at most two neighbors on $C$.
Counting incidences with $C$ gives
$
 \ell|G\cap\pi(\mathcal P_k)|\le2|G|,
$
and $ \frac{|G\cap\pi(\mathcal P_k)|}{|G|}\le\frac{2}{\ell}\le\frac{2}{\min\{j\ge1:j\text{ is odd},\ jk\ge2r+1\}}
 \le\frac{2k}{2r+1}.
$
\end{proof}

\subsection{Sign-restricted shadows and integral defect bounds}
Differences of points in the same fiber belong to $L$.
After removing heavy coordinate contributions, we compare the resulting
vectors through their sign-compatible cuts. For $\mathbf w\in\Z^n$, define
\[
 \mathscr K_j(\mathbf w)
 =\{\mathbf z\in\Z^n:\|\mathbf z\|_1=j,\quad
 z_iw_i\ge0,\ |z_i|\le|w_i|\text{ for every }i\}.
\]
This rank-$j$ shadow has cardinality
$|\mathscr K_j(\mathbf w)|=[z^j]\prod_i(1+z+\cdots+z^{|w_i|})$.
For a finite sequence of nonnegative integers $w_i$, write
$M=\sum_iw_i$ and
$\prod_i(1+z+\cdots+z^{w_i})=\sum_{j=0}^M A_jz^j$.
\begin{lemma}\label{lem:chain-binomial-monotonicity}
For every such sequence, $jA_j\le(M-j+1)A_{j-1}$ for
$1\le j\le\lfloor M/2\rfloor$. Consequently,
\[
 A_j\le\frac{\binom Mj}{\binom Mk}A_k
 \qquad(0\le k\le j\le\lfloor M/2\rfloor).
\]
\end{lemma}
\begin{proof}
We induct on the factors, retaining symmetry and unimodality.
Extend coefficient sequences by zero. Suppose $B(z)=\sum_iB_i z^i$
has degree $m$ and the three asserted properties, and put
$D_i=(m-i+1)B_{i-1}-iB_i$.
Then $D_{m-i+1}=-D_i$ and $D_i\ge0$ for $i\le(m+1)/2$;
when $m$ is odd the middle term is zero.
For $A(z)=B(z)(1+z+\cdots+z^w)$ and $M=m+w$,
\[
 A_j=\sum_{a=0}^wB_{j-a},\qquad
 A_j-A_{j-1}=B_j-B_{j-w-1}\ge0\quad(j\le M/2).
\]
The last inequality holds because the two indices have midpoint
$j-(w+1)/2\le(m-1)/2$; the larger is at least as close to $m/2$.
Symmetry of $A$ follows by reversing its coefficients.
Direct expansion gives
\begin{equation*}
 (M-j+1)A_{j-1}-jA_j
 =\sum_{a=0}^wD_{j-a}
    +\sum_{a=1}^w(w+1-2a)B_{j-a}.
\end{equation*}
If an index $i\in[j-w,j]$ exceeds $(m+1)/2$, its reflection
$i'=m-i+1$ also lies in that interval, since
$
 i'\ge m-j+1\ge j-w,
$
$
i'\le(m+1)/2<i\le j.
$
The terms $D_i,D_{i'}$ cancel; all unpaired terms are nonnegative.
Pairing $a$ with $w+1-a$ in the second sum gives, for
$a<(w+1)/2$,
$
 (w+1-2a)(B_{j-a}-B_{j-w-1+a})\ge0,
$
by the same midpoint comparison; an unpaired middle coefficient is
zero. 
This proves the adjacent-rank inequality and preserves all three
induction hypotheses. The initial polynomial $1$ satisfies these
hypotheses. Iterating the adjacent-rank inequality gives
$A_j\le\binom Mj A_k/\binom Mk$.
\end{proof}

In the following moment sums, $\mathbf b$ ranges over $S_r$. Put
\[
 P_{\mathbf b}=\sum_{\mathbf x\in F_{\mathbf b}}\#\{i:|x_i|\ge2\},\quad
 H_{\mathbf b}=\sum_{\mathbf x\in F_{\mathbf b}}\sum_i(|x_i|-2)_+,
\]
\[
 \MomentU=\sum_{\mathbf b}P_{\mathbf b}^2,\quad
 \MomentT=\sum_{\mathbf b}\bigl(r-s(\mathbf b)\bigr)P_{\mathbf b}.
\]
Lemma~\ref{lem:exact-fiber} gives
$D_{\mathbf b}=P_{\mathbf b}+H_{\mathbf b}$ and
$D_{\mathbf b}+r-s(\mathbf b)=r-c+(r+1)(f_{\mathbf b}-q)$.
Subtracting $\operatorname{sgn}(x_i)\mathbf e_i$ from a pair
$(\mathbf x,i)$ with $\mathbf x\in S_{r+1}$ and $|x_i|\ge3$ maps it
bijectively to $(\mathbf y,i)$ with $\mathbf y\in S_r$ and
$|y_i|\ge2$, and preserves the weight
$|x_i|-2=|y_i|-1$. Equations~\eqref{eq:support-first-moment} and
\eqref{eq:boundary-total-defect} therefore give
\begin{align*}
 &\sum_{\mathbf b}H_{\mathbf b}
 =\sum_{\mathbf b}\bigl(r-s(\mathbf b)\bigr)=(r-1)\sigma_{r-1},\\
 &\sum_{\mathbf b}P_{\mathbf b}=r\sigma_r-(r-1)\sigma_{r-1},\\
 &\sum_{\mathbf b}\bigl(r-s(\mathbf b)\bigr)^2
=(r-1)\sigma_{r-1}+(r-1)(r-2)\sigma_{r-2}.
\end{align*}
For $k\ge0$ define
\begin{align*}
 \Phi_k=\sum_{\mathbf z\in S_k}
             \bigl((2n-s(\mathbf z))^2-2n+2s(\mathbf z)\bigr), \quad
 \Psi_k=\sum_{\mathbf z\in S_k}
             \bigl((2n-s(\mathbf z))^2-4n+3s(\mathbf z)\bigr).
\end{align*}
Expanding in $k-s(\mathbf z)$ and using~\eqref{eq:support-first-moment}, one has
\[
 \sum_{\mathbf z\in S_k}(k-s(\mathbf z))^2
 =(k-1)\sigma_{k-1}+(k-1)(k-2)\sigma_{k-2},
\]
which yields
\begin{align*}
 \Phi_k={}&\bigl((2n-k)^2-2n+2k\bigr)\sigma_k
 +(4n-2k-1)(k-1)\sigma_{k-1}
 +(k-1)(k-2)\sigma_{k-2},\\
 \Psi_k={}&\bigl((2n-k)^2-4n+3k\bigr)\sigma_k
 +(4n-2k-2)(k-1)\sigma_{k-1}
 +(k-1)(k-2)\sigma_{k-2}.
\end{align*}
The rank ratio bound now controls pairs of heavy coordinate incidences,
including pairs involving the boundary representative.
\begin{lemma}
\label{lem:sign-restricted-heavy-pairs}
For $r\ge4$,
\begin{align*}
 \MomentU&\le r\sigma_r-(r-1)\sigma_{r-1}+\Psi_{r-3}
                         +\frac{r}{r-1}\Phi_{r-2},\\
 \MomentT&\le\Phi_0+\sum_{j=2}^{r-1}\frac{r}{r-j}\Phi_{r-j-1}.
\end{align*}
\end{lemma}
\begin{proof}
First take distinct $\mathbf x,\mathbf y\in F_{\mathbf b}$ and heavy
signed directions $\mathbf a\in\ssupp(\mathbf x)$,
$\mathbf a'\in\ssupp(\mathbf y)$. Their disjoint signed supports give
$\mathbf a\ne\mathbf a'$, and
\[
 \boldsymbol\ell=\mathbf x-\mathbf y\in L,\quad
 \|\boldsymbol\ell\|_1=2(r+1),\quad
 \mathbf w=\boldsymbol\ell-2\mathbf a+2\mathbf a',\quad
 \|\mathbf w\|_1=2r-2.
\]
The vector $\mathbf u=\mathbf x-2\mathbf a$ belongs to
$\mathscr K_{r-1}(\mathbf w)$. Indeed,
$\mathbf w=\mathbf u-(\mathbf y-2\mathbf a')$, and the two
summands have compatible signs in every coordinate. For fixed
$\boldsymbol\ell,\mathbf a,\mathbf a'$, the cut determines the pair
uniquely, since $\mathbf x=\mathbf u+2\mathbf a$ and
$\mathbf y=\mathbf u+2\mathbf a-\boldsymbol\ell$.
These assertions also apply when $\mathbf a'=-\mathbf a$.
For fixed $\mathbf a,\mathbf a'$, let $\boldsymbol\ell$ range over
the lattice vectors arising from such pairs, and write
$\mathbf w_{\boldsymbol\ell}=\boldsymbol\ell-2\mathbf a+2\mathbf a'$.
If $\mathbf z$ belonged to both
$\mathscr K_{r-2}(\mathbf w_{\boldsymbol\ell})$ and
$\mathscr K_{r-2}(\mathbf w_{\boldsymbol\ell'})$, then
$
 \|\boldsymbol\ell-\boldsymbol\ell'\|_1
 \le\|\mathbf w_{\boldsymbol\ell}-\mathbf z\|_1
    +\|\mathbf w_{\boldsymbol\ell'}-\mathbf z\|_1=2r.
$
The minimum norm in $L$ forces $\boldsymbol\ell=\boldsymbol\ell'$.
Lemma~\ref{lem:chain-binomial-monotonicity} therefore bounds the
number of ordered pairs for these fixed directions by
\[
 \sum_{\boldsymbol\ell}|\mathscr K_{r-1}(\mathbf w_{\boldsymbol\ell})|
 \le\frac{r}{r-1}\sum_{\boldsymbol\ell}
                    |\mathscr K_{r-2}(\mathbf w_{\boldsymbol\ell})|
 =\frac{r}{r-1}\left|\bigcup_{\boldsymbol\ell}
                    \mathscr K_{r-2}(\mathbf w_{\boldsymbol\ell})\right|.
\]
A shadow $\mathbf z$ excludes
$-\ssupp(\mathbf z)$ as choices for $\mathbf a$ and
$\ssupp(\mathbf z)$ as choices for $\mathbf a'$.
There are $2n-s(\mathbf z)$ choices of each; equality $\mathbf a=\mathbf a'$
is allowed by both restrictions on exactly $2(n-s(\mathbf z))$ zero-axis
directions. Thus the number of admissible ordered choices is
$(2n-s(\mathbf z))^2-2(n-s(\mathbf z))=(2n-s(\mathbf z))^2-2n+2s(\mathbf z)$.
This also counts opposite directions $\mathbf a'=-\mathbf a$.
Summing bounds the distinct-point contribution by
$r\Phi_{r-2}/(r-1)$.

For the same exterior point twice, diagonal heavy incidences total
$\sum P_{\mathbf b}=r\sigma_r-(r-1)\sigma_{r-1}$.
For two distinct heavy axes, remove $2\mathbf a+2\mathbf a'$ to
obtain a residual $\mathbf z\in S_{r-3}$. Each direction has
$2n-s(\mathbf z)$ choices; pairs on the same axis number
$4(n-s(\mathbf z))+s(\mathbf z)$.
Consequently this contribution is exactly
\[
 \sum_{\mathbf z\in S_{r-3}}
       \bigl((2n-s(\mathbf z))^2-4n+3s(\mathbf z)\bigr)=\Psi_{r-3}.
\]
Adding the three contributions proves the first inequality.

An incidence in $\MomentT$ specifies $\mathbf b\in S_r$,
$\mathbf x\in F_{\mathbf b}$, a direction $\mathbf a$ of
$\mathbf b$, a heavy direction $\mathbf a'$ of $\mathbf x$, and
$2\le j\le b_{\mathbf a}$. Fix $\mathbf a,\mathbf a',j$.
The signed supports are disjoint, and
\[
 \boldsymbol\ell=\mathbf x-\mathbf b\in L,\quad
 \|\boldsymbol\ell\|_1=2r+1,\quad
 \mathbf w=\boldsymbol\ell-2\mathbf a'+j\mathbf a,
 \quad\|\mathbf w\|_1=2r-j-1.
\]
The cut $\mathbf u=-(\mathbf b-j\mathbf a)$ belongs to
$\mathscr K_{r-j}(\mathbf w)$, since
$\mathbf w=(\mathbf x-2\mathbf a')+\mathbf u$ is a sum of
sign-compatible vectors. For fixed $\boldsymbol\ell,\mathbf a,
\mathbf a',j$, the inverse formulas
$\mathbf b=j\mathbf a-\mathbf u$ and
$\mathbf x=\boldsymbol\ell+\mathbf b$ recover the incidence uniquely. If $j<r$, rank-$(r-j-1)$
shadows for different lattice vectors are disjoint, since a common
shadow would give distance at most
$2\{(2r-j-1)-(r-j-1)\}=2r$.
Their adjacent-rank factor is
$(2r-j-1-(r-j)+1)/(r-j)=r/(r-j)$.
The residual signs exclude $\ssupp(\mathbf z)$ from the choices
of $\mathbf a$ and $-\ssupp(\mathbf z)$ from those of $\mathbf a'$;
the number of allowed unequal ordered directions is again the
summand defining $\Phi_{r-j-1}$.
If $j=r$, the residual has norm $r-1$, so for each ordered pair of
directions there is at most one lattice vector. There are
$2n(2n-1)=\Phi_0$ such pairs. Summing over $j$ gives the second bound.
\end{proof}

For the lower bounds, we use the gaps between consecutive possible values
of $P_{\mathbf b}+r-s(\mathbf b)+H_{\mathbf b}=r-c+(r+1)(f_{\mathbf b}-q)$.
\begin{lemma}\label{lem:two-shell-defect}
For every $0\le c\le r$, one has
\begin{equation*}
 \MomentU+2\MomentT+(r+1)(r-1)\sigma_{r-1}+(r-3)(r-2)\sigma_{r-2}
 \ge\bigl(r(r+1)+(r-c)(c-1)\bigr)\sigma_r.
\end{equation*}
If $c=r$ and $f_{\mathbf b}\ne q+1$ for every $\mathbf b\in S_r$,
then also
\[
 \MomentU+2\MomentT+2(r+1)(r-1)\sigma_{r-1}
 +(r-1)(r-2)\sigma_{r-2}\ge2r(r+1)\sigma_r.
\]
\end{lemma}
\begin{proof}
For nonnegative integers $z,H,a$ and an integer $G\ge3$, suppose
$z+H\notin(a,a+G)$. We claim
\begin{equation}\label{eq:two-shell-gap-pointwise}
 (z-a)(z-a-G)\ge-(G-1)H+2(H-1)_+.
\end{equation}
The right side is nonpositive, so only $a<z<a+G$ requires proof.
In this case $v=a+G-z$ is a positive integer and $H\ge v$. Thus
\[
 (G-1)H-2(H-1)_+
 =(G-3)H+2\ge(G-3)v+2\ge v(G-v).
\]
Since $(z-a)(z-a-G)=-v(G-v)$, this proves the claim.
The boundary identities give
$P_{\mathbf b}+r-s(\mathbf b)+H_{\mathbf b}=r-c+(r+1)(f_{\mathbf b}-q)$,
where $f_{\mathbf b}\ge q$, and
\begin{align*}
 &\sum_{\mathbf b}\bigl(P_{\mathbf b}+r-s(\mathbf b)\bigr)=r\sigma_r,\\
 &\sum_{\mathbf b}\bigl(P_{\mathbf b}+r-s(\mathbf b)\bigr)^2
 =\MomentU+2\MomentT+(r-1)\sigma_{r-1}+(r-1)(r-2)\sigma_{r-2},\\
& \sum_{\mathbf b}H_{\mathbf b}=(r-1)\sigma_{r-1}.
\end{align*}
For nonnegative integers $a_i$,
$(\sum_i a_i-1)_+\ge\sum_i(a_i-1)_+$: if some $a_i\ge2$, the
right side is the sum of those $a_i$, less their number, and hence
at most $\sum_i a_i-1$; otherwise it is zero. Applying this to
$H_{\mathbf b}=\sum_{\mathbf x\in F_{\mathbf b}}\sum_i(|x_i|-2)_+$
and then summing the fibers yields
\begin{align*}
 \sum_{\mathbf b}(H_{\mathbf b}-1)_+
 &\ge\sum_{\mathbf x\in S_{r+1}}\sum_i(|x_i|-3)_+
 =\sum_{\mathbf y\in S_{r-1}}\sum_i(|y_i|-1)_+
 =(r-2)\sigma_{r-2}.
\end{align*}
For the middle equality, pairs $(\mathbf x,i)$ with $|x_i|\ge4$
correspond bijectively to $(\mathbf y,i)$ with $|y_i|\ge2$ by
$\mathbf y=\mathbf x-2\operatorname{sgn}(x_i)\mathbf e_i$,
preserving the displayed weight.

Taking $z=P_{\mathbf b}+r-s(\mathbf b)$ and $(a,G)=(r-c,r+1)$
in~\eqref{eq:two-shell-gap-pointwise} and summing gives
\begin{align*}
 \sum_{\mathbf b}\bigl(P_{\mathbf b}+r-s(\mathbf b)\bigr)^2
 &\ge\bigl((3r-2c+1)r-(r-c)(2r-c+1)\bigr)\sigma_r-r(r-1)\sigma_{r-1}
                                      +2(r-2)\sigma_{r-2}\\
 &=\bigl(r(r+1)+(r-c)(c-1)\bigr)\sigma_r-r(r-1)\sigma_{r-1}
                                      +2(r-2)\sigma_{r-2}.
\end{align*}
Substituting the displayed square-sum identity and collecting the two shell
coefficients gives the first assertion.
If $c=r$ and $f_{\mathbf b}\ne q+1$ for every $\mathbf b$, then
$P_{\mathbf b}+r-s(\mathbf b)+H_{\mathbf b}
\in\{0,2(r+1),3(r+1),\ldots\}$.
The same pointwise inequality with $(a,G)=(0,2(r+1))$, after dropping
its nonnegative last term and summing, gives
\[
 \sum_{\mathbf b}\bigl(P_{\mathbf b}+r-s(\mathbf b)\bigr)^2
 \ge2r(r+1)\sigma_r-(2r+1)(r-1)\sigma_{r-1}.
\]
Substituting the same square-sum identity and collecting the
$\sigma_{r-1}$ coefficient gives the second assertion.
\end{proof}

\subsection{Scalar-kernel collisions and cancellation}
Scalar collisions lead to intersections of a projected ball with its
nonzero shifts. Lifting these intersections allows us to apply the
mixed-radius bounds. For $Y\subseteq A$ and $a\in A$, write
$C_Y(a)=|Y\cap(a+Y)|$ for the autocorrelation of $Y$.
\begin{lemma}\label{lem:nonzero-autocorrelation}
For $0\le k\le r-1$ and $0\ne a\in A$,
\begin{equation}\label{eq:nonzero-autocorrelation}
 C_{\pi(B_k)}(a)\le2V_{k-1}.
\end{equation}
\end{lemma}
\begin{proof}
For $k=0$ both sides vanish. Suppose $k\ge1$.
Injectivity on $B_r$ gives the exact lift sum
\[
 C_{\pi(B_k)}(a)
 =\sum_{\pi(\mathbf w)=a}C_{k,k}(\mathbf w).
\]
If $a\in\pi(B_2)$, let $\mathbf w_0\in B_2$ be its unique lift.
Every other contributing lift would satisfy
$\|\mathbf w-\mathbf w_0\|_1\le2k+2\le2r$, which is impossible.
For a unit displacement, or a signed displacement on two distinct axes,
direct section counting gives
\begin{equation}\label{eq:short-shift-autocorrelation}
 C_{k,k}(\mathbf e_i)
 =C_{k,k}(\mathbf e_i\pm\mathbf e_j)
 =V_{k-1}(n)+V_{k-1}(n-1)\le2V_{k-1}(n)
 \qquad(i\ne j).
\end{equation}
Indeed, a one-dimensional section of residual radius $a\ge1$
has $2a$ intersection points; subtracting the $2a-1$ points of
the corresponding $B_{k-1}$ section leaves one point for each
member of $B_{n-1}(k-1)$. In two coordinates, the intersections
for shifts $(1,0)$ and $(1,1)$ both have $2a^2$ points at residual
radius $a$; summing over the remaining coordinates gives the equality.
For displacement $2\mathbf e_i$, the sections have
$\max(0,2a-1)$ points, hence $C_{k,k}(2\mathbf e_i)=V_{k-1}$.
Sign changes give all nonzero cases in $B_2$.

If $a\notin\pi(B_2)$, every lift has norm at least three.
Lemma~\ref{thm:lens-three-halves} gives
\[
 \sum_{\pi(\mathbf w)=a}C_{k,k}(\mathbf w)
 \le\frac{3}{2}\sum_{\pi(\mathbf w)=a}C_{k-1,k+1}(\mathbf w)
 \le\frac{3}{2} V_{k-1}.
\]
For the last bound, each $\mathbf x\in B_{k-1}$ permits at most
one $\mathbf y\in B_{k+1}\subseteq B_r$ with
$\pi(\mathbf x-\mathbf y)=a$. This proves~\eqref{eq:nonzero-autocorrelation}.
\end{proof}

Assume $0\le c\le r-1$, so $1\le r-c\le r$.
For each scalar representative $\boldsymbol\eta\in B_{r-c}$, let
$E_{\boldsymbol\eta}=\{\pi(\mathbf b):f_{\mathbf b}=q,
\ \boldsymbol\eta_{\mathbf b}=\boldsymbol\eta\}$ and put
$m_{\boldsymbol\eta}=|E_{\boldsymbol\eta}|$.
By~\eqref{eq:minimum-scalar-representative} and injectivity on $B_r$,
two minimum images have the same representative exactly when they
belong to the same fiber of multiplication by $q+1$.
Each scalar fiber has size $|A[q+1]|$.
\begin{lemma}\label{lem:full-minimum-lifting}
With the notation above, one has
\begin{equation}\label{eq:full-weighted-energy}
 \sum_{\boldsymbol\eta}\|\boldsymbol\eta\|_1
       (m_{\boldsymbol\eta}^2-m_{\boldsymbol\eta})
 \le4n(|A[q+1]|-1)\sum_{j=1}^{r-c}V_{r-j-1}.
\end{equation}
\end{lemma}
\begin{proof}
Adjoin $\mathbf b$ to $F_{\mathbf b}$. Its signed supports partition
$S_1$, and
$
 \boldsymbol\eta_{\mathbf b}
 =\sum_{\mathbf x\in\{\mathbf b\}\cup F_{\mathbf b}}
                (\mathbf x-\boldsymbol\varepsilon(\mathbf x)).
$
Suppose a signed direction $\mathbf u$ has net coefficient $v>0$
in $\boldsymbol\eta_{\mathbf b}$. Its positive excess comes from
one member $\mathbf x$ of the extended fiber. Subtracting the
opposing nonnegative excess leaves $v$, so $x_{\mathbf u}\ge v+1$.
For every $1\le j\le v$, put
$\mathbf z=\mathbf x-(j+1)\mathbf u$.
If $\mathbf x$ is exterior, $\mathbf z\in S_{r-j}$; if
$\mathbf x=\mathbf b$, then $\mathbf z\in S_{r-j-1}$.
In both cases $\mathbf z\in B_{r-j}$ and
$\pi(\mathbf z)=\pi(\mathbf b)-(j+1)\pi(\mathbf u)$.
Two minimum images differing by $a\in A[q+1]$ have the same
$\boldsymbol\eta$ by uniqueness in $B_r$. Apply this construction
in both fibers for each net direction and each $j$.
There are $\|\boldsymbol\eta\|_1$ choices, and the resulting
images of $\mathbf z,\mathbf z'$ differ by $a$.
For fixed $j,\mathbf u,\mathbf z,\mathbf z'$, the two original
images are uniquely recovered by adding $(j+1)\pi(\mathbf u)$.
There are $2n$ directions, so for each $0\ne a\in A[q+1]$,
\begin{align*}
 \sum_{\boldsymbol\eta}\|\boldsymbol\eta\|_1
       C_{E_{\boldsymbol\eta}}(a)
 &\le2n\sum_{j=1}^{r-c}
       C_{\pi(B_{r-j})}(a)
\le4n\sum_{j=1}^{r-c} V_{r-j-1}.
\end{align*}
The last bound is~\eqref{eq:nonzero-autocorrelation}; its $j=r$
term vanishes because the nonzero-shift autocorrelation of $B_0$ is zero.
Summing over $a\in A[q+1]\setminus\{0\}$ counts every ordered
distinct pair in each $E_{\boldsymbol\eta}$ once and proves
\eqref{eq:full-weighted-energy}.
\end{proof}

Define the cancellation loss by
$\mathcal C=\sum_{\mathbf b:f_{\mathbf b}=q}
(r-c-\|\boldsymbol\eta_{\mathbf b}\|_1)$.
The weighted estimate~\eqref{eq:full-weighted-energy} gives the following constraint when the
scalar kernel is nontrivial.
\begin{lemma}\label{lem:all-minimum-kernel}
If $A[q+1]\ne\{0\}$, then
$
 (r-c)\max\{2E-V_r,E-V_r/3\}
 \le\mathcal C+4n\sum_{j=1}^{r-c}V_{r-j-1}.
$
\end{lemma}
\begin{proof}
Since $V_r$ is odd and the kernel is nontrivial, $|A[q+1]|$ is odd and $|A[q+1]|\ge3$.
For every integer $0\le m\le |A[q+1]|$,
\begin{align*}
 &m(m-1)-(|A[q+1]|-1)(2m-|A[q+1]|)=(|A[q+1]|-m)(|A[q+1]|-m-1)\ge0,\\
 &m(m-1)-(|A[q+1]|-1)(m-|A[q+1]|/3)\ge\frac{(|A[q+1]|-1)(|A[q+1]|-3)}{12}\ge0.
\end{align*}
The second estimate uses that $|A[q+1]|$ is odd. Apply these inequalities to the number of
minimum images in each of the $V_r/|A[q+1]|$ scalar fibers, including zero
counts. Since their sum is $E$,
\[
 \sum_{\boldsymbol\eta}m_{\boldsymbol\eta}
           (m_{\boldsymbol\eta}-1)
 \ge(|A[q+1]|-1)\max\{2E-V_r,E-V_r/3\}.
\]
Moreover $m_{\boldsymbol\eta}-1\le |A[q+1]|-1$ whenever
$m_{\boldsymbol\eta}>0$, and
$\mathcal C=\sum_{\boldsymbol\eta}
(r-c-\|\boldsymbol\eta\|_1)m_{\boldsymbol\eta}$. Therefore
\begin{align*}
 \sum_{\boldsymbol\eta}\|\boldsymbol\eta\|_1
           m_{\boldsymbol\eta}(m_{\boldsymbol\eta}-1)
 &\ge (r-c)\sum_{\boldsymbol\eta}
            m_{\boldsymbol\eta}(m_{\boldsymbol\eta}-1)
                   -(|A[q+1]|-1)\mathcal C\\
 &\ge(|A[q+1]|-1)\bigl((r-c)\max\{2E-V_r,E-V_r/3\}-\mathcal C\bigr).
\end{align*}
Comparison with Lemma~\ref{lem:full-minimum-lifting}, followed by division
by $|A[q+1]|-1$, proves the assertion.
\end{proof}

To bound the cancellation term, we count successive levels of opposing
coordinate excesses. For $k\ge0$ and $0\le m\le2r$, put
\[
 W_k=(2n-k)\sigma_k+(k-1)\sigma_{k-1},\qquad
 \Omega_m=\frac{\binom{m}{\lfloor m/2\rfloor}}{\binom{m}{(m-r)_+}}W_{(m-r)_+}.
\]
\begin{lemma}\label{lem:exact-layered-cancellation}
One has
\begin{equation}\label{eq:layered-cancellation}
\mathcal C\le
\sum_{j=1}^{\lfloor (r-c)/2\rfloor}
\bigl(\Omega_{2r-2j}+2\Omega_{2r-2j-1}\bigr).
\end{equation}
If $n\ge r^2/3$, then
\begin{equation}\label{eq:layered-cancellation-ratio}
\frac{\mathcal C}{V_r}\le\frac{r^2}{2n}
\left(1+\frac{2(r-1)}{2n}+
\frac{2(r-1)^2}{4n^2-(r-1)^2}\right)
<\frac{2r^2}{2n}.
\end{equation}
\end{lemma}

\begin{proof}
Opposing coordinate magnitudes $u,v$ contribute
$2\min(u-1,v-1)
=2\sum_{j\ge1}\mathbf1_{u\ge j+1,\ v\ge j+1}$.
Since the total excess is $r-c$, only $2j\le r-c$ can occur.

Fix such a level $j$ and opposing origins $\mathbf x,\mathbf y$.
Orient $\mathbf a$ toward $\mathbf x$, taking $\mathbf x=\mathbf b$
for a mixed pair, and put
$\mathbf w=\mathbf x-\mathbf y-2(j+1)\mathbf a$.
The disjoint signed supports give
$m:=\|\mathbf w\|_1=2r-2j$ for an exterior pair and
$m=2r-2j-1$ for a mixed pair.
In either case, $\mathbf x-(j+1)\mathbf a$ is a
rank-$\lfloor m/2\rfloor$ cut of $\mathbf w$.
For fixed $\mathbf a,j,\mathbf w$, the cut recovers
$\mathbf x$, then $\mathbf y$, and hence their fiber.

For fixed $\mathbf a,j$ and either type of pair, distinct
residuals have disjoint rank-$(m-r)_+$ shadows:
a common shadow would give a nonzero lattice difference
of norm at most $2(m-(m-r)_+)\le2r$.
Every shadow omits $-\mathbf a$ from its signed support, and
\eqref{eq:support-first-moment} gives
\[
\sum_{\mathbf z\in S_k}(2n-s(\mathbf z))
=(2n-k)\sigma_k+(k-1)\sigma_{k-1}=W_k.
\]
Since $(m-r)_+\le\lfloor m/2\rfloor$,
Lemma~\ref{lem:chain-binomial-monotonicity} bounds the
number of oriented pairs by $\Omega_m$.
Exterior pairs have two orientations, accounting for their
loss factor two; mixed pairs have one orientation and retain
that factor. Summing proves~\eqref{eq:layered-cancellation}.

Now assume $2n\ge2r^2/3$ and put
$\mathcal E_j=(r)_j^2/(2n)^{2j-1}$.
By~\eqref{eq:successive-ratios},
$\sigma_k/V_r\le r!/(k!(2n)^{r-k})$ for $0\le k\le r$.
Using $W_k\le2n\sigma_k$ and cancelling factorials gives
\[
\begin{aligned}
&\frac{\Omega_{2r-2j}}{V_r}
\le\frac{(r!)^2}{((r-j)!)^2(2n)^{2j-1}}=\mathcal E_j,\\
&\frac{2\Omega_{2r-2j-1}}{V_r}
\le\frac{2(r!)^2}{(r-j)!(r-j-1)!(2n)^{2j}}
=\frac{2(r-j)}{2n}\mathcal E_j,
\end{aligned}
\]
where the second inequality assumes $2j\le r-1$.
The only remaining case is $c=0$ and $r=2j$, when
\[
\frac{2\Omega_{r-1}}{V_r}
=\frac{2\binom{r-1}{j-1}W_0}{V_r}
\le\frac{2\binom{r-1}{j-1}r!}{(2n)^{r-1}}
=\mathcal E_j.
\]
Because $r\ge3$, this exception has $j\ge2$.
For every ordinary level $j\ge2$,
$2(r-j)/(2n)\le3(r-2)/r^2<1$. Hence
\[
\frac{\mathcal C}{V_r}
\le \mathcal E_1\left(1+\frac{2(r-1)}{2n}\right)
+2\sum_{j=2}^{\lfloor (r-c)/2\rfloor}\mathcal E_j,
\qquad \mathcal E_1=\frac{r^2}{2n}.
\]
For consecutive terms,
\[
\frac{\mathcal E_{j+1}}{\mathcal E_j}
=\frac{(r-j)^2}{4n^2}
\le\frac{(r-1)^2}{4n^2}\le\frac{1}{9}.
\]
Summing the geometric series proves the first inequality
in~\eqref{eq:layered-cancellation-ratio}.
Finally, its parenthetical factor is at most
$1+2/3+1/4=23/12<2$.
\end{proof}

\subsection{Character coefficient bounds}
The ball generating function converts the tiling identity into a coefficient recurrence. A nonnegative matrix bounds this recurrence and supplies the character estimates needed for the repeated-coordinate count.
For $r\ge2$ and $M>0$, let $\mathsf H_r(M)$ be the $r\times r$
matrix indexed by $0,\ldots,r-1$, with
$(\mathsf H_r(M))_{i,i+1}=i+1$ for $0\le i\le r-2$,
$(\mathsf H_r(M))_{i,j}=M$ for $0\le j<i\le r-1$,
and all other entries zero.
Given $B\ge1$ and $n_0\ge1$, define $h_0=1$ and
\[
 kh_k=Bh_{k-1}-(2n_0+B)\sum_{j=0}^{k-2}h_j
 \qquad(1\le k\le r).
\]

\begin{lemma}\label{lem:positive-hessenberg}
If $r\ge2$, $B\ge1$, $n_0\ge1$, and $h_1,\ldots,h_r>0$,
then every lattice tiling of radius $r$ in dimension $n\ge n_0$
satisfies $\Lambda<B\sqrt{n/n_0}$.
\end{lemma}
\begin{proof}
Fix a nonprincipal character $\chi$.
The ball generating function and its logarithmic derivative give
\[
 \sum_{k\ge0}\chi(B_n(k))z^k
 =\exp\left(\sum_{j\ge1}\frac{c_jz^j}{j}\right),\qquad
 c_j=\lambda_{\chi^j}-2n\mathbf1_{2\mid j},
\]
and hence $k\chi(B_n(k))=\sum_{j=1}^kc_j\chi(B_n(k-j))$.
The tiling identity gives $\chi(B_n(r))=0$.
If $\Lambda<1$, then $\Lambda<1\le B\sqrt{n/n_0}$. Otherwise, for even $j$, a principal
$\chi^j$ gives $c_j=1$, and a nonprincipal $\chi^j$ gives
$|c_j|\le2n+\Lambda$. For odd $j$, the corresponding alternatives are
$|c_j|=2n+1\le2n+\Lambda$ and $|c_j|\le\Lambda$. Thus
$|c_j|\le2n+\Lambda$ in every case. Set $v_i=\chi(B_n(i))$ for $0\le i<r$, and set $v_r=0$.
Then $v_0=1$, and the coefficient recurrence gives
$
 \lambda_\chi v_i=(i+1)v_{i+1}
       -\sum_{j=0}^{i-1}c_{i-j+1}v_j
       $ for
       $
   0\le i<r.
$
Consequently $\lambda_\chi$ is an eigenvalue of the matrix with
superdiagonal $i+1$ and entry $-c_{i-j+1}$ in position $(i,j)$
for $j<i$. Taking absolute values entrywise and then the maximum
over nonprincipal characters gives
$\Lambda\le\rho\bigl(\mathsf H_r(2n+\Lambda)\bigr)$.
Here $\rho$ denotes the Perron root of a nonnegative irreducible
matrix. To justify the comparison directly, an eigenvector $\mathbf v\ne\mathbf{0}$
satisfies $|\lambda_\chi||\mathbf v|\le \mathsf H_r(2n+\Lambda)|\mathbf v|$; multiplication by a
strictly positive left Perron vector gives $|\lambda_\chi|\le\rho$.

Set $M_0=2n_0+B$ and $\mathbf w=(h_0,\ldots,h_{r-1})^{\mathsf T}>\mathbf{0}$.
The defining recurrence gives
\[
 (\mathsf H_r(M_0)\mathbf w)_i=Bh_i\quad(0\le i<r-1),\qquad
 (\mathsf H_r(M_0)\mathbf w)_{r-1}=Bh_{r-1}-rh_r<Bh_{r-1}.
\]
Multiplying by a positive left Perron vector proves
$\rho(\mathsf H_r(M_0))<B$.
Write $a=\sqrt{n/n_0}\ge1$ and
$\mathsf D=\operatorname{diag}(1,a,\ldots,a^{r-1})$. 
Conjugating by $\mathsf D$ and applying the bound $2n+aB\le a^2(2n_0+B)=a^2M_0$ yields the entrywise estimates
$
 (\mathsf D^{-1}\mathsf H_r(2n+aB)\mathsf D)_{i,i+1}=a(i+1)
 $
 and $
 (\mathsf D^{-1}\mathsf H_r(2n+aB)\mathsf D)_{i,i-k}
 =(2n+aB)a^{-k}\le a^{2-k}M_0\le aM_0$
 for $k\ge1$.
Consequently
$\rho(\mathsf H_r(2n+aB))\le a\rho(\mathsf H_r(M_0))<aB$.
If $\Lambda\ge aB$, put $t=(2n+\Lambda)/(2n+aB)\ge1$.
Entrywise,
$
 \mathsf H_r(2n+\Lambda)\le t\mathsf H_r(2n+aB)$
 and 
 $
 \rho(\mathsf H_r(2n+\Lambda))<taB\le\Lambda,
$
since $aB(2n+\Lambda)\le\Lambda(2n+aB)$ is equivalent to
$2naB\le2n\Lambda$. This contradicts $\Lambda\le\rho(\mathsf H_r(2n+\Lambda))$.
\end{proof}

The same matrix comparison gives a spectral bound without endpoint certificates.
\begin{corollary}\label{thm:uspec}
For $r\ge2$, one has
$\Lambda\le3(r-1)+2\sqrt{2(r-1)(n+r-1)}$.
\end{corollary}
\begin{proof}
For $\Lambda<1$ the conclusion is immediate. Otherwise set
$M=2n+\Lambda$ and $\mathsf D_t=\operatorname{diag}(1,t,\ldots,t^{r-1})$.
For $t>1$, the matrix comparison already proved gives
\[
 \Lambda\le\rho(\mathsf H_r(M))
 \le\|\mathsf D_t^{-1}\mathsf H_r(M)\mathsf D_t\|_\infty
 \le (r-1)t+M\sum_{j\ge1}t^{-j}=(r-1)t+\frac M{t-1}.
\]
Taking $t=1+\sqrt{M/(r-1)}$ gives $\Lambda\le r-1+2\sqrt{(r-1)(2n+\Lambda)}$.
If $\Lambda\le r-1$ there is nothing further to prove. 
If $\Lambda>r-1$, the inequality $(\Lambda-(r-1))^2\le4(r-1)(2n+\Lambda)$ algebraically rearranges to
$
 (\Lambda-3(r-1))^2\le8(r-1)(n+r-1).
$
The stated bound follows by taking the positive square root.
\end{proof}

For fixed radii, successive coefficient ratios give a direct way to
verify the positivity required by the matrix criterion.
For $r\ge2$, $n_0\ge1$, $B\ge1$, $D>0$, and $a_1,\ldots,a_r>0$, put
\[
 \delta_k=a_{k-1}\{D(B+k-1)-a_k\}
              -D^2(k-1)(2n_0+2B)
 \qquad(2\le k\le r).
\]

\begin{lemma}\label{lem:spectral-ratio-certificate}
If $a_1=DB$ and $\delta_k\ge0$ for
$2\le k\le r$, then the hypotheses of
Lemma~\ref{lem:positive-hessenberg} hold.
\end{lemma}
\begin{proof}
For the sequence in Lemma~\ref{lem:positive-hessenberg}, put
$\widetilde h_k=k!h_k$. Subtracting consecutive recurrences gives
\[
 kh_k-(k-1)h_{k-1}
 =B(h_{k-1}-h_{k-2})-(2n_0+B)h_{k-2}.
\]
Multiplication by $(k-1)!$ yields
\[
 \widetilde h_0=1,\quad \widetilde h_1=B,\qquad
 \widetilde h_k=(B+k-1)\widetilde h_{k-1}-(k-1)(2n_0+2B)\widetilde h_{k-2}\quad(k\ge2).
\]
Starting from $\widetilde h_1/\widetilde h_0=B=a_1/D>0$, induction gives
\[
\begin{aligned}
 \frac{\widetilde h_k}{\widetilde h_{k-1}}
 &=B+k-1-\frac{(k-1)(2n_0+2B)}{\widetilde h_{k-1}/\widetilde h_{k-2}}\\
 &\ge B+k-1-\frac{D(k-1)(2n_0+2B)}{a_{k-1}}\\
 &=\frac{a_k}{D}+\frac{\delta_k}{Da_{k-1}}>0.
\end{aligned}
\]
Hence all $\widetilde h_k$, and therefore all $h_k=\widetilde h_k/k!$, are positive.
\end{proof}

Applying this ratio criterion to the following rational data gives the
spectral estimates needed at the smaller radii.
Define $(n_r,\beta_r,D_r)$ for $4\le r\le16$ by
\[
\begin{array}{c|ccccccccccccc}
 r&4&5&6&7&8&9&10&11&12&13&14&15&16\\ \hline
 n_r&15&21&28&32&45&50&66&72&78&91&105&120&136\\
 \beta_r&17&25&34&165/4&53&62&75&84&94&106&119&132&145\\
 D_r&4&2&2&4&10&1&2&10&1&2&1&2&10
\end{array}
\]

\begin{lemma}\label{lem:explicit-spectral-certificates}
For every $4\le r\le16$, Lemma~\ref{lem:positive-hessenberg}
applies with $(n_0,B)=(n_r,\beta_r)$.
\end{lemma}
\begin{proof}
For each $r$, take $(n_0,B,D)=(n_r,\beta_r,D_r)$.
In the table, the two lines for each radius give $(a_1,\ldots,a_r)$
and the corresponding $(\delta_2,\ldots,\delta_r)$ from
Lemma~\ref{lem:spectral-ratio-certificate}, respectively.
\begingroup
\small
\renewcommand{\arraystretch}{1}
\setlength{\LTpre}{3pt}
\setlength{\LTpost}{3pt}
\begin{longtable}{@{}r@{\quad}l@{}}
 $r$&$(a_1,\ldots,a_r)$ and $(\delta_2,\ldots,\delta_r)$\\
 \midrule
\endhead
 4&$(68, 56, 39, 1)$\\*
   &$(64, 24, 9)$\\[2pt]
 5&$(50, 44, 37, 26, 1)$\\*
   &$(32, 12, 6, 10)$\\[2pt]
 6&$(68, 62, 56, 47, 33, 2)$\\*
   &$(48, 0, 24, 37, 28)$\\[2pt]
 7&$(165, 154, 142, 127, 107, 75, 1)$\\*
   &$(131, 86, 68, 22, 50, 36)$\\[2pt]
 8&$(530, 503, 472, 435, 389, 328, 231, 6)$\\*
   &$(10, 34, 200, 335, 28, 152, 14)$\\[2pt]
 9&$(62, 59, 56, 53, 49, 44, 37, 26, 1)$\\*
   &$(24, 24, 0, 5, 7, 20, 23, 2)$\\[2pt]
 10&$(150, 144, 138, 131, 123, 114, 102, 86, 61, 1)$\\*
   &$(72, 48, 66, 73, 18, 72, 60, 6, 35)$\\[2pt]
 11&$(840, 812, 783, 750, 713, 671, 621, 558, 472, 335, 8)$\\*
   &$(720, 124, 360, 450, 147, 9, 192, 384, 40, 220)$\\[2pt]
 12&$(94, 91, 88, 85, 81, 77, 73, 68, 61, 52, 37, 2)$\\*
   &$(32, 40, 24, 69, 62, 15, 1, 36, 15, 44, 27)$\\[2pt]
 13&$(212, 206, 200, 194, 187, 179, 171, 161, 149, 134, 114, 81, 2)$\\*
   &$(120, 144, 72, 98, 161, 31, 83, 111, 120, 52, 106, 42)$\\[2pt]
 14&$(119, 116, 113, 110, 106, 102, 98, 94, 88, 82, 74, 63, 45, 2)$\\*
   &$(28, 32, 12, 78, 92, 66, 0, 82, 16, 30, 30, 42, 26)$\\[2pt]
 15&$(264, 258, 252, 246, 239, 231, 223, 214, 204, 193, 179, 162, 138, 100, 9)$\\*
   &$(96, 96, 0, 54, 197, 147, 160, 136, 12, 105, 20, 108, 12, 76)$\\[2pt]
 16&$(1450, 1421, 1390, 1358, 1324, 1287, 1247, 1204, 1156, 1102, 1040, 965, 871, 741, 528, 3)$\\*
   &$(350, 1280, 980, 628, 1012, 1281, 652, 696, 528, 20, 600, 135, 169, 142, 216)$\\
\end{longtable}
\endgroup
Every $a_k$ is positive, every $\delta_k$ is nonnegative,
and $a_1=D_r\beta_r$. Hence
Lemma~\ref{lem:spectral-ratio-certificate} applies in each case.
\end{proof}

\section{Uniform exclusion}\label{sec:uniform-exclusion}
A signed-orbit inequality first excludes the smaller dimensions. We then apply Section~\ref{sec:boundary-methods} to the three remainder ranges $1\le c\le r-1$, $c=0$, and $c=r$.

\subsection{A signed-orbit obstruction}
We compare four signed-permutation orbits using weights of total mass zero.

Fix $r\ge3$, $n\ge r+3$, and let $\mathcal O_-$, $\mathcal O_0$,
$\mathcal O_2$, $\mathcal O_3$ be the signed-permutation orbits of $(1^{r+1})$, $(1^{r+3})$, $(1^{r+1},2)$, and $(1^r,3)$, respectively.
Choosing the nonzero coordinates, their signs, and the distinguished
coordinate in the last two orbits gives
\[
\begin{array}{c|cccc}
j&-&0&2&3\\\hline
M_j=|\mathcal O_j|&2^{r+1}\binom n{r+1}&2^{r+3}\binom n{r+3}
 &2^{r+2}n\binom{n-1}{r+1}&2^{r+1}n\binom{n-1}r
\end{array}
\]
For $\|\mathbf x\|_1=2r+1$, signed-permutation symmetry lets us assume $x_i\ge0$. Put $m=|\{i:x_i>0\}|$, $a=|\{i:x_i\ge2\}|$, $b=|\{i:x_i\ge3\}|$, and $I_j=|\mathcal O_j\cap(\mathbf x+B_n(r))|$.
The identity
\[
\|\mathbf x-\mathbf y\|_1=\|\mathbf x\|_1+\|\mathbf y\|_1
 -2\sum_{i:y_i>0}\min(x_i,y_i)
\]
shows that a negative-orbit vector contributes exactly when all its
$r+1$ entries are positive and inside the support. Thus
$I_-=\binom m{r+1}$.
For a positive orbit, its norm is $r+3$, and the allowed defect
$\|\mathbf y\|_1-\sum_{y_i>0}\min(x_i,y_i)$ is at most one.
A good unit entry is positive and inside the support; a bad unit
entry is either negative inside it or has either sign outside it.
Consequently the three counts are
\begin{align*}
I_0&=\binom m{r+3}
 +\{m-r-2+2(n-m)\}\binom m{r+2},\\
I_2&=(m-a)\binom{m-1}{r+1}
 +a\left[\binom{m-1}{r+1}
 +\{m-r-1+2(n-m)\}\binom{m-1}r\right]\\
 &=m\binom{m-1}{r+1}+a(2n-m-r-1)\binom{m-1}r,\\
I_3&=(a-b)\binom{m-1}r
 +b\left[\binom{m-1}r
 +\{m-r+2(n-m)\}\binom{m-1}{r-1}\right]\\
 &=a\binom{m-1}r+b(2n-m-r)\binom{m-1}{r-1}.
\end{align*}
In the second count, a positive entry two at a coordinate equal to
one spends the defect allowance; at a coordinate at least two it
leaves room for one bad unit. The third count has the same alternatives
for an entry three at coordinates equal to two or at least three.
A nonexistent choice contributes zero.

When $m\ge r+1$, set $t=m-r-1$ and $v=n-r-1$. Then
$0\le t\le\min(r,v)$, $v\ge2$, and division by $I_-/M_-$ gives
\begin{align*}
F_0=\frac{I_0/M_0}{I_-/M_-}
 &=\frac{t\{(2v-t-1)(r+3)+t-1\}}{4v(v-1)},\\
F_2=\frac{I_2/M_2}{I_-/M_-}
 &=\frac t{2v}+\frac{a(r+1)(2v-t)}{2v(r+t+1)},\\
F_3=\frac{I_3/M_3}{I_-/M_-}
 &=\frac a{r+t+1}+\frac{br(2v-t+1)}{(r+t+1)(t+1)}.
\end{align*}
These expressions follow by substituting the orbit cardinalities and cancelling binomial coefficients.
If $b=0$, then $2r+1=m+a$, so $a=r-t$.
If $b\ge1$, then $a\ge b\ge1$ and
$2r+1\ge m+a+b\ge r+3+t$, so $t\le r-2$.
Since the coefficients of $a,b$ in $F_2,F_3$ are nonnegative,
only the following two cases need to be considered:
\[
(a,b)=(r-t,0),\quad 0\le t\le r;
\qquad
(a,b)=(1,1),\quad0\le t\le r-2.
\]
The additional restriction $t\le v$ remains in force whenever it is
needed in a lower-dimensional calculation.
For positive weights $w_0,w_2,w_3$, write
$\mathscr F(t,a,b)=\sum_{j=0,2,3}w_j(F_j-1)$.
These orbit ratios give the following signed-mass criterion.

\begin{lemma}\label{lem:four-orbit-duality}
If $\mathscr F(t,a,b)>0$ for all integers $0\le t\le\min(r,v)$
in the first case and $0\le t\le\min(r-2,v)$ in the second,
then $B_n(r)$ has no lattice tiling.
\end{lemma}
\begin{proof}
Put
$g=\sum_{j=0,2,3}w_j\bigl(\mathbf1_{\mathcal O_j}/M_j
 -\mathbf1_{\mathcal O_-}/M_-\bigr)$.
Then $g=0$ on $B_n(r)$ and $\sum g=0$. At norm $2r+1$,
\[
(g*\mathbf1_{B_n(r)})(\mathbf x)
=\frac{I_-}{M_-}\mathscr F(t,a,b)>0\quad\text{if }I_->0;
\]
if $I_-=0$, the convolution is nonnegative. At larger norms the
negative orbit contributes nothing, so the convolution is again
nonnegative.
If a subgroup $L$ tiles, every nonzero $\boldsymbol{\ell}\in L$ has
$\|\boldsymbol{\ell}\|_1\ge2r+1$, since $B_n(r)-B_n(r)=B_n(2r)$.
Choose $\mathbf y\in\mathcal O_-$. Its unique covering center
$\boldsymbol{\ell}$ is nonzero, has $I_->0$, and satisfies
$\|\boldsymbol{\ell}\|_1\le\|\mathbf y\|_1+r=2r+1$.
The summand at this center is strictly positive; all other summands
are nonnegative, including the zero summand at the origin. However,
\[
0<\sum_{\boldsymbol{\ell}\in L}(g*\mathbf1_{B_n(r)})(\boldsymbol{\ell})
=\sum_{\mathbf y}g(\mathbf y)
 \sum_{\boldsymbol{\ell}\in L}\mathbf1_{B_n(r)}(\boldsymbol{\ell}-\mathbf y)
=\sum_{\mathbf y}g(\mathbf y)=0,
\]
a contradiction. All sums containing $g$ are finite.
\end{proof}

We shall use one elementary endpoint principle. If $p(y)=a_0+a_1y+a_2y^2-a_3y^3$ and $a_0,a_1,a_3>0$,
then
\[
\left(\frac{p(y)}{y^2}\right)'
=-\frac{2a_0}{y^3}-\frac{a_1}{y^2}-a_3<0\quad(y>0).
\]
If $p(Y)>0$, then $\frac{p(y)}{y^2}\ge\frac{p(Y)}{Y^2}>0$ $(0<y\le Y)$, and $p(0)=a_0>0$.
This reduces each cubic below to its value and derivative at the
lower endpoint and its value at the upper endpoint.
The following weights satisfy that criterion throughout the geometric range needed below. Put $N_3=7$ and $N_r=r(r+1)/2+2$ for $r\ge4$.

\begin{proposition}\label{thm:uniform-triangular-orbits}
There is no lattice tiling by $B_n(r)$ whenever $r\ge3$ and $3\le n\le N_r$.
\end{proposition}
\begin{proof}
Lemma~\ref{lem:known-small} rules out $n\le5$. Moreover, Lemma~\ref{lem:known-post} excludes several low-dimensional ranges: it covers $6\le n\le r+2$ for $r\ge6$, as $\frac n{\sqrt2}-\frac{3}{2\sqrt2}-\frac12\le\frac{r+1/2}{\sqrt2}-\frac12<r$; it covers $n\le2r-1$ for $r=4,5$, since $r+\frac12-\frac{2r-5/2}{\sqrt2}\ge\frac{11\sqrt2-15}{2\sqrt2}>0$ (from $(11\sqrt2)^2=242>225=15^2$); and it covers $n=6$ for $r=3$, where the condition $\frac9{2\sqrt2}-\frac12\le3$ simplifies to $81\le98$. The remaining case at radius three is $(r,n)=(3,7)$, which is excluded by Lemma~\ref{lem:four-orbit-duality} with weights $(w_0,w_2,w_3)=(15,10,4)$, as direct substitution into the orbit ratios yields
\[
\begin{array}{c|cccc|cc}
(a,b)&\multicolumn{4}{c|}{(3-t,0)}&\multicolumn{2}{c}{(1,1)}\\
t&0&1&2&3&0&1\\\hline
\mathscr F&4&13/5&115/36&9/4&3&7/3
\end{array}
\]

Next, suppose $r\ge6$ and $r+3\le n<2r$, so that $2\le v<r-1$. Choosing $(w_0,w_2,w_3)=\bigl(\frac4{r+3},\frac2{r+1},\frac14\bigr)$ gives $\sum_{j=0,2,3}w_j\le\frac49+\frac27+\frac14=\frac{247}{252}<1$; by Lemma~\ref{lem:four-orbit-duality}, it suffices to verify $\sum_jw_jF_j\ge1$, which implies $\mathscr F\ge1-\sum_{j=0,2,3}w_j>0$. In the first branch $(a,b)=(r-t,0)$, the case $t=0$ is immediate from $w_2F_2=2r/(r+1)>1$. For $t\ge1$, the nonnegative quantities $v-t$ and $r-v$ give
\begin{align*}
w_0F_0&=1-\frac{(v-t)(v-t-1)}{v(v-1)}
 +\frac{t(t-1)}{(r+3)v(v-1)}
 \ge1-\frac{(v-t)(v-t-1)}{v(v-1)},\\
w_2F_2&=\frac{t}{v(r+1)}+\frac{(r-t)(2v-t)}{v(r+t+1)}
 \ge\frac{(r-t)(2v-t)}{v(r+t+1)}.
\end{align*}
Cross-multiplying these lower bounds shows that $w_0F_0+w_2F_2\ge1$, because the numerator difference is
\begin{align*}
&(r-t)(2v-t)(v-1)-(v-t)(v-t-1)(r+t+1)\\
&\quad=(r-v)\{(v-t)(v-t-1)+t(t-1)+3(v-t)t\}\\
&\qquad+(v-t)\{(v-t-1)^2+(v-t)t+t^2+t\}\ge0.
\end{align*}
In the second branch $(a,b)=(1,1)$, if $t=0$, then $w_3F_3=\frac{1+r(2v+1)}{4(r+1)}\ge\frac{5r+1}{4(r+1)}=1+\frac{r-3}{4(r+1)}>1$. If $1\le t\le r-2$, then $r/(r+t+1)\ge1/2$, whence
\[
w_0F_0+w_3F_3\ge
\frac{t(2v-t-1)}{v(v-1)}+\frac{2v-t+1}{8(t+1)}.
\]
Multiplying the difference between this lower bound and $1$ by $8(t+1)v(v-1)$ gives
\begin{align*}
 &8t(t+1)(2v-t-1)+(2v-t+1)v(v-1)-8(t+1)v(v-1)\\
 &\quad=t(t-1)(t+1)
 +\frac{v-t}{8}\{(4(v-t)-3t-9)^2+(t-1)(23t+25)\}\ge0,
\end{align*}
which completes the range $n<2r$.

It remains to treat the main range $r\ge4$ and $2r\le n\le N_r$, where $v$ lies between $v_{\min}=r-1$ and $v_{\max}=(r^2-r+2)/2$. Setting $D_{\mathrm{orb}}=v(2r+5)-r^2-3r-4\ge r^2-9>0$, we employ the unified weight family
\[
w_0=\frac{4v(v-1)}{rD_{\mathrm{orb}}},\qquad
w_2=\frac8{4r+7},\qquad
w_3=\frac{r-1}{4v}.
\]
Clearing the positive denominators of $\mathscr F$ in the two branches, we define
\begin{align*}
\mathscr F_0(t,v)&=4rv(4r+7)(r+t+1)D_{\mathrm{orb}}\,\mathscr F(t,r-t,0),\\
\mathscr F_1(t,v)&=4rv(4r+7)(r+t+1)(t+1)D_{\mathrm{orb}}\,\mathscr F(t,1,1).
\end{align*}
As polynomials in $v$, both $\mathscr F_0$ and $\mathscr F_1$ are cubics with negative leading coefficients $-16(4r+7)(r+t+1)$ and $-16(4r+7)(r+t+1)(t+1)$, respectively. By the cubic endpoint principle, positivity on $[v_{\min},v_{\max}]$ reduces to verifying $\mathscr F_i(t,v_{\min})>0$, $\partial_v\mathscr F_i(t,v_{\min})>0$, and $\mathscr F_i(t,v_{\max})>0$ for $i=0,1$, which we establish below without any parameter search.

For the first branch ($b=0$), let $z=t/r\in[0,1]$. Expanding at the lower endpoint $v_{\min}=r-1$ in the Bernstein basis gives
\[
 \mathscr F_0(t,v_{\min})=\sum_{i=0}^3\binom3i L_i^{(0)} z^i(1-z)^{3-i},\qquad
 \partial_v\mathscr F_0(t,v_{\min})=\sum_{i=0}^3\binom3i J_i^{(0)} z^i(1-z)^{3-i},
\]
with coefficients
\begin{align*}
 L_0^{(0)}&=(r-1)\{(r-4)[32r^4+r^2(60r-71)+40r+495]+1756\},\\
 L_1^{(0)}&=\frac{r-1}{3}\{(r-4)[72r^4+r^2(62r-73)+1378r+6013]+23380\},\\
 3L_2^{(0)}&=(r-4)(80r^5+68r^4+447r^3+3433r^2+11697r+46119)+185148,\\
 L_3^{(0)}&=(2r+1)\{(r-4)[r^3(12r-7)+15r^2+211r+341]+1588\},\\
 J_0^{(0)}&=(r-1)\{(r-4)(96r^3+376r^2+910r+3701)+15364\},\\
 3J_1^{(0)}&=(r-4)[208r^4+r^2(196r-348)+1687r+7237]+27268,\\
 3J_2^{(0)}&=(r-4)(208r^4+164r^3+134r^2+4225r+16381)+63844,\\
 J_3^{(0)}&=(2r+1)\{(r-4)(40r^3+42r^2+211r+1455)+5260\}.
\end{align*}
Since $r\ge4$, the factors $60r-71\ge169$, $62r-73\ge175$, $12r-7\ge41$, and $196r-348\ge436$ are all positive, whence $L_i^{(0)},J_i^{(0)}>0$. Because $\sum_{i=0}^3\binom3i z^i(1-z)^{3-i}=1$, we obtain $\mathscr F_0(t,v_{\min})>0$ and $\partial_v\mathscr F_0(t,v_{\min})>0$. At the upper endpoint $v_{\max}$, we set $\mathcal R(x)=\mathscr F_0(r-x,v_{\max})$ for the nonnegative integer $x=r-t$. Its first two values factor directly as
\begin{align*}
 \mathcal R(0)&=\frac{r(2r+1)}2
 \{(r-4)[r^3(8r-6)+r(87r+268)]+929r+174\}>0,\\
 \mathcal R(1)&=\frac{r(r-1)}2
 \{4r^3[(r-4)(4r-7)+5]+(r-4)(37r-49)+170\}>0.
\end{align*}
For $x\ge2$, expansion in powers of $x-2\ge0$ gives
$\mathcal R(x)=q_0+q_1(x-2)+q_2(x-2)^2+q_3(x-2)^3$, where
\begin{align*}
 q_0&=(r-4)\left[r^5(8r-26)+40r^4
                 +\frac{43r^3+417r}{2}+3r^2+21\right]+420,\\
 q_1&=(r-4)(4r^5+54r^4+107r^3+490r^2+1771r+6338)+25968,\\
 q_2&=(r-4)(8r^5+38r^4+142r^3+618r^2+2382r+9444)+38112,\\
 q_3&=2(r+2)(4r+7)(r^2-r+2).
\end{align*}
Every $q_i$ is manifestly positive for $r\ge4$, establishing $\mathscr F_0(t,v_{\max})=\mathcal R(r-t)>0$.

For the second branch ($a=b=1$), set $z=t/(r-2)\in[0,1]$ and use $r-4\ge0$. The lower-endpoint expansions are
\begin{align*}
 \mathscr F_1((r-2)z,v_{\min})
 &=\sum_{j=0}^4\binom4j L_j^{(1)}z^j(1-z)^{4-j},\\
 \partial_v\mathscr F_1((r-2)z,v_{\min})
 &=\sum_{j=0}^4\binom4j J_j^{(1)}z^j(1-z)^{4-j},
\end{align*}
where
\begin{align*}
 L_0^{(1)}&=2(r-1)^2(4r+7)\bigl[32+(r-4)\{r(r^2-1)+4r^2+4\}\bigr],\\
 L_1^{(1)}&=\frac{r-1}{4}\bigl[9064+(r-4)\{r^3(24(r-4)^2+278(r-4)+693)+24r^2+823r+2126\}\bigr],\\
 6L_2^{(1)}&=r^5(52(r-4)^2+235(r-4)+104)+r^3(324(r-4)+478)
       +r(1077(r-4)+3394)+336,\\
 4L_3^{(1)}&=r^5(80(r-4)^2+288(r-4)+203)+r^3(783(r-4)+447)
       +r(2465(r-4)+8998)+112,\\
 L_4^{(1)}&=(r-1)\bigl[r^4(28(r-4)^2+133(r-4)+87)+r(373(r-4)^2+2401(r-4)+3490)+112\bigr],\\
 J_0^{(1)}&=2(r-1)(4r+7)\bigl[r^2(3(r-4)^2+27(r-4)+22)+16r+40\bigr],\\
 4J_1^{(1)}&=r^4(80(r-4)^2+652(r-4)+348)+r(585(r-4)^2+6909(r-4)+15638)+896,\\
 6J_2^{(1)}&=r\bigl[r^3(152(r-4)^2+798(r-4)+347)+(1458r+5191)(r-4)+20622\bigr],\\
 4J_3^{(1)}&=r^4(208(r-4)^2+836(r-4)+254)+r^2(2553(r-4)+6211)+2222(r-4)+8440,\\
 J_4^{(1)}&=(r-1)\bigl[r^3(96(r-4)^2+492(r-4)+195)+(1105r+3766)(r-4)+14952\bigr].
\end{align*}
All coefficients are positive for $r\ge4$, so both lower-endpoint quantities are positive.

At the upper endpoint $v_{\max}$, set $Q_r(z)=\mathscr F_1((r-2)z,v_{\max})$. For $r=4$, the integer restriction $0\le t\le2$ gives $(\mathscr F_1(0,7),\mathscr F_1(1,7),\mathscr F_1(2,7))=(370944,1736,119496)>0$ directly. For $r\ge5$, applying the identities
\[
 z^3=(z-\tfrac35)^2(z+\tfrac65)+\tfrac{27}{25}z-\tfrac{54}{125},
 \qquad
 z^2=(z-\tfrac23)^2+\tfrac43z-\tfrac49
\]
yields the following decomposition
\[
 Q_r(z)=A_rz^3(1-z)+C_r(z-\tfrac35)^2(z+\tfrac65)
 +\mathcal H_r(z-\tfrac23)^2+\frac{(r-2)\ell_r}{150}z
 +\frac{k_r}{2250},
\]
where
\begin{align*}
 A_r&=2(r-2)^4(r+2)(4r+7)(r^2-r+2),\\
 C_r&=2(r-2)^3
 (4r^6+3r^5-15r^4+8r^3-28r^2-74r+28),\\
 \mathcal H_r&=\frac{(r-2)^2}{2}
 (56r^6+46r^5+87r^4+184r^3-11r^2-178r+112),\\
 \ell_r&=96r^8+488r^7-9264r^6+26620r^5-20555r^4
 -13570r^3+56463r^2-71742r+22288,\\
 k_r&=1224r^9-2926r^8+27215r^7-226913r^6+440945r^5
 -444765r^4\\
 &\hspace{2em}-341968r^3+1141908r^2-1098032r+211456.
\end{align*}
Here $A_r>0$, and grouping terms confirms the positivity of $C_r$ and $\mathcal H_r$:
\begin{align*}
 \frac{C_r}{2(r-2)^3}
 &=r^4(4r^2-15)+r^2(3r^3-28)+r(8r^2-74)+28>0,\\
 \frac{2\mathcal H_r}{(r-2)^2}
 &=56r^6+46r^5+r^2(87r^2-11)+r(184r^2-178)+112>0.
\end{align*}
Moreover, expansion in powers of $r-5\ge0$ shows that $\ell_r$ and $k_r$ have strictly positive expansions:
\begin{align*}
 \ell_r
 &=96(r-5)^8+4328(r-5)^7+75016(r-5)^6\\
 &\quad+676900(r-5)^5+3505945(r-5)^4+10545330(r-5)^3\\
 &\quad+17219663(r-5)^2+12060138(r-5)+594528>0,\\
 k_r
 &=1224(r-5)^9+52154(r-5)^8+1011775(r-5)^7\\
 &\quad+11529412(r-5)^6+83829430(r-5)^5+398489610(r-5)^4\\
 &\quad+1223494607(r-5)^3+2301528888(r-5)^2\\
 &\quad+2336640948(r-5)+908810496>0.
\end{align*}
Since $0\le z\le1$, every summand in the decomposition of $Q_r(z)$ is nonnegative and the constant term $k_r/2250$ is strictly positive, ensuring $Q_r(z)>0$.

Having verified all three endpoint conditions for both branches, we apply the cubic endpoint principle with $y=v-v_{\min}$ and $Y=v_{\max}-v_{\min}=(r^2-3r+4)/2>0$ to conclude that $\mathscr F_i(t,v)>0$, and hence $\mathscr F(t,a,b)>0$. Lemma~\ref{lem:four-orbit-duality} completes the proof.
\end{proof}

The orbit obstruction leaves the common range
\begin{equation}\label{eq:geometric-range}
 n\ge8\quad(r=3),\qquad n>\frac{r(r+1)}2+2\quad(r\ge4).
\end{equation}
The next three subsections treat the disjoint cases
$1\le c\le r-1$, $c=0$, and $c=r$, respectively.

\subsection{Intermediate remainders}
We first combine the moment and scalar-kernel bounds for
$1\le c\le r-1$ in the range~\eqref{eq:geometric-range}.
In particular $2n>r(r+1)>r^2$, and
$2n\ge r(r+1)+6$ for $r\ge4$. The division identity is $2n=(r+1)q+c$.
Retain the number $E$ of minimum fibers.

We collect the density, capacity, and collision estimates in a form that can be checked at one endpoint. For $r\ge3$, $1\le k\le r-1$, and $2n\ge r^2$, define
\begin{align*}
 &e_{r,k}(2n)=\frac{k+1-r(r-1)/(2n)}{r+1}
                     \left(1-\frac r{2n+1}\right),\\
 &\gamma_r(2n)=\frac{r^2}{2n}
 \left(1+\frac{2(r-1)}{2n}+\frac{2(r-1)^2}{(2n)^2-(r-1)^2}\right),\\
 &\omega_r(2n)=\gamma_r(2n)+\frac{4nr(r-1)}{(2n+1)(2n+1-r)}.
\end{align*}
The density estimate uses this family with second index $r-c$.

\begin{lemma}\label{lem:positive-scalar-comparison}
Fix $r\ge3$ and $1\le c\le r-1$. On $2n\ge r^2$,
$e_{r,r-c}(2n)$ increases with $n$, whereas
$\gamma_r(2n)$, $\omega_r(2n)$, and $(r)_c/(2n+1)^c$ decrease.
Every lattice tiling satisfies
$e_{r,r-c}(2n)\le E/V_r\le |A[q+1]|\frac{(r)_c}{(2n+1)^c}$.
If $|A[q+1]|>1$, it also satisfies
\[
 (r-c)\max\left\{2e_{r,r-c}(2n)-1,e_{r,r-c}(2n)-\frac13\right\}\le \omega_r(2n).
\]
If $(q+1) A=\{0\}$, then $(r-c) e_{r,r-c}(2n)\le\gamma_r(2n)$.
Finally, $r-c\ge2$ forces $|A[q+1]|>1$.
\end{lemma}
\begin{proof}
Lemma~\ref{lem:minimal-fibers} and the ratios in~\eqref{eq:successive-ratios} give
\begin{align*}
 \frac E{V_r}\ge\frac{r-c+1-(r-1)\sigma_{r-1}/\sigma_r}{r+1}
                                      \frac{\sigma_r}{V_r}\ge\frac{r-c+1-r(r-1)/(2n)}{r+1}
                         \left(1-\frac{V_{r-1}}{V_r}\right)\ge e_{r,r-c}(2n).
\end{align*}
Scalar capacity gives the upper bound
$E/V_r\le |A[q+1]| V_{r-c}/V_r\le |A[q+1]|\frac{(r)_c}{(2n+1)^c}$.
Both factors of $e_{r,r-c}$ are positive and increasing. The terms
$(2n)^{-1}$, $(2n)^{-2}$, and $[2n\{(2n)^2-(r-1)^2\}]^{-1}$ decrease, since $\frac{1}{2}\frac{d}{dn}\{2n((2n)^2-(r-1)^2)\}=3(2n)^2-(r-1)^2>0$.
Moreover
\[
 \frac{1}{2}\frac{d}{dn}\frac {2n}{(2n+1)(2n+1-r)}
 =\frac{1-r-(2n)^2}{(2n+1)^2(2n+1-r)^2}<0.
\]
Finally,
$\frac{d}{dn}\bigl((r)_c/(2n+1)^c\bigr)
=-2c(r)_c/(2n+1)^{c+1}<0$.
This proves the asserted monotonicities.
Lemma~\ref{lem:exact-layered-cancellation}, specifically~\eqref{eq:layered-cancellation-ratio}, gives
$\mathcal C/V_r\le\gamma_r(2n)$, while the volume ratios give
\begin{align*}
 \frac{4n}{V_r}\sum_{j=1}^{r-c}V_{r-j-1}
 \le4n\frac{V_{r-2}}{V_r}\sum_{j\ge0}\left(\frac r{2n+1}\right)^j
 \le\frac{4nr(r-1)}{(2n+1)(2n+1-r)}.
\end{align*}
The two inequalities in Lemma~\ref{lem:all-minimum-kernel} now prove the bound involving $\omega_r$.
If $(q+1) A=\{0\}$, every minimum representative is zero; hence
$\mathcal C=\sum(r-c)=(r-c)E$, giving the bound involving $\gamma_r$.
For $r-c\ge2$, the inequality $r(r-1)/(2n)<1$ gives
\[
 e_{r,r-c}(2n)>\frac{2(2n+1-r)}{(r+1)(2n+1)}>\frac r{2n+1}
 \ge\frac{V_{r-c}}{V_r},
\]
where the second strict inequality follows from $2(2n+1-r)-r(r+1)\ge r^2-3r+2=(r-1)(r-2)>0$.
Thus $|A[q+1]|=1$ would contradict $E\le V_{r-c}$.
\end{proof}

Lemma~\ref{lem:odd-tail-occupancy} now contradicts these density estimates for the shortest tails.

\begin{lemma}\label{lem:unit-tail}
There is no lattice tiling with $r-c=1$ in the remaining geometric range.
The same holds for $(r,r-c)=(5,3),(6,3)$.
\end{lemma}
\begin{proof}
For $r-c=1$, the identity $2n=(r+1)(q+1)-2>r(r+1)$ implies
$q+1\ge r+1$, hence $2n\ge 2n_0=r^2+2r-1$. The odd-tail
occupancy bound and the scalar comparison would require $e_{r,1}(2n_0)\le\frac E{V_r}\le\frac2{2r+1}$.
But direct common-denominator subtraction gives
\[
 e_{r,1}(2n_0)-\frac2{2r+1}
 =\frac{(r-1)(3r-2)}{(r+2)(2r+1)(r^2+2r-1)}>0.
\]
For $r-c=3$, the least odd integer $\ell$ with $3\ell\ge2r+1$
is $5$ at both $r=5$ and $r=6$. Thus $E/V_r\le2/5$.
The remainder congruences give $n\equiv1\pmod3$ and
$n\equiv5\pmod7$, respectively. The geometric cutoffs therefore
force $n\ge19$ and $n\ge26$. In these two cases,
\[
 e_{5,3}(38)-\frac25=\frac{388}{3705}>0,
 \qquad
 e_{6,3}(52)-\frac25=\frac{1623}{48230}>0.
\]
Monotonicity gives the contradiction in every larger dimension.
\end{proof}

For the middle tails, we combine the two-shell defect inequality with the sign-restricted moment bounds.

\begin{lemma}\label{lem:tail-two}
A lattice tiling by a Lee ball is impossible in each of the following ranges:
\[
\begin{array}{c|c|c}
 r&r-c&2n\\\hline
 r\ge7&3\le r-c\le r-4&2n\ge r^2+r+4\\
 r\ge10&r-c\in\{2,r-3\}&2n\ge r^2+2r-2\\
 r=8&r-c\in\{2,5\}&2n\ge82.
 \end{array}
\]
\end{lemma}
\begin{proof}
For the three rows of the lemma, take
\[
 (g,2n_0)=(3r-12,r^2+r+4),\quad
 (2r-6,r^2+2r-2),\quad (10,82),
\]
respectively. In every case $(r-c)(c-1)\ge g$, $2n\ge2n_0\ge r^2$,
and $g\ge2r-6$. Since $\sigma_0/\sigma_1=1/(2n)$, the shell
recurrence gives, for $2\le k\le r$,
\[
 \frac{\sigma_{k-1}}{\sigma_k}
 =\frac{k}{2n+(k-2)\sigma_{k-2}/\sigma_{k-1}},\qquad
 \frac{k}{2n+(k-2)(k-1)/(2n)}
 \le\frac{\sigma_{k-1}}{\sigma_k}\le\frac{k}{2n}.
\]
For $3\le k\le r$, applying the lower bound with index $k-1$ gives
\[
 \frac{\sigma_{k-1}}{\sigma_k}\le\frac{k}{2n}
 \left(1-\frac{(k-1)(k-2)}{(2n)^2+2(k-2)^2}\right).
\]
To combine the bounds at indices $r$ and $r-1$, use $2n\ge r^2$
and $(1+x)^{-1}\ge1-x$ for $x\ge0$ to obtain
\begin{align*}
 &\frac{(r-1)(r-2)}{(2n)^2+2(r-2)^2}
  +\frac{(r-2)(r-3)}{(2n)^2+2(r-3)^2}\\
 &\qquad-\frac{(r-2)^2(r-1)(r-3)}
 {\bigl((2n)^2+2(r-2)^2\bigr)\bigl((2n)^2+2(r-3)^2\bigr)}\\
 &\quad\ge\frac{2(r-2)^2}{(2n)^2+2(r-2)^2}
       -\frac{(r-2)^2(r-1)(r-3)}{(2n)^4}\\
 &\quad\ge\frac{2(r-2)^2}{(2n)^2}
       -\frac{5(r-2)^4-(r-2)^2}{(2n)^4}
 \ge\frac{2(r-2)^2-5}{(2n)^2}.
\end{align*}
The last step uses $5(r-2)^4-(r-2)^2\le5(2n)^2$. Therefore
\begin{align*}
 \frac{\sigma_{r-2}}{\sigma_r}
 &=\frac{\sigma_{r-1}}{\sigma_r}\frac{\sigma_{r-2}}{\sigma_{r-1}}\\
 &\le\frac{(r)_2}{(2n)^2}
 \left(1-\frac{(r-1)(r-2)}{(2n)^2+2(r-2)^2}\right)
 \left(1-\frac{(r-2)(r-3)}{(2n)^2+2(r-3)^2}\right)\\
 &\le\frac{(r)_2}{(2n)^2}
 \left(1-\frac{2(r-2)^2-5}{(2n)^2}\right).
\end{align*}
For every other shell we use
$\sigma_{r-j}/\sigma_r\le(r)_j/(2n)^j$.
The explicit expansions of $\Phi_k$ and $\Psi_k$ in
Section~\ref{sec:boundary-methods} consequently give, for $j=3,4$,
\begin{align*}
 \frac{\Phi_{r-j}}{\sigma_r}
 &\le\bigl((2n-r+j)^2-2n+2r-2j\bigr)\frac{(r)_j}{(2n)^j}\\
 &\quad +(4n-2r+2j-1)(r-j-1)\frac{(r)_{j+1}}{(2n)^{j+1}}\\
 &\quad +(r-j-1)(r-j-2)\frac{(r)_{j+2}}{(2n)^{j+2}},\\
 \frac{\Psi_{r-3}}{\sigma_r}
 &\le\bigl((2n-r+3)^2-4n+3r-9\bigr)\frac{(r)_3}{(2n)^3}\\
 &\quad +(4n-2r+4)(r-4)\frac{(r)_4}{(2n)^4}
       +(r-4)(r-5)\frac{(r)_5}{(2n)^5}.
\end{align*}
The remaining mixed shadows satisfy
\begin{align*}
 \frac1{\sigma_r}\left(\Phi_0+
       \sum_{j=4}^{r-1}\frac r{r-j}\Phi_{r-j-1}\right)
 &\le2n(2n-1)\left(\frac{(r)_r}{(2n)^r}
             +\frac r{2n}\sum_{j=4}^{r-1}\frac{(r)_j}{(2n)^j}\right)\\
 &\le2n(2n-1)\frac r{2n}\frac{(r)_4}{(2n)^4}
       \left(\left(\frac r{2n}\right)^{r-5}
                    +\frac1{1-(r-4)/(2n)}\right)\\
 &\le2n(2n-1)\frac r{2n}\frac{(r)_4}{(2n)^4}
              \left(\frac19+\frac98\right)\\
 &<\frac43(2n)^2\frac r{r-4}\frac{(r)_5}{(2n)^5}.
\end{align*}
Here $(r/(2n))^{r-5}\le1/49<1/9$ and
$(r-4)/(2n)\le(r-4)/r^2<1/9$, because
$r^2-9r+36=(r-4)(r-5)+16>0$.
Also $\Phi_k\le2n(2n-1)\sigma_k$ follows by summing
\[
 (2n-a)^2-2n+2a=2n(2n-1)-a(4n-a-2)\le2n(2n-1)
 \qquad(0\le a\le k<2n).
\]
The estimates for the finite sum use
\[
 \frac{(r)_r/(2n)^r}{(r)_4/(2n)^4}\le\left(\frac r{2n}\right)^{r-4},
 \qquad
 \frac{(r)_{j+1}/(2n)^{j+1}}{(r)_j/(2n)^j}
 =\frac{r-j}{2n}\le\frac{r-4}{2n}\quad(4\le j<r).
\]
Lemma~\ref{lem:two-shell-defect} gives
\[
 \MomentU+2\MomentT+(r+1)(r-1)\sigma_{r-1}
 +(r-3)(r-2)\sigma_{r-2}
 \ge\bigl(r(r+1)+(r-c)(c-1)\bigr)\sigma_r.
\]
Insert the upper bound for $\MomentU$ from
Lemma~\ref{lem:sign-restricted-heavy-pairs}. Its term $r\sigma_r$
reduces $r(r+1)\sigma_r$ to $r^2\sigma_r$, and the coefficient
of $\sigma_{r-1}$ becomes
$(r+1)(r-1)-(r-1)=r(r-1)$. Thus
\begin{align*}
 r^2+(r-c)(c-1)
 &\le\frac{r\Phi_{r-2}}{(r-1)\sigma_r}
       +\frac{\Psi_{r-3}+2\MomentT}{\sigma_r}\\
 &\quad+r(r-1)\frac{\sigma_{r-1}}{\sigma_r}
       +(r-2)(r-3)\frac{\sigma_{r-2}}{\sigma_r}.
\end{align*}
Apply the upper bound for $\MomentT$ from the same lemma, the improved
bound for $\sigma_{r-2}/\sigma_r$ in the $\Phi_{r-2}$ term,
and the other shell bounds above. All their coefficients are
nonnegative. We obtain $r^2+g\le r^2+(r-c)(c-1)\le\mathscr W$, where
\begin{align*}
 \mathscr W={}&\frac r{r-1}\biggl[
 \bigl((2n-r+2)^2-2n+2r-4\bigr)\frac{(r)_2}{(2n)^2}
 \left(1-\frac{2(r-2)^2-5}{(2n)^2}\right)\\
 &\qquad +(4n-2r+3)(r-3)\frac{(r)_3}{(2n)^3}
       +(r-3)(r-4)\frac{(r)_4}{(2n)^4}\biggr]\\
 &+\bigl((2n-r+3)^2-4n+3r-9\bigr)\frac{(r)_3}{(2n)^3}\\
 &+(4n-2r+4)(r-4)\frac{(r)_4}{(2n)^4}
       +(r-4)(r-5)\frac{(r)_5}{(2n)^5}\\
 &+r(r-1)\frac{r}{2n}+(r-2)(r-3)\frac{(r)_2}{(2n)^2}\\
 &+2\sum_{j=3}^4\frac r{r-j+1}\biggl[
       \bigl((2n-r+j)^2-2n+2r-2j\bigr)\frac{(r)_j}{(2n)^j}\\
 &\qquad +(4n-2r+2j-1)(r-j-1)\frac{(r)_{j+1}}{(2n)^{j+1}}\\
 &\qquad +(r-j-1)(r-j-2)\frac{(r)_{j+2}}{(2n)^{j+2}}\biggr]
       +\frac83(2n)^2\frac r{r-4}\frac{(r)_5}{(2n)^5}.
\end{align*}
We now prove the opposite strict inequality on the stated ranges.
Direct multiplication gives
$3(2n)^6(r^2+g-\mathscr W)=\sum_{j=0}^6 C_j(2n)^j$,
where
\begin{align*}
 C_6&=3g,\qquad C_5=-6r^3+9r^2-6r,\\
 C_4&=6r^4-24r^3+6r^2-6r,\\
 C_3&=-29r^5+201r^4-460r^3+417r^2-144r,\\
 C_2&=3r^6-60r^4+6r^3+312r^2-288r,\\
 C_1&=3r^7-27r^6-9r^5+759r^4-2694r^3+3408r^2-1440r,\\
 C_0&=-6r^8+138r^7-1266r^6+5862r^5-14208r^4
                                      +16680r^3-7200r^2.
\end{align*}
For $r\ge7$,
\begin{align*}
 C_2&=3r^4(r^2-20)+6r(r^2+52r-48)>0,\\
 C_1&=r\bigl((r-7)[3r^3(r^2-2r-17)+402r^2+120r+4248]
                                  +28296\bigr)>0,\\
 C_0+6r^8&=6r^2\bigl((r-7)[r^3(23r-50)+627r^2+2021r+16927]
                                  +117289\bigr)>0,\\
 C_3&>-29r^5+201r^4-460r^3,\qquad
 C_4=6r((r-4)r^2+r-1)>0.
\end{align*}
Here $r^2-2r-17=(r-7)(r+5)+18>0$. Define
\[
 F(2n)=C_6(2n)^3+C_5(2n)^2+2nC_4-29r^5+201r^4-460r^3.
\]
Since $C_5<0$ and $2n\ge r^2$,
\begin{align*}
 F'(2n)=3C_6(2n)^2+4nC_5+C_4\ge 2n(3C_6r^2+2C_5)+C_4
   \ge12nr\{(r-7)(r+1)+5\}+C_4>0.
\end{align*}
Therefore
\[
 \sum_{j=0}^6C_j(2n)^j>(2n)^3F(2n)-6r^8
 \ge (2n)^3(F(2n_0)-6r^2).
\]
For the first row, expansion in powers of $r-7\ge0$ at $2n_0=r^2+r+4$ gives
\[
 \begin{aligned}
 F(2n_0)-6r^2&=3(r-7)^7+141(r-7)^6+2773(r-7)^5\\
 &\quad+29348(r-7)^4+177135(r-7)^3\\
 &\quad+582019(r-7)^2+831833(r-7)+82284>0.
 \end{aligned}
\]
For the second row,
\[
 F(2n_0)-6r^2=9r^6-83r^5-33r^4-286r^3+414r^2-492r+144.
\]
Expansion in powers of $r-10\ge0$ gives
\[
 \begin{aligned}
 &9(r-10)^6+457(r-10)^5+9317(r-10)^4+95394(r-10)^3\\
 &\qquad+492034(r-10)^2+1039988(r-10)+120624>0.
 \end{aligned}
\]
For the third row, $r=8$, $g=10$, $2n_0=82$, and direct substitution gives $F(82)-6\cdot8^2=107472>0$.
Thus $\mathscr W<r^2+g$ in every row, a contradiction.
\end{proof}

The two remaining short even tails are controlled by the density of a scalar image inside the second shell.

\begin{lemma}\label{lem:small-tail-two}
In the remaining geometric range, the cases $(r,r-c)=(4,2)$ and
$(6,2)$ are impossible, except possibly $(r,n)=(6,30)$.
\end{lemma}
\begin{proof}
First suppose $ (q+1)A\ne\{0\}$, and set $Y=\{\mathbf{0}\}\cup\{\mathbf{x}\in S_2:\pi(\mathbf{x})\in(q+1)A\}$ and $X=Y\setminus\{\mathbf{0}\}$.
Equation~\eqref{eq:minimum-parity-capacity} and $V_r={|A[q+1]|}|(q+1)A|$ give $E/V_r\le|Y|/|(q+1)A|$. For $r+1\in\{5,7\}$,
Lucas's identity and $\binom{r}{j}\equiv(-1)^j\pmod{r+1}$ give
\[
 {V_r}=\sum_{j=0}^{r}2^j\binom nj\binom{r}j
 \equiv\sum_{j=0}^{r}(-2)^j\binom{n\bmod (r+1)}{j}
 =(-1)^{n\bmod (r+1)}\not\equiv0\pmod{r+1}.
\]
Thus $(r+1)\nmid |(q+1)A|$. If $X$ is empty, $|Y|/|(q+1)A|\le1/3$.
If $X$ spans a line, $Y=\{\mathbf{0},\mathbf{v},-\mathbf{v}\}$; the odd order $k$ of
$\pi(\mathbf{v})$ satisfies $2k\ge2r+1$. Hence $k\ge r+1$;
since $(r+1)\nmid k$, in fact $k\ge r+3$, and
$|Y|/|(q+1)A|\le3/(r+3)$.
Otherwise put $t=r/2\in\{2,3\}$. Lemma~\ref{lem:planar-sumset} gives $|2Y|\ge3|X|$, $|3Y|\ge4|X|$, and $|tY|\ge(t+1)|X|$.
Since $tY\subseteq B_r$, its image is injective, and $|X|\ge4$.
The odd integer $|(q+1)A|$ therefore satisfies $|(q+1)A|\ge4(t+1)+1$.
Consequently
\[
 \frac{|Y|}{|(q+1)A|}\le\frac1{t+1}+\frac1{4t+5}<\frac3{2t+3}
 =\frac3{r+3},
\]
where the strict difference has numerator $3(t+1)(4t+5)-(2t+3)(5t+6)=2t^2-3>0$.
We have proved $E/V_r\le3/(r+3)$ in every nonzero-image case.
If $ (q+1)A=\{0\}$, Lemma~\ref{lem:positive-scalar-comparison} instead
requires $2e_{r,2}(2n)\le\gamma_r(2n)$.
Both alternatives contradict the following endpoint inequalities:
\[
\begin{array}{c|c|c|c}
 r&2n_0&e_{r,2}(2n_0)-3/(r+3)&2e_{r,2}(2n_0)-\gamma_r(2n_0)\\\hline
 4&32&101/3080&114369/357280\\
 6&74&149/19425&54551664/435307775.
\end{array}
\]
Every displayed number is positive; monotonicity extends the
contradictions to $2n\ge {2n_0}$. The congruences
$n\equiv1\pmod5$ and $n\equiv2\pmod7$, together with the
geometric cutoffs, leave only $(r,n)=(6,30)$ below these endpoints.
\end{proof}

For longer tails, the collision budget forces the boundary scalar to be injective, contradicting capacity.

\begin{lemma}\label{lem:large-tail}

No lattice tiling has $r-2\le r-c\le r-1$ if $r\ge11$ and $2n\ge r^2$.
For the smaller radii, a lattice tiling is impossible whenever
\[
\begin{array}{c|rrrrrrrr}
 r&4&6&6&7&8&8&9&10\\\hline
 r-c&3&4&5&5&6&7&7&8,9\\
 n\ge n_0&21&45&28&47&50&37&53&56.
\end{array}
\]
\end{lemma}
\begin{proof}
By Lemma~\ref{lem:positive-scalar-comparison}, $|A[q+1]|>1$, so it
suffices to contradict its collision bound. For $r\ge11$, monotonicity gives
\[
 e_{r,r-c}(2n)\ge e_{r,r-2}(r^2)
 =\frac{(r-1)^2(r^2-r+1)}{r(r+1)(r^2+1)}
 \qquad(r-c\ge r-2).
\]
Moreover, $r^2/(2n)\le1$ and $(r-1)/(2n)\le1/12$, since
$r^2-12(r-1)=(r-11)(r-1)+1>0$.
Furthermore,
\[
 \frac{4nr(r-1)}{(2n+1)(2n+1-r)}
 <\frac{2r^2}{2n-r+1}.
\]
Consequently
\begin{align*}
 \omega_r(2n)
 &<\frac{r^2}{2n}\left(1+\frac{r-1}{n}
     +\frac{2(r-1)^2}{(2n)^2-(r-1)^2}
     +\frac{4n}{2n-r+1}\right)\\
 &\le1+\frac16+\frac2{143}+\frac{24}{11}
 =\frac{2885}{858}<\frac{27}{8}.
\end{align*}
The last difference is $27/8-2885/858=43/3432>0$.
On the other hand, expanding in powers of $r-11\ge0$,
\begin{align*}
 &8r(r+1)(r^2+1)
 \left[(r-2)(2e_{r,r-2}(r^2)-1)-\frac{27}{8}\right]\\
 &\quad=8r^5-99r^4+141r^3-195r^2+101r-32\\
 &\quad=8(r-11)^5+341(r-11)^4+5465(r-11)^3+39064(r-11)^2+105558(r-11)+4104>0.
\end{align*}
Thus $2e_{r,r-2}(r^2)-1>0$ and
$(r-c)(2e_{r,r-c}(2n)-1)>27/8>\omega_r(2n)$, the required contradiction.

For the smaller radii, one endpoint suffices. In the table below,
the last column is $r-c$ times the indicated collision term,
less $\omega_r(2n_0)$. Direct substitution gives
\[
\begin{array}{c|r|c|c}
 (r,r-c)&n_0&\text{collision term}&\text{margin}\\\hline
 (4,3)&21&2e_{r,r-c}(2n_0)-1&535/246519\\
 (6,4)&45&e_{r,r-c}(2n_0)-1/3&1268/101745\\
 (6,5)&28&2e_{r,r-c}(2n_0)-1&1404367/23170728\\
 (7,5)&47&e_{r,r-c}(2n_0)-1/3&295991/25182600\\
 (8,6)&50&e_{r,r-c}(2n_0)-1/3&5231948/209385625\\
 (8,7)&37&2e_{r,r-c}(2n_0)-1&69088739/185739075\\
 (9,7)&53&2e_{r,r-c}(2n_0)-1&120553113/2798241530\\
 (10,8)&56&2e_{r,r-c}(2n_0)-1&192031977/2208244192.
\end{array}
\]
Every margin is positive. Since $e_{r,r-c}(2n)$ increases and $\omega_r(2n)$ decreases,
it remains positive for all $n\ge n_0$. At $r=10$, the positive
margin also gives $2e_{10,8}(112)-1>0$; increasing $r-c$ increases
both $e_{10,10-c}$ and $(r-c)(2e_{10,10-c}-1)$. Thus the last row covers
$r-c=9$ as well.
\end{proof}

Only finitely many dimensions lie below the preceding scalar thresholds.
The cases with $r=4,6$ are also covered by
\cite[Theorem~5.1]{XZ}; the following scalar-capacity argument treats
all the listed parameters directly.

\begin{lemma}\label{lem:positive-finite-capacity}
No lattice tiling exists at the following parameters:
\[
\begin{array}{c|l}
 r&n\\\hline
 4&13,18\\
 6&25,29,30,36,43\\
 7&33,37,41,45\\
 8&39,46\\
 9&51.
\end{array}
\]
\end{lemma}
\begin{proof}
Every listed pair satisfies $2n\ge r^2$ and $r-c\ge2$.
Thus Lemma~\ref{lem:positive-scalar-comparison} requires
$|A[q+1]|>1$.
At $(r,n)=(4,18),(6,25)$, one has $q+1=8$; since $V_r$ is odd,
multiplication by $q+1$ is injective, a contradiction.
For the other pairs, the volume formula gives
\[
\begin{array}{c|r|r|r}
 (r,n)&q+1&m&V_r\bmod m\\\hline
 (4,13)&6&9&6\\
 (6,29)&9&3&1\\
 (6,30)&9&3&2\\
 (6,36)&11&11&3\\
 (6,43)&13&13&2\\
 (7,33)&9&3&1\\
 (7,37)&10&25&20\\
 (7,41)&11&11&8\\
 (7,45)&12&3&1\\
 (8,39)&9&3&2\\
 (8,46)&11&11&2\\
 (9,51)&11&11&7.
\end{array}
\]
Each residue is obtained from
$V_r\equiv\sum_{j=0}^{r}2^j\binom nj\binom rj\pmod m$.
Except at $(r,n)=(4,13),(7,37)$, the displayed residue and the
oddness of $V_r$ give $\gcd(q+1,V_r)=1$, hence $|A[q+1]|=1$,
again a contradiction.
At $(4,13)$, one has $c=1$ and $V_4(13)\equiv6\pmod9$.
Thus $v_3(V_4(13))=1$ and $|A[6]|\le3$, whereas
\[
 e_{4,3}(26)-\frac{3(4)_1}{27}
 =e_{4,3}(26)-\frac{12}{27}
 =\frac{278}{1755}>0.
\]
At $(7,37)$, one has $c=2$ and $V_7(37)\equiv20\pmod{25}$.
Thus $v_5(V_7(37))=1$ and $|A[10]|\le5$, whereas
\[
 e_{7,5}(74)-\frac{5(7)_2}{75^2}
 =e_{7,5}(74)-\frac{210}{75^2}
 =\frac{16049}{27750}>0.
\]
Both inequalities contradict
$e_{r,r-c}(2n)\le |A[q+1]|(r)_c/(2n+1)^c$.
\end{proof}

We now combine the tail criteria and their congruence restrictions to exhaust the positive remainders.

\begin{proposition}\label{prop:positive-complete}

Let $r\ge3$ and let $n$ satisfy~\eqref{eq:geometric-range}.
If $1\le c\le r-1$, then $B_n(r)$ admits no lattice tiling
of $\Z^n$.
\end{proposition}
\begin{proof}
Lemma~\ref{lem:unit-tail} excludes $r-c=1$. If $r$ is odd, both $2n$ and
$r+1$ are even, so $r-c$ is odd. For $r=3,4,5,6$, the remaining
possibilities are therefore
\[
\begin{array}{c|c|c}
 r&r-c&\text{applicable comparison}\\\hline
 3&\varnothing&\\
 4&2,3&\text{short even tail; collision}\\
 5&3&\text{odd occupancy}\\
 6&2,3,4,5&\text{short even tail; odd occupancy; collision}.
\end{array}
\]
These comparisons are Lemmas~\ref{lem:small-tail-two}, \ref{lem:unit-tail}, and \ref{lem:large-tail}, respectively. Their exceptional dimensions are included in Lemma~\ref{lem:positive-finite-capacity}.

Suppose now $r\ge7$. The geometric bound gives
$2n\ge r(r+1)+6>r^2+r+4$, so the first row of
Lemma~\ref{lem:tail-two} excludes $3\le r-c\le r-4$,
including $r=7,8$.
For $r-c=2$ or $r-c=r-3$, odd $r$ is impossible. If $r$ is even,
write $2n=(r+1)q+c$. For $r-c=2$, $q$ is even; the inequality
$2n>r(r+1)$ therefore gives $q\ge r$ and $2n\ge r^2+2r-2$. For $r-c=r-3$, $q$ is odd, hence $q\ge r+1$
and $2n\ge r^2+2r+4$. The second moment row applies for $r\ge10$.
For $r=8$, the third row applies once $n\ge41$; the congruences
$n\equiv3,6\pmod9$ leave only $(n,r-c)=(39,2)$ below that threshold.
This point is in the finite-capacity lemma.

Only $r-c\ge r-2$ remains. Lemma~\ref{lem:large-tail} covers $r\ge11$;
for $r=10$, the geometric bound gives $n\ge58>56$, so its last row applies.
For the smaller radii, the remainder congruences and the thresholds
in that lemma leave exactly
\[
\begin{array}{c|c|c|l}
 (r,r-c)&\text{congruence}&n_0&n<n_0\\\hline
 (4,3)&3\pmod5&21&13,18\\
 (6,4)&1\pmod7&45&29,36,43\\
 (6,5)&4\pmod7&28&25\\
 (7,5)&1\pmod4&47&33,37,41,45\\
 (8,6)&1\pmod9&50&46\\
 (8,7)&5\pmod9&37&\varnothing\\
 (9,7)&1\pmod5&53&51.
\end{array}
\]
Every listed point is in Lemma~\ref{lem:positive-finite-capacity}.
Together with $(6,30)$ from the short-even comparison, this exhausts
all positive remainders.
\end{proof}

\subsection{Zero remainders and scalar images}
When the remainder is zero, the repeated-coordinate incidence involves a
second scalar. We first prove injectivity of the boundary scalar; the
doubling deficit then treats an annihilated second image, while the
character bounds control a nonzero image.
Throughout this subsection, $c=0$, so $2n=(r+1)q$.
\begin{lemma}
\label{lem:zero-kernel}
Suppose that $B_n(r)$ is a transversal for a finite abelian quotient
$A$ and $2n=(r+1)q$.
If $r\ge4$ and $n\ge\max\{r^2/3,4(r-1)\}$,
then multiplication by $q+1$ is injective on $A$.
For $r=3$, the same conclusion holds whenever $n\ge8$ is even.
\end{lemma}
\begin{proof}
Let $E$ be the number of minimum
boundary fibers. In the zero-remainder branch their tail length is $r$.
The minimum-fiber inequality and the ball recurrence give
$E\ge\sigma_r-\frac{r-1}{r+1}\sigma_{r-1},$
$\frac{V_{r-1}}{V_r}\le\frac{r}{2n+1}$ and
$
 \frac{E}{V_r}
 \ge1-\frac{V_{r-1}}{V_r}-
                 \frac{r-1}{r+1}\frac{\sigma_{r-1}}{V_r}
 \ge1-\frac{2r}{r+1}\frac{V_{r-1}}{V_r}
 \ge1-\frac{2r^2}{(r+1)(2n+1)}.
$
Suppose $A[q+1]\ne\{0\}$. Lemma~\ref{lem:all-minimum-kernel}, including
its zero-remainder endpoint, implies
\[
 r\left(2\frac{E}{V_r}-1\right)
 \le\frac{\mathcal C}{V_r}+
              \frac{4n}{V_r}\sum_{j=1}^{r}V_{r-j-1}.
\]
For $n\ge r^2/3$, the layered cancellation estimate gives
\[
 \frac{\mathcal C}{V_r}\le\frac{r^2}{2n}
 \left(1+\frac{2(r-1)}{2n}+
                \frac{2(r-1)^2}{(2n)^2-(r-1)^2}\right).
\]
To estimate the remaining sum, apply the ball recurrence twice:
\[
 \frac{V_{r-2}}{V_r}\le\frac{r(r-1)}{(2n+1)^2},\qquad
 V_{j-1}\le\frac{j}{2n+1}V_j\le\frac{r}{2n+1}V_j.
\]
Since $r<2n+1$, the geometric series consequently gives
\begin{align*}
 \frac{4n}{V_r}\sum_{j=1}^{r}V_{r-j-1}
 &=\frac{4n}{V_r}\sum_{j=0}^{r-2}V_j 
 \le\frac{4n}{1-r/(2n+1)}\frac{V_{r-2}}{V_r}
 \le\frac{4nr(r-1)}{(2n+1)(2n+1-r)}.
\end{align*}
It follows that $\Delta_r(n)\le0$, where
\begin{align*}
 \Delta_r(n)
 &=r-\frac{4r^3}{(r+1)(2n+1)}-\frac{r^2}{2n}
 \left(1+\frac{2(r-1)}{2n}+
              \frac{2(r-1)^2}{(2n)^2-(r-1)^2}\right)
 -\frac{4nr(r-1)}{(2n+1)(2n+1-r)}.
\end{align*}
We now show the reverse inequality without a parameter enumeration.
For $2n>r-1$, each of $\frac{1}{2n+1}$, $\frac{1}{2n}$, $\frac{1}{(2n)^2}$, and $\frac{1}{2n\{(2n)^2-(r-1)^2\}}$
is strictly decreasing. The last denominator has derivative
$24n^2-2(r-1)^2>0$. Moreover
\[
 \frac{d}{dn}\frac {2n}{(2n+1)(2n+1-r)}
 =\frac{2(1-r-4n^2)}{(2n+1)^2(2n+1-r)^2}<0.
\]
Every nonconstant subtracted term in $\Delta_r(n)$ is therefore
decreasing; hence $\Delta_r$ is strictly increasing for $2n>r-1$.
If $r\ge4$, then $2n\ge8(r-1)$, and direct common-denominator
simplification yields
\[
 \Delta_r(4(r-1))
 =\frac{r(872r^3-1763r^2-9259r+9504)}{2016(r-1)(r+1)(8r-7)}.
\]
For $r\ge4$, its numerator is exactly
\[
 872(r-4)^4+12189(r-4)^3+53297(r-4)^2+74040(r-4)+272>0.
\]
Its denominator is positive. Thus
$\Delta_r(n)\ge\Delta_r(4(r-1))>0$, contrary to the necessary
inequality. This proves injectivity for $r\ge4$.

For $r=3$ and $n\ge9$, the same argument applies and
\[
 \Delta_3(9)
 =3-\frac{108}{76}
   -\frac{1}{2}\left(1+\frac{2}{9}+\frac{1}{40}\right)
   -\frac{216}{304}
 =\frac{3349}{13680}>0.
\]
Finally, when $(r,n)=(3,8)$, one has $q+1=5$ and
\[
 |A|=V_3(8)
 =1+2\binom{8}{1}\binom{3}{1}
     +4\binom{8}{2}\binom{3}{2}+8\binom{8}{3}\binom{3}{3}
 =1+48+336+448=833.
\]
Since $833=5\cdot166+3$ and $\gcd(5,833)=1$, multiplication
by five is an automorphism of $A$.
\end{proof}
We next count heavy coordinate incidences in the zero-remainder branch. Put
$J=|\pi(S_{r-1})\setminus2\pi(B_{r-1})|$.
We separate the heavy coordinate incidences according to whether their
boundary fibers are minimum.

\begin{lemma}\label{lem:all-pairs}

If $r\ge3$ and $2n=(r+1)q$, then
\begin{align*}
 \sum_{\mathbf b\in S_r}P_{\mathbf b}=W_{r-1}&=n\bigl(\sigma_{r-1}(n)+\sigma_{r-1}(n-1)\bigr)
   =(2n-r+1)\sigma_{r-1}+(r-2)\sigma_{r-2},\\
 W_{r-1}&\le \frac{(2r+1)(r-1)}{r+1}\sigma_{r-1}+
 [0]S_1^{(2q+1)}S_{r-1}^{(q+1)}B_{r-1}.
\end{align*}
If $(2q+1) A=\{0\}$, then also
$
 (2n-r+1)J\le \frac{(2r+1)(r-1)}{r+1}\sigma_{r-1}.
$
\end{lemma}
\begin{proof}
The heavy coordinate incidences are the pairs $(\mathbf x,\mathbf a)$ with
$\mathbf x\in S_{r+1}$, $\mathbf a\in S_1$, and $x_{\mathbf a}\ge2$.
The map
$
 (\mathbf x,\mathbf a)\longmapsto
 (\mathbf z,\mathbf a)=(\mathbf x-2\mathbf a,\mathbf a)
$
is a bijection onto the pairs with $\mathbf z\in S_{r-1}$ for which
the coordinate of $\mathbf z$ in direction $\mathbf a$ is
nonnegative. A zero coordinate allows two signs, and a nonzero
coordinate allows one. Hence
\begin{align*}
 \sum_{\mathbf b\in S_r}P_{\mathbf b}
 &=\sum_{\mathbf z\in S_{r-1}}(2n-s(\mathbf z))\\
 &=2n\sigma_{r-1}-\bigl((r-1)\sigma_{r-1}-(r-2)\sigma_{r-2}\bigr)\\
 &=(2n-r+1)\sigma_{r-1}+(r-2)\sigma_{r-2}=W_{r-1}.
\end{align*}
For each fixed coordinate the two signs contribute
$\sigma_{r-1}(n)+\sigma_{r-1}(n-1)$, because the zero-coordinate vectors
are counted twice and the other vectors once. Summing over the
$n$ coordinates proves the desired result.

The exact boundary partition gives
\[
 f_{\mathbf b}\ge q,\qquad
 D_{\mathbf b}=s(\mathbf b)+(r+1)(f_{\mathbf b}-q),\qquad
 \sum_{\mathbf b\in S_r}(f_{\mathbf b}-q)
 =\frac{r-1}{r+1}\sigma_{r-1}.
\]
Each point
$\mathbf x$ contributes at most $r+1-s(\mathbf x)$ pairs, because
each coordinate of magnitude at least two contributes at least one
to that difference. The pairs in a nonminimum fiber are thus at most
\[
 \sum_{\mathbf x\in F_{\mathbf b}}(r+1-s(\mathbf x))
 =D_{\mathbf b}
 \le r+(r+1)(f_{\mathbf b}-q)
 \le(2r+1)(f_{\mathbf b}-q).
\]
Their total is at most
$(2r+1)\sum_{\mathbf b}(f_{\mathbf b}-q)=\frac{(2r+1)(r-1)}{r+1}\sigma_{r-1}$.
For a pair in a minimum fiber, put $\mathbf x=2\mathbf a+\mathbf z$.
The vector
$\boldsymbol\delta_{\mathbf b}
=\sum_{\mathbf y\in F_{\mathbf b}}(\mathbf y-
\boldsymbol\varepsilon(\mathbf y))$
is a sum of exactly $D_{\mathbf b}=s(\mathbf b)$ signed units.
The summand belonging to $\mathbf x$ contains $\mathbf a$.
Consequently $\boldsymbol\delta_{\mathbf b}=\mathbf a+\boldsymbol\eta$
with $\|\boldsymbol\eta\|_1\le s(\mathbf b)-1$.
The exact scalar identity gives
\begin{align*}
 (2q+1)\pi(\mathbf a)+(q+1)\pi(\mathbf z)
 &=(q+1)\pi(2\mathbf a+\mathbf z)-\pi(\mathbf a)
 =(q+1)\pi(\mathbf b)-\pi(\mathbf a)
=\pi\bigl(\mathbf b-\boldsymbol\varepsilon(\mathbf b)
                         +\boldsymbol\eta\bigr),
\end{align*}
and the vector in the final expression has norm at most
$r-s(\mathbf b)+s(\mathbf b)-1=r-1$.
The map $(\mathbf x,\mathbf a)\mapsto(\mathbf a,\mathbf z)$ is
injective. Since $B_{r-1}=-B_{r-1}$, each such pair contributes to
$[0]S_1^{(2q+1)}S_{r-1}^{(q+1)}B_{r-1}$.
Adding the nonminimum contribution proves
the required result. 

Finally suppose $(2q+1) A=\{0\}$. Since $2(q+1)-1=2q+1$, multiplication
by $q+1$ is multiplication by $2^{-1}$ on $A$.
For each $\mathbf z\in S_{r-1}$ with $\pi(\mathbf z)\notin2\pi(B_{r-1})$,
the above bijection gives $2n-s(\mathbf z)\ge2n-r+1$ pairs.
None belongs to a minimum fiber: otherwise
$
 2^{-1}\pi(\mathbf z)
 =(2q+1)\pi(\mathbf a)+(q+1)\pi(\mathbf z)\in\pi(B_{r-1}).
$
Their number is therefore at least $(2n-r+1)J$ and at most
$\frac{(2r+1)(r-1)}{r+1}\sigma_{r-1}$, proving $
 (2n-r+1)J\le \frac{(2r+1)(r-1)}{r+1}\sigma_{r-1}
$.
\end{proof}

The two lens constants now enter the same root-graph calculation.
The resulting deficit bounds every annihilated scalar image and leads
to a coordinate-order obstruction. With $J$ as defined above, put
\begin{gather*}
 a(n,r)=4n^2-24n(r-1)+14n-9(r-1)+3,\\
 b(n,r)=4n^2+4n-2n(r-1)+(r-1)^2-3(r-1)+2,\\
 \delta(n,r)=\frac{a(n,r)}{6b(n,r)}.
\end{gather*}

\begin{lemma}\label{lem:rootgraph}
If $r\ge3$ and $n\ge16(r-1)$, then
\begin{equation}\label{eq:rootgraph-sixteenth}
 J\ge\frac{1}{16}\sigma_{r-1}.
\end{equation}
If $r\ge5$ and $n\ge6(r-1)$, then
\begin{equation}\label{eq:rootgraph-delta}
 \frac{J}{\sigma_{r-1}}\ge\delta(n,r)>0.
\end{equation}
\end{lemma}
\begin{proof}
Define
\[
 R_0=\{\pi(\pm\mathbf e_i)\}
 \cup\{\pi(\epsilon\mathbf e_i+\eta\mathbf e_j):
          i<j,\ \epsilon,\eta\in\{\pm1\}\}.
\]
The listed vectors have distinct nonzero images, because their
differences have norm at most four. Thus
$|R_0|=2n+4\binom{n}{2}=2n^2$.
The short-shift formula~\eqref{eq:short-shift-autocorrelation} gives $C_{\pi(B_{r-1})}(g)=V_{r-2}+V_{r-2}(n-1)$ $(g\in R_0)$.
Choose $\kappa=3/2$ when $r\ge3$, or $\kappa=4/3$ when $r\ge5$.
For one of the listed vectors $\mathbf u$, put $g=\pi(\mathbf u)$.
If $\pi(\mathbf v)=g/2$, then $2\mathbf v-\mathbf u\in L\setminus\{\mathbf{0}\}$, since $\mathbf u$
has an odd coordinate. Therefore
\[
 2r+1\le\|2\mathbf v-\mathbf u\|_1
 \le2\|\mathbf v\|_1+2,
 \qquad \|\mathbf v\|_1\ge r.
\]
The corresponding lens inequality, summed over lifts, gives
\begin{align*}
 C_{\pi(B_{r-1})}(g/2)&\le\kappa\sum_{\pi(\mathbf v)=g/2}
                    C_{r-2,r}(\mathbf v)\\
 &=\kappa\sum_{\mathbf x\in B_{r-2}}
       |\{\mathbf y\in B_r:\pi(\mathbf x-\mathbf y)=g/2\}|
 =\kappa V_{r-2}.
\end{align*}
The last equality uses the transversal property of $B_r$.
Multiplication by two is bijective, so
\begin{equation}\label{eq:root-correlation-loss}
 C_{\pi(B_{r-1})}(g)-C_{2\pi(B_{r-1})}(g)=C_{\pi(B_{r-1})}(g)-C_{\pi(B_{r-1})}(g/2)
 \ge V_{r-2}(n-1)-(\kappa-1)V_{r-2}.
\end{equation}

Summing~\eqref{eq:root-correlation-loss}, and then deleting the
vertices in $\pi(B_{r-1})\setminus2\pi(B_{r-1})$, gives
\begin{align*}
 &2n^2\bigl(V_{r-2}(n-1)-(\kappa-1)V_{r-2}\bigr)\\
 &\le\sum_{g\in R_0}\bigl(C_{\pi(B_{r-1})}(g)-C_{2\pi(B_{r-1})}(g)\bigr)\\
 &\le\sum_{g\in R_0}\bigl(C_{\pi(B_{r-1})}(g)
                 -C_{\pi(B_{r-1})\cap2\pi(B_{r-1})}(g)\bigr)\\
 &\le2\sum_{x\in\pi(B_{r-1})\setminus2\pi(B_{r-1})}
                   \deg_{\pi(B_{r-1})}(x).
\end{align*}
The induced root graph is the ordinary graph on $B_{r-1}$:
an additional edge would give a nonzero lattice vector of norm at
most $2(r-1)+2=2r$. For $\mathbf x\in B_{r-1}$, put $s=s(\mathbf x)$. The
two-coordinate moves that do not increase its norm consist of
$\binom{s}{2}$ inward--inward moves and $s(2n-s-1)$ inward--outward
moves. On $S_{r-1}$ there are $s$ inward unit moves; on $S_{r-2}$
all $2n$ unit moves are allowed. Hence
\begin{align*}
 \deg_{\pi(B_{r-1})}(\pi(\mathbf x))
 &=s+\binom{s}{2}+s(2n-s-1)
   =2ns-\frac{s(s+1)}{2}\le2n(r-1)
       &&(\mathbf x\in S_{r-1}),\\
 \deg_{\pi(B_{r-1})}(\pi(\mathbf x))
 &=2n+\binom{s}{2}+s(2n-s-1)
   =2n(s+1)-\frac{s(s+3)}{2}\le2n(r-1)
       &&(\mathbf x\in S_{r-2}).
\end{align*}
Every point of $B_{r-3}$ has degree $2n^2$, so
\begin{align*}
 &2n^2\bigl(V_{r-2}(n-1)-(\kappa-1)V_{r-2}\bigr)\\
 &\le4n(r-1)|\pi(B_{r-1})\setminus2\pi(B_{r-1})|+4n^2V_{r-3},\\
 &|\pi(B_{r-1})\setminus2\pi(B_{r-1})|\\
 &\ge\frac{n\bigl(V_{r-2}(n-1)-(\kappa-1)V_{r-2}\bigr)}{2(r-1)}
       -\frac{n}{r-1}V_{r-3}.
\end{align*}
Coordinate symmetry and $s(\mathbf x)\le r-2$ on $B_{r-2}$ give
\begin{gather*}
 n\bigl(V_{r-2}-V_{r-2}(n-1)\bigr)
 =\sum_{\mathbf x\in B_{r-2}}s(\mathbf x)\le(r-2)V_{r-2},\\
 V_{r-2}(n-1)-(\kappa-1)V_{r-2}
 \ge\left(2-\kappa-\frac{r-2}{n}\right)V_{r-2}.
\end{gather*}
At most $V_{r-2}$ points of $\pi(B_{r-1})\setminus2\pi(B_{r-1})$
lie below $S_{r-1}$, so
\begin{align*}
 J&\ge|\pi(B_{r-1})\setminus2\pi(B_{r-1})|-V_{r-2}\\
 &\ge\frac{1}{r-1}\left(\frac{(2-\kappa)n-3(r-1)+1}{2}V_{r-2}
                         -nV_{r-3}\right).
\end{align*}
The ball recurrence gives
\[
 \frac{V_{r-3}}{V_{r-2}}\le\frac{r-2}{2n+1},\qquad
 \sigma_{r-1}=\frac{(2n-r+2)V_{r-2}+(r-2)V_{r-3}}{r-1}.
\]
Thus, whenever the numerator at the upper bound for $V_{r-3}/V_{r-2}$ is positive,
\begin{align*}
 \frac{J}{\sigma_{r-1}}&\ge\frac{((2-\kappa)n-3(r-1)+1)V_{r-2}/2-nV_{r-3}}{(2n-r+2)V_{r-2}+(r-2)V_{r-3}}\\
 &\ge\frac{2(2-\kappa)n^2-8n(r-1)+(6-\kappa)n-3(r-1)+1}{2b(n,r)}.
\end{align*}
After division by $V_{r-2}$, the numerator decreases and the denominator
increases with $V_{r-3}/V_{r-2}$;
substitution of $V_{r-3}/V_{r-2}=(r-2)/(2n+1)$ gives the last fraction.
For $\kappa=4/3$ its numerator and denominator, multiplied by three,
are $a(n,r)$ and $6b(n,r)$. Moreover,
\[
 a(6(r-1),r)=75(r-1)+3>0,
 \quad \partial_na=8n-24(r-1)+14>0\quad(n\ge6(r-1)).
\]
This proves~\eqref{eq:rootgraph-delta}.
For $\kappa=3/2$ the fraction is $\frac{a_2(n,r)}{4b(n,r)}$, where $a_2=2n^2-16n(r-1)+9n-6(r-1)+2$.
At $n=16(r-1)$, $a_2=256(r-1)^2+138(r-1)+2>0$, and
$\partial_na_2=4n-16(r-1)+9>0$ for $n\ge16(r-1)$.
Furthermore, $n\ge16(r-1)$ gives
\begin{align*}
 4a_2-b
&=4\bigl(n-16(r-1)\bigr)^2+(66(r-1)+32)\bigl(n-16(r-1)\bigr)\\
 &\quad+31(r-1)^2+491(r-1)+6>0.
\end{align*}
Consequently $a_2/(4b)>1/16$, proving
\eqref{eq:rootgraph-sixteenth}.
\end{proof}

Suppose $2n=(r+1)q$ and $(2q+1)A=\{0\}$. The vector $(2q+1)\mathbf e_i\in L\setminus\{\mathbf{0}\}$
gives $2q+1\ge2r+1$, hence $q\ge r$.
For $r\ge5$ and $n\ge6(r-1)$, Lemma~\ref{lem:all-pairs} and the bound~\eqref{eq:rootgraph-delta} in Lemma~\ref{lem:rootgraph} imply
\begin{equation}\label{eq:exact-zero-necessary}
 (2n-r+1)\delta(n,r)\le\frac{(2r+1)(r-1)}{r+1}<2(r-1).
\end{equation}

The deficit inequality bounds the scalar multiplier independently of the radius.

\begin{lemma}\label{lem:zero-r22-m21}
If $r\ge5$, $2n=(r+1)q$, and $(2q+1)A=\{0\}$, then
$r\le q\le21$.
\end{lemma}
\begin{proof}
Suppose $q\ge22$. Since $r-1\ge4$ and $r-1\le q-1$, $n=\frac{q(r+1)}{2}\ge11(r+1)>6(r-1)$.
The polynomials $a,b$ defined before Lemma~\ref{lem:rootgraph} satisfy
\begin{align*}
 \mathscr G(r,q)&=(2n-r+1)a(n,r)-12(r-1)\,b(n,r)
 =c_3(r-1)^3+c_2(r-1)^2+c_1(r-1)+c_0,
\end{align*}
where substitution of $2n=q(r+1)$ and multiplication give
\begin{align*}
 c_3&=q^3-25q^2+24q-12,\\
 c_2&=6q^3-93q^2+8q+45
     =3q^2(2q-31)+8q+45>0,\\
 c_1&=12q^3-72q^2-77q-27,\\
 c_0&=8q^3+28q^2+6q>0.
\end{align*}
Since $q\ge22$, the two endpoint combinations expand as
\begin{align*}
 c_2+(q-1)c_3
 &=q^4-20q^3-44q^2-28q+57\\
 &=(q-22)^4+68(q-22)^3+1540(q-22)^2+11588(q-22)-559\ge-559,\\
 c_1-559(q-1)
 &=12q^3-72q^2-636q+532\\
 &=12(q-22)^3+720(q-22)^2+13620(q-22)+79468>0.
\end{align*}
Since $0<r-1\le q-1$,
\[
 c_2+c_3(r-1)=
 \left(1-\frac{r-1}{q-1}\right)c_2
 +\frac{r-1}{q-1}\bigl(c_2+(q-1)c_3\bigr)\ge-559.
\]
Therefore
\begin{align*}
 \mathscr G(r,q)&=(r-1)^2(c_2+c_3(r-1))+c_1(r-1)+c_0\\
 &\ge-559(r-1)^2+c_1(r-1)+c_0\\
 &\ge\{c_1-559(q-1)\}(r-1)+c_0>0.
\end{align*}
As $b(n,r)>0$, this gives
$(2n-r+1)\delta(n,r)>2(r-1)$, contradicting
\eqref{eq:exact-zero-necessary}. Hence $q\le21$.
\end{proof}

The multiplier bound leaves a finite range. Its cases with
$V_r(n)<2^{30}$ are covered by Xiao and Zhou~\cite[Theorem~5.1]{XZ};
we give a complete proof using coordinate orders and volume congruences.

\begin{proposition}\label{thm:zero-image-all-radii}

For every $r\ge4$, the simultaneous conditions
$2n=(r+1)q$ and $(2q+1)A=\{0\}$ are impossible.
\end{proposition}
\begin{proof}
The coordinate-order bound gives $q\ge r$.
If $q=r$, then $n=r(r+1)/2$, contrary to
Proposition~\ref{thm:uniform-triangular-orbits}.
Consequently $q\ge r+1$ for every $r\ge4$.
For $r\ge5$, Lemma~\ref{lem:zero-r22-m21} gives $q\le21$.
Even $r$ requires even $q$, so $r=20,21$ permit no values of $q$;
all remaining radii satisfy $5\le r\le19$.

For $r=4$, if $n\ge48$, then
$
 (2n-r+1)J\le \frac{(2r+1)(r-1)}{r+1}\sigma_{r-1}
$ and
\eqref{eq:rootgraph-sixteenth} imply
\[
 \frac{27}{5}\ge\frac{(2n-3)J}{\sigma_3}
 \ge\frac{2n-3}{16}\ge\frac{93}{16}>\frac{27}{5},
\]
where the last inequality is $465>432$.
Thus $n<48$. Since $q\ge5$ and $n=5q/2$ is an integer,
$q\in\{6,8,10,12,14,16,18\}$.

For all pairs in the resulting finite range, the table gives the least nonnegative residues
\[
 V_r(n)\bmod(2q+1)
 =\left(\sum_{j=0}^r2^j\binom{(r+1)q/2}{j}\binom rj\right)
       \bmod(2q+1).
\]
Here $n=(r+1)q/2$, and $a\bmod m$ denotes the unique integer in
$\{0,\ldots,m-1\}$ congruent to $a$ modulo $m$.
The entries are listed in the order of the second column.
The last column gives the maximum of $\gcd(V_r(n),2q+1)$ over the
nonzero residues in that row.
\[
\begin{array}{c|c|l|c}
 r&q&V_r(n)\bmod(2q+1)&\max\gcd(V_r(n),2q+1)\\\hline
4&6,8,\ldots,18&(2,10,6,6,13,30,23)&3\\
5&6,7,\ldots,21&(6,9,9,1,12,2,5,3,18,4,21,26,24,6,27,7)&5\\
6&8,10,\ldots,20&(15,1,0,3,13,9,1)&1\\
7&8,9,\ldots,21&(2,4,18,8,10,12,14,16,18,5,22,24,26,28)&5\\
8&10,12,\ldots,20&(7,4,23,19,24,8)&7\\
9&10,11,\ldots,21&(5,16,15,12,2,20,26,20,29,12,3,25)&5\\
10&12,14,\ldots,20&(0,19,12,15,1)&3\\
11&12,13,\ldots,21&(10,23,25,14,6,20,31,34,22,36)&5\\
12&14,16,18,20&(24,0,8,33)&1\\
13&14,15,\ldots,21&(22,10,9,0,10,21,5,35)&3\\
14&16,18,20&(12,34,14)&3\\
15&16,17,\ldots,21&(8,0,4,3,11,40)&3\\
16&18,20&(28,27)&1\\
17&18,19,20,21&(32,15,28,22)&3\\
18&20&(21)&1\\
19&20,21&(12,28)&1
\end{array}
\]
For any coordinate image $g_i=\pi(\mathbf e_i)$, its order $d_i$
satisfies
\[
 2r+1\le d_i,\qquad d_i\mid V_r(n),\qquad d_i\mid(2q+1),
 \qquad d_i\mid\gcd(V_r(n),2q+1).
\]
Every nonzero table entry has gcd at most seven, whereas $2r+1\ge9$.
These entries are therefore impossible.

The five zero entries are
\[
\begin{aligned}
 (r,q,n,2q+1)={}&(6,12,42,25),(10,12,66,25),(12,16,104,33),\\
              &(13,17,119,35),(15,17,136,35).
\end{aligned}
\]
If $(2q+1)A=\{0\}$, every prime divisor of $V_r(n)$ divides $2q+1$.
For $(r,n)=(6,42)$, the base-$13$ expansions
$42=3+3\cdot13$ and $6=6+0\cdot13$ give
\[
 V_6(42)\equiv V_6(3)V_0(3)
 =1+2\cdot3\cdot6+4\cdot3\cdot15+8\cdot20
 =377=13\cdot29\equiv0\pmod{13}.
\]
This contradicts $13\nmid25$.
For $(r,n)=(10,66)$, the base-three expansions are
\[
 66=(2110)_3,\quad10=(0101)_3,\qquad
 V_{10}(66)\equiv V_1(0)V_0(1)V_1(1)V_0(2)=3\equiv0\pmod3.
\]
For $(r,n)=(13,119)$ they are
\[
 119=(11102)_3,\quad13=(00111)_3,
 \quad V_r(n)\equiv V_{1}(2)V_{1}(0)V_{1}(1)V_{0}(1)V_{0}(1)
 =5\cdot1\cdot3\cdot1\cdot1\equiv0\pmod3.
\]
Both contradict $3\nmid(2q+1)$.

For the remaining two pairs the binomial volume formula gives
\begin{align*}
 V_{12}(104)&\equiv6\pmod9,
 &V_{12}(104)&\equiv88\pmod{121},\\
 V_{15}(136)&\equiv10\pmod{25},
 &V_{15}(136)&\equiv35\pmod{49}.
\end{align*}
Thus $v_3(V_{12}(104))=v_{11}(V_{12}(104))=1$.
Since only $3$ and $11$ may divide that volume, it would equal
$3\cdot11=33$, contradicting $V_{12}(104)\ge1+2\cdot104=209$.
Likewise $v_5(V_{15}(136))=v_7(V_{15}(136))=1$ would force
$V_{15}(136)=35$, contradicting $V_{15}(136)\ge273$.
All cases are excluded.
\end{proof}

For the nonzero scalar image, we combine the character bounds of
Section~\ref{sec:boundary-methods} with the repeated-coordinate count
in Lemma~\ref{lem:all-pairs}. Projection onto functions constant on
the cosets of $(2q+1)A$ separates the two invariant subspaces.
The parameter $\mathcal L$ below may be any upper bound for $\Lambda+1$.
For $m\ge1$, define the comparison functions by
\[
\begin{aligned}
 F_m(n,\mathcal L,e)
 &=\frac{2n}{m}+\left(1-\frac1m\right)\mathcal L
 +e\left[\frac{\mathcal L}2+(2n-\mathcal L)\left(2-\frac{3}{2m}\right)\right],\\
 G(e,z)&=\frac23+\left(\frac23-e\right)z+3e.
\end{aligned}
\]

\begin{lemma}\label{lem:scalar-image-comparison}
Suppose $r\ge3$, $2n=(r+1)q$, multiplication by $q+1$ is
injective on $A$, $1\le\Lambda+1\le\mathcal L<2n$, and
$V_{r-2}/\sigma_{r-1}<1/2$. Then
\[
 2n-\frac{(3r+2)(r-1)}{r+1}
 \le F_{|(2q+1)A|}\left(n,\mathcal L,\frac{V_{r-2}}{\sigma_{r-1}}\right).
\]
If $(2q+1)A=A$, one also has
\begin{equation}\label{eq:full-image-necessary}
 2n-\frac{(3r+2)(r-1)}{r+1}\le
 \frac{2nV_{r-1}}{V_r}+
 \mathcal L\left(\frac{\sigma_{r-1}+V_{r-1}}{2\sigma_{r-1}}
             -\frac{\sigma_{r-1}^2+V_{r-1}^2}{2V_r\sigma_{r-1}}\right).
\end{equation}
If $|(2q+1)A|\ge3$, $V_{r-2}/\sigma_{r-1}\le e<1/2$, and $\mathcal L/n\le z<2$, then
$2-\frac{(3r+2)(r-1)}{(r+1)n}\le G(e,z)$, where $G$ increases in each variable
on $0\le e<1/2$, $0\le z<2$.
\end{lemma}
\begin{proof}
Let $f=\mathbf1_{(q+1)\pi(S_{r-1})}$ and
$g=\mathbf1_{\pi(B_{r-1})}$. Injectivity gives
$\|f\|_2^2=\sigma_{r-1}$ and $\|g\|_2^2=V_{r-1}$. Define the projection by
\[
 (\mathsf P_{(2q+1)A}u)(x)
 =\frac1{|(2q+1)A|}\sum_{a\in(2q+1)A}u(x+a).
\]
Since $2(q+1)=(2q+1)+1$, reduction modulo $(2q+1)A$ gives
$(q+1)y\equiv y/2$, whence
$\mathsf P_{(2q+1)A}f
 =\mathsf P_{(2q+1)A}\mathbf1_{2^{-1}\pi(S_{r-1})}$.
Lemma~\ref{lem:nonzero-autocorrelation} gives, for $a\ne0$,
\begin{align*}
 C_{\pi(B_{r-1})}(a)&\le2V_{r-2},\\
 C_{2^{-1}\pi(S_{r-1})}(a)
 &=C_{\pi(S_{r-1})}(2a)
 \le C_{\pi(B_{r-1})}(2a)\le2V_{r-2}.
\end{align*}
Here $2a\ne0$ because $|A|$ is odd. Therefore
\begin{align*}
 \|\mathsf P_{(2q+1)A}f\|_2^2
 &=\frac1{|(2q+1)A|}\sum_{a\in(2q+1)A}C_{2^{-1}\pi(S_{r-1})}(a)\\
 &\le\frac{\sigma_{r-1}}{|(2q+1)A|}
      +2\left(1-\frac1{|(2q+1)A|}\right)V_{r-2},\\
 \|\mathsf P_{(2q+1)A}g\|_2^2
 &=\frac1{|(2q+1)A|}\sum_{a\in(2q+1)A}C_{\pi(B_{r-1})}(a)\\
 &\le\frac{V_{r-1}}{|(2q+1)A|}
      +2\left(1-\frac1{|(2q+1)A|}\right)V_{r-2}.
\end{align*}
The convolution operator associated with $S_1^{(2q+1)}$ has
eigenvalue $2n$ on $\operatorname{im}\mathsf P_{(2q+1)A}$;
on its orthogonal complement its eigenvalues are
$\lambda_{\chi^{2q+1}}-1$, of absolute value at most $\mathcal L$.
Thus the repeated-pair incidence satisfies
\begin{align*}
 &[0]S_1^{(2q+1)}S_{r-1}^{(q+1)}B_{r-1}\\
 &\quad\le2n\|\mathsf P_{(2q+1)A}f\|_2
             \|\mathsf P_{(2q+1)A}g\|_2\\
 &\qquad+\mathcal L
 \sqrt{\bigl(\sigma_{r-1}-\|\mathsf P_{(2q+1)A}f\|_2^2\bigr)
       \bigl(V_{r-1}-\|\mathsf P_{(2q+1)A}g\|_2^2\bigr)}\\
 &\quad\le\frac{\mathcal L}{2}(\sigma_{r-1}+V_{r-1})
 +\frac{2n-\mathcal L}{2}
   \bigl(\|\mathsf P_{(2q+1)A}f\|_2^2+\|\mathsf P_{(2q+1)A}g\|_2^2\bigr)\\
 &\quad\le\frac{\mathcal L}{2}(\sigma_{r-1}+V_{r-1})
 +\frac{2n-\mathcal L}{2|(2q+1)A|}(\sigma_{r-1}+V_{r-1})\\
 &\qquad+2(2n-\mathcal L)\left(1-\frac1{|(2q+1)A|}\right)V_{r-2}.
\end{align*}
The second inequality uses $2\sqrt{xy}\le x+y$ on both terms;
the final inequality uses $2n-\mathcal L>0$ and the preceding bounds
for the projected squared norms.
Since $V_{r-1}/\sigma_{r-1}=1+V_{r-2}/\sigma_{r-1}$,
division by $\sigma_{r-1}$ gives
\[
 \frac{[0]S_1^{(2q+1)}S_{r-1}^{(q+1)}B_{r-1}}{\sigma_{r-1}}
 \le F_{|(2q+1)A|}\left(n,\mathcal L,\frac{V_{r-2}}{\sigma_{r-1}}\right).
\]
Lemma~\ref{lem:all-pairs}, including its bound for the contribution
from nonminimum fibers, gives
\[
 \frac{[0]S_1^{(2q+1)}S_{r-1}^{(q+1)}B_{r-1}}{\sigma_{r-1}}
 \ge2n-\frac{(3r+2)(r-1)}{r+1},
\]
which proves the first inequality.

When $(2q+1)A=A$, the projection is onto constant functions, so
\[
 \|\mathsf P_{(2q+1)A}f\|_2^2=\frac{\sigma_{r-1}^2}{V_r},\qquad
 \|\mathsf P_{(2q+1)A}g\|_2^2=\frac{V_{r-1}^2}{V_r}.
\]
The first incidence inequality above and $\sqrt{xy}\le(x+y)/2$
applied only to the complementary terms give
\begin{align*}
 \frac{[0]S_1^{(2q+1)}S_{r-1}^{(q+1)}B_{r-1}}{\sigma_{r-1}}
 &\le\frac{2nV_{r-1}}{V_r}\\
 &\quad+\frac{\mathcal L}{2\sigma_{r-1}}
 \left(\sigma_{r-1}-\frac{\sigma_{r-1}^2}{V_r}
       +V_{r-1}-\frac{V_{r-1}^2}{V_r}\right),
\end{align*}
which proves \eqref{eq:full-image-necessary}.
Finally,
\begin{align*}
 F_{|(2q+1)A|}\left(n,\mathcal L,\frac{V_{r-2}}{\sigma_{r-1}}\right)
 &=\mathcal L+\frac{V_{r-2}}{\sigma_{r-1}}
                    \left(4n-\frac{3\mathcal L}{2}\right)\\
 &\quad+\frac{2n-\mathcal L}{|(2q+1)A|}
             \left(1-\frac{3V_{r-2}}{2\sigma_{r-1}}\right).
\end{align*}
Since $2n-\mathcal L>0$ and $1-3V_{r-2}/(2\sigma_{r-1})>0$,
this decreases with $|(2q+1)A|$.
For $|(2q+1)A|\ge3$ it gives
\begin{align*}
 \frac1nF_{|(2q+1)A|}\left(n,\mathcal L,\frac{V_{r-2}}{\sigma_{r-1}}\right)
 &\le\frac1nF_3\left(n,\mathcal L,\frac{V_{r-2}}{\sigma_{r-1}}\right)\\
 &=\frac23+\left(\frac23-\frac{V_{r-2}}{\sigma_{r-1}}\right)
               \frac{\mathcal L}{n}+\frac{3V_{r-2}}{\sigma_{r-1}}.
\end{align*}
The partial derivatives $\partial_z G=2/3-e>0$ and $\partial_e G=3-z>0$ prove both
the inequality and its stated monotonicity.
\end{proof}

For a nonzero image, the spectral comparison imposes a uniform
upper bound on the dimension.

\begin{lemma}\label{lem:nonzero-image-9r}
Suppose $r\ge3$, $2n=(r+1)q$, and $B_n(r)$ is a transversal
for a finite abelian quotient $A$.
If $q+1$ is injective and $(2q+1)A\ne\{0\}$, then $n<9r$.
\end{lemma}
\begin{proof}
Assume $n\ge9r$. Since $|A|$ is odd, $|(2q+1)A|\ge3$.
The volume ratio and Corollary~\ref{thm:uspec} give
\[
 \frac{V_{r-2}}{\sigma_{r-1}}\le\frac{r-1}{2n-r+2}<\frac{r}{2n-r}
 =\frac{1}{2n/r-1}\le\frac{1}{17},
 \qquad \Lambda+1\le 3r+\sqrt{8r(n+r)}.
\]
For fixed $r>0$, both $3r/n$ and $8r/n+8r^2/n^2$ decrease for $n>0$. Consequently
\[
 \frac{3r+\sqrt{8r(n+r)}}{n}=\frac{3r}{n}+\sqrt{\frac{8r}{n}+\frac{8r^2}{n^2}}
 \le\frac{3+\sqrt{80}}{9}<\frac{4}{3}<2.
\]
The normalized scalar-image comparison therefore yields
\[
 2-\frac{(3r+2)(r-1)}{n(r+1)}
 \le G\left(\frac{1}{17},\frac{3+\sqrt{80}}{9}\right)
 =\frac{480+31\sqrt{80}}{459}<\frac{5}{3}.
\]
The last inequality follows from
$285^2-31^2\cdot80=4345>0$, since $285$ and $31\sqrt{80}$ are
positive. On the other hand,
\[
 \frac{(3r+2)(r-1)}{r+1}=3r-\frac{4r+2}{r+1}<3r,
 \qquad 2-\frac{(3r+2)(r-1)}{n(r+1)}>2-\frac{3r}{n}\ge\frac{5}{3},
\]
a contradiction.
\end{proof}

The monotonicity of the scalar-image comparison extends an endpoint
certificate to an entire interval of dimensions. For the smaller radii,
the following endpoints sharpen the bound $n<9r$.
For $4\le r\le16$, define the continuation endpoints
$(n_r^*,\beta_r^*)$ by the following table. Its last two columns
record the indicated quantities, where
\[
 \Delta_r^*=\frac{4n_r^*}{3}-\frac{(3r+2)(r-1)}{r+1}-\frac{2(\beta_r^*+1)}3
 -\frac{r-1}{2n_r^*+2-r}\bigl(3n_r^*-\beta_r^*-1\bigr).
\]
\[
\begin{array}{c|r|r|r|l}
 r&n_r^*&\beta_r^*&(\beta_r^*)^2n_r-\beta_r^2n_r^*&\Delta_r^*\\ \hline
4&20&20&220&677/570\\
5&27&29&786&2/3\\
6&35&39&2128&365/462\\
7&44&49&7853/4&2147/996\\
8&54&59&4959&43/9\\
9&60&68&560&1662/565\\
10&66&75&0&4699/2046\\
11&72&84&0&25/54\\
12&91&102&7436&32335/3354\\
13&91&106&0&1280/399\\
14&105&119&0&997/110\\
15&120&132&0&389/24\\
16&136&145&0&18019/731
\end{array}
\]

\begin{lemma}\label{lem:zero-dimension-continuation}
For $4\le r\le16$, no zero-remainder lattice tiling with
injective multiplication by $q+1$ and $(2q+1)A\ne\{0\}$ exists in dimension
$n\ge n_r^*$.
\end{lemma}
\begin{proof}
The nonnegative fourth column and
Lemma~\ref{lem:explicit-spectral-certificates} give
$
 \Lambda<\beta_r\sqrt{n/n_r}\le \beta_r^*\sqrt{n/n_r^*}
$
for all $n\ge n_r^*$.
Put $e(n)=(r-1)/(2n+2-r)$ and
$z(n)=(\beta_r^*\sqrt{n/n_r^*}+1)/n$. The ball recurrence gives $\frac{V_{r-2}}{\sigma_{r-1}}\le e(n)$. Indeed,
\[
 (r-1)V_{r-1}=(2n+1)V_{r-2}+(r-2)V_{r-3}
 \ge(2n+1)V_{r-2}.
\]
Substituting $V_{r-1}=\sigma_{r-1}+V_{r-2}$ and rearranging gives
$(2n+2-r)V_{r-2}\le(r-1)\sigma_{r-1}$, which is the claimed bound.
Moreover
\[
 e'(n)=-\frac{2(r-1)}{(2n+2-r)^2}<0,\qquad
 z'(n)=-\frac{\beta_r^*}{2\sqrt{n_r^*}n^{3/2}}-\frac1{n^2}<0.
\]
The table gives $e(n_r^*)<1/2$ and $z(n_r^*)=(\beta_r^*+1)/n_r^*<2$.
Since $|(2q+1)A|\ge3$, the normalized comparison would imply
\[
 2-\frac{(3r+2)(r-1)}{(r+1)n}
 \le G(e(n),z(n))\le G(e(n_r^*),z(n_r^*)).
\]
But the last column is exactly
\[
 \Delta_r^*=n_r^*\left\{2-\frac{(3r+2)(r-1)}{(r+1)n_r^*}
                          -G(e(n_r^*),z(n_r^*))\right\}>0.
\]
For $n\ge n_r^*$, the left side of the necessary comparison is at least
$2-\frac{(3r+2)(r-1)}{(r+1)n_r^*}$, a contradiction.
\end{proof}

The remaining initial dimensions require sharper scalar-image bounds and,
for $(r,n)=(5,18)$, a prime-quotient moment calculation.
The first six pairs below are covered by \cite[Theorem~5.1]{XZ};
we include direct proofs for every pair.
\begin{lemma}\label{lem:zero-initial-dimensions}
No zero-remainder lattice tiling with injective multiplication by $q+1$ exists at
\[
\begin{gathered}
 (r,n)=(4,15),(5,18),(5,21),(5,24),(6,28),(7,32),\\
 (7,36),(7,40),(8,45),(9,50),(9,55).
\end{gathered}
\]
\end{lemma}
\begin{proof}
We first exclude $(r,n)=(5,18)$. The volume formula gives
$V_5(18)=590557=407\cdot1451$. The integer $p=1451$ is prime:
$p<39^2$, and division by every prime at most $38$ gives
\[
\begin{array}{c|rrrrrrrrrrrr}
\ell&2&3&5&7&11&13&17&19&23&29&31&37\\\hline
1451\bmod\ell&1&2&1&2&10&8&6&7&2&1&25&8.
\end{array}
\]
The structure theorem for finite abelian groups therefore gives a surjection
$\varphi:A\to\mathbb F_p$. Put $a_i=\varphi(g_i)$ for $1\le i\le18$.
The $a_i$ generate $\mathbb F_p$, and every field element occurs
exactly $407$ times on $B_{18}(5)$.

Work in $\mathbb F_p[t]/(t^{37})$, where
$\exp(bt)=\sum_{j=0}^{36}b^jt^j/j!$ satisfies
$\exp(bt)\exp(dt)=\exp((b+d)t)$ by the binomial theorem.
Set $M_k=\sum_{i=1}^{18}a_i^{2k}$. The ball generating function is
\[
 \mathscr F(z,t)=\frac1{1-z}\prod_{i=1}^{18}
 \frac{1-z^2}{(1-z\exp(a_it))(1-z\exp(-a_it))},
 \qquad \mathscr F(z,0)=\frac{(1+z)^{18}}{(1-z)^{19}}.
\]
Expanding the logarithm through $z$-degree five gives
\[
\begin{aligned}
 \log\frac{\mathscr F(z,t)}{\mathscr F(z,0)}
 &=\sum_{a=1}^{5}\frac{z^a}{a}
       \sum_{i=1}^{18}\bigl(\exp(aa_it)+\exp(-aa_it)-2\bigr)\\
 &=2\sum_{k=1}^{18}\frac{M_kt^{2k}}{(2k)!}
       \sum_{a=1}^{5}a^{2k-1}z^a\pmod{z^6}.
\end{aligned}
\]
All denominators are invertible modulo $p$. Uniformity on the ball gives
\[
 [z^5]\mathscr F(z,t)=407\sum_{b\in\mathbb F_p}\exp(bt),\qquad
 [z^5t^{2k}]\mathscr F(z,t)=0\quad(1\le k\le18).
\]
Indeed, $X^{2k}-1$ has at most $2k\le36<p-1$ roots, so some
$u\in\mathbb F_p^*$ satisfies
$u^{2k}\ne1$; multiplication by $u$ permutes the field and implies
$(u^{2k}-1)\sum_b b^{2k}=0$, hence $\sum_b b^{2k}=0$.
Put $\ell_k=\sum_{a=1}^{5}V_{5-a}(18)a^{2k-1}\in\mathbb F_p$.
Their values, in increasing order of $k$, are
\[
\begin{array}{rrrrrrrrr}
826&796&309&1306&1307&264&1365&914&1234\\
779&485&533&452&94&570&212&953&335.
\end{array}
\]
To verify the table, start with
\[
 (V_4(18),2V_3(18),3V_2(18),4V_1(18),5)
 \equiv(535,985,604,148,5)\pmod{1451};
\]
at each step multiply its $a$th component by $a^2$ modulo $1451$
and sum the five components. Thus every $\ell_k$ is nonzero.
If $M_1=\cdots=M_{k-1}=0$, the logarithm above has no positive
$t$-degree below $2k$, so nonlinear terms in its exponential have
$t$-degree at least $4k>2k$. Consequently
\[
 0=[z^5t^{2k}]\mathscr F(z,t)=\frac{2\ell_kM_k}{(2k)!},\qquad M_k=0.
\]
Induction gives $M_1=\cdots=M_{18}=0$. For the elementary symmetric
functions $e_j$ of $a_1^2,\ldots,a_{18}^2$, with $e_0=1$, Newton's
identities yield
\[
 je_j=\sum_{k=1}^{j}(-1)^{k-1}e_{j-k}M_k=0
 \quad(1\le j\le18<p),\qquad
 \prod_{i=1}^{18}(X-a_i^2)=X^{18}.
\]
Substitution of each $a_i^2$ gives $a_i^{36}=0$, hence every $a_i=0$,
contrary to surjectivity.

For the other ten pairs we use the scalar-image comparison.
For seven of these pairs, the volume formula gives
\begin{align*}
V_4(15)&=39041=13\cdot3003+2,\\
V_5(24)&=2383409=17\cdot140200+9,\\
V_6(28)&=48442297=17\cdot2849546+15,\\
V_7(32)&=993696769=17\cdot58452751+2,\\
V_7(36)&=2229491905=19\cdot117341679+4,\\
V_9(50)&=3074466656021=21\cdot146403174096+5,\\
V_9(55)&=7166445945583=23\cdot311584606329+16.
\end{align*}
Here $13,17,17,17,19,21,23$ are the corresponding values of
$2q+1=4n/(r+1)+1$. Each displayed remainder is coprime to its
respective $2q+1$, so $\gcd(2q+1,V_r)=1$ and $(2q+1)A=A$
for these seven pairs.
For the other three pairs,
\begin{align*}
V_5(21)&=9\cdot138331,&
138331&\equiv1\pmod3,\quad138331\equiv1\pmod5,\\
V_7(40)&=243\cdot18938939,&
18938939&\equiv2\pmod3,\quad18938939\equiv5\pmod7,\\
V_8(45)&=49\cdot2420592481,&
2420592481&\equiv1\pmod3,\quad2420592481\equiv6\pmod7.
\end{align*}
The factors $138331,18938939,2420592481$ are coprime to the
respective scalars $2q+1=15,21,21$. Multiplication by $2q+1$ is
bijective on their primary components, and hence
\[
\begin{aligned}
 |(2q+1)A|&\ge138331>1001&&(r,n)=(5,21),\\
 |(2q+1)A|&\ge18938939>5&&(r,n)=(7,40),\\
 |(2q+1)A|&\ge2420592481>5&&(r,n)=(8,45).
\end{aligned}
\]

The table below defines $\beta_{r,n}$ and gives the spectral bound
$\Lambda<\beta_{r,n}$. Lemma~\ref{lem:positive-hessenberg} gives this bound
because $\beta_{r,n}^2n_r\ge \beta_r^2n$. Equality holds when $n=n_r$; the
four other differences are
\[
\begin{array}{c|rrrr}
(r,n)&(5,24)&(7,36)&(7,40)&(9,55)\\ \hline
 \beta_{r,n}^2n_r-\beta_r^2n&309&2783/4&5251/2&6380.
\end{array}
\]
Put $e=(r-1)/(2n+2-r)$ and $t=r/(2n+1)$.
The volume ratios give $\frac{V_{r-2}}{\sigma_{r-1}}\le e$, $V_{r-1}/\sigma_{r-1}\le1+e$,
and $V_{r-1}/V_r\le t<1$. Since $n\ge3r$, we have $e<1/2$;
the positive fourth column below gives $\beta_{r,n}+1<3n/2$.
For $1\le a\le1+e$ and $0\le v\le t$, define
\[
 \mathscr P(a,v)=2nv+\frac{\beta_{r,n}+1}{2}\{1+a-v(a+a^{-1})\}.
\]
Then
\[
 \partial_a\mathscr P=\frac{\beta_{r,n}+1}{2}\{1-v(1-a^{-2})\}>0,\qquad
 \partial_v\mathscr P=2n-\frac{\beta_{r,n}+1}{2}(a+a^{-1})
 >2n-\frac{3n}{4}\cdot\frac{13}{6}=\frac{3n}{8}>0,
\]
where $a+a^{-1}\le13/6$ follows from $1\le a<3/2$.
Thus, when $(2q+1)A=A$, the right side of
\eqref{eq:full-image-necessary} is
$\mathscr P(V_{r-1}/\sigma_{r-1},V_{r-1}/V_r)\le \mathscr P(1+e,t)$.
For the other pairs, $F_{|(2q+1)A|}$ decreases with $|(2q+1)A|$ by
Lemma~\ref{lem:scalar-image-comparison}, and
\[
 \begin{aligned}
 \frac{\partial F_{|(2q+1)A|}(n,\beta_{r,n}+1,e)}{\partial e}
 &=\frac{\beta_{r,n}+1}{2}\\
 &\quad +(2n-\beta_{r,n}-1)\left(2-\frac{3}{2|(2q+1)A|}\right)>0
 \quad(|(2q+1)A|\ge1).
 \end{aligned}
\]
We may therefore replace $\frac{V_{r-2}}{\sigma_{r-1}}$ by $e$ and $|(2q+1)A|$ by the indicated
lower bound. The last two columns give the resulting upper bound and
the difference between $2n-\frac{(3r+2)(r-1)}{r+1}$ and that bound:
\[
\begin{array}{cc|c|c|c|c}
 r&n&\beta_{r,n}&3n-2(\beta_{r,n}+1)&\text{upper bound}&\text{margin}\\ \hline
4&15&17&9&\mathscr P(1+e,t)&147909/134540\\
5&21&25&11&F_{1001}(n,\beta_{r,n}+1,e)&134/3549\\
5&24&27&16&\mathscr P(1+e,t)&83222/15435\\
6&28&34&14&\mathscr P(1+e,t)&2236447/788424\\
7&32&165/4&23/2&\mathscr P(1+e,t)&1179/38350\\
7&36&44&18&\mathscr P(1+e,t)&7372101/1428172\\
7&40&47&24&F_5(n,\beta_{r,n}+1,e)&1039/500\\
8&45&53&27&F_5(n,\beta_{r,n}+1,e)&221/180\\
9&50&62&24&\mathscr P(1+e,t)&12352644/1581155\\
9&55&66&31&\mathscr P(1+e,t)&9677858/705035.
\end{array}
\]
All margins are positive, contradicting the scalar-image comparison.
\end{proof}

The preceding exclusions now cover the entire zero-remainder branch
at radii at least four. We verify the injectivity hypotheses and the
initial dimensions explicitly.

\begin{proposition}\label{prop:zero-complete}
Let $r\ge4$ and $n>r(r+1)/2+2$. If $c=0$, then $B_n(r)$
admits no lattice tiling of $\Z^n$.
\end{proposition}
\begin{proof}
Suppose that such a lattice tiling exists, and let $A$ be its finite
abelian quotient. Since $c=0$, the division identity gives
$2n=(r+1)q$. This divisibility condition is equivalent to
\begin{equation}\label{eq:zero-progression}
 n\equiv0\pmod{\frac{r+1}{\gcd(2,r+1)}}.
\end{equation}
For $r=4,5,6$, the first dimensions in these progressions satisfying
$n>r(r+1)/2+2$ are $15,18,28$, respectively. They satisfy
\[
 15\ge\max\{16/3,12\},\qquad
 18\ge\max\{25/3,16\},\qquad
 28\ge\max\{12,20\}.
\]
For $r\ge7$,
\[
 n>\frac{r(r+1)}2>\frac{r^2}{3},\qquad
 \frac{r(r+1)}2-4(r-1)
 =\frac{r(r-7)+8}{2}>0.
\]
Thus every dimension under consideration satisfies
$n\ge\max\{r^2/3,4(r-1)\}$, and Lemma~\ref{lem:zero-kernel}
makes multiplication by $q+1$ injective on $A$.
Proposition~\ref{thm:zero-image-all-radii} applies because $r\ge4$
and $2n=(r+1)q$; it excludes $(2q+1)A=\{0\}$.
Consequently $(2q+1)A\ne\{0\}$, so both scalar-image hypotheses
required below hold.

If $r\ge17$, then
$n>r(r+1)/2\ge9r$, contradicting
Lemma~\ref{lem:nonzero-image-9r}.
For $4\le r\le16$, the progressions in
\eqref{eq:zero-progression} give the following complete division.
The column $n_{\min}$ records the first dimension satisfying
$n>r(r+1)/2+2$, and $n_r^*$ is the continuation endpoint in
Lemma~\ref{lem:zero-dimension-continuation}.
\[
\begin{array}{c|r|r|r|l}
 r&\text{step}&n_{\min}&n_r^*&n<n_r^*\\ \hline
4&5&15&20&15\\
5&3&18&27&18,21,24\\
6&7&28&35&28\\
7&4&32&44&32,36,40\\
8&9&45&54&45\\
9&5&50&60&50,55\\
10&11&66&66&\varnothing\\
11&6&72&72&\varnothing\\
12&13&91&91&\varnothing\\
13&7&98&91&\varnothing\\
14&15&120&105&\varnothing\\
15&8&128&120&\varnothing\\
16&17&153&136&\varnothing.
\end{array}
\]
For $n\ge n_r^*$, Lemma~\ref{lem:zero-dimension-continuation}
applies because multiplication by $q+1$ is injective and
$(2q+1)A\ne\{0\}$. Every pair in the last column is excluded
by Lemma~\ref{lem:zero-initial-dimensions}, including $(r,n)=(5,18)$.
These alternatives exhaust every radius $r\ge4$ and every admissible
dimension, a contradiction.
\end{proof}

\subsection{The terminal remainder}
We first exclude a terminal scalar that annihilates the quotient.
The subsequent coefficient estimates show when every terminal scalar
must have this property.
\begin{lemma}\label{lem:annihilated-terminal}
Let $r\ge4$ be even and $2n\ge r^2$. If
$2n+1=(r+1)(q+1)$ and $(q+1)A=\{0\}$, no lattice tiling by $B_n(r)$ exists.
\end{lemma}
\begin{proof}
Fix a boundary point $\mathbf b$ and suppose $f_{\mathbf b}=q+1$.
Then $D_{\mathbf b}=s(\mathbf b)+1$, so
\[
 \pi(\boldsymbol\delta_{\mathbf b})
 =(q+1)\pi(\mathbf b)+\pi(\boldsymbol\varepsilon(\mathbf b))
 =\pi(\boldsymbol\varepsilon(\mathbf b)),
 \qquad\|\boldsymbol\delta_{\mathbf b}\|_1\le s(\mathbf b)+1.
\]
If $s(\mathbf b)<r$, the lattice vector
$\boldsymbol\delta_{\mathbf b}-\boldsymbol\varepsilon(\mathbf b)$ has norm
at most $2s(\mathbf b)+1\le2r-1$ and hence is zero. This is impossible:
$\boldsymbol\delta_{\mathbf b}$ is a sum of $s(\mathbf b)+1$ signed units, whereas
$\boldsymbol\varepsilon(\mathbf b)$ is a sum of $s(\mathbf b)$ signed units,
so their coordinate sums have opposite parities.

If $s(\mathbf b)=r$, then $\boldsymbol\varepsilon(\mathbf b)=\mathbf b$.
The vector $\boldsymbol\delta_{\mathbf b}$ has odd parity and norm at most $r+1$.
If its norm were at most $r$, injectivity on $B_r$ would give
$\boldsymbol\delta_{\mathbf b}=\mathbf b$, again contradicting parity.
Thus $\boldsymbol\delta_{\mathbf b}\in S_{r+1}$ and
$\pi(\boldsymbol\delta_{\mathbf b})=\pi(\mathbf b)$, so
$\boldsymbol\delta_{\mathbf b}\in F_{\mathbf b}$.
Choose $\mathbf u\in\ssupp(\boldsymbol\delta_{\mathbf b})$.
The signed-support partition makes $\boldsymbol\delta_{\mathbf b}$ the unique
member of $F_{\mathbf b}$ containing direction $\mathbf u$.
Its defining sum therefore gives
\begin{align*}
 (\boldsymbol\delta_{\mathbf b})_{\mathbf u}
 &=(\boldsymbol\delta_{\mathbf b})_{\mathbf u}-1
 +\sum_{\mathbf y\in F_{\mathbf b}\setminus\{\boldsymbol\delta_{\mathbf b}\}}
       (\mathbf y-\boldsymbol\varepsilon(\mathbf y))_{\mathbf u}\\
 &\le(\boldsymbol\delta_{\mathbf b})_{\mathbf u}-1,
\end{align*}
a contradiction. Hence every fiber has $f_{\mathbf b}\ne q+1$.
The terminal case of Lemma~\ref{lem:two-shell-defect} now gives
\begin{equation*}
 \frac{\MomentU+2\MomentT}{\sigma_r}
 +2(r+1)(r-1)\frac{\sigma_{r-1}}{\sigma_r}
 +(r-1)(r-2)\frac{\sigma_{r-2}}{\sigma_r}\ge2r(r+1).
\end{equation*}

Lemma~\ref{lem:sign-restricted-heavy-pairs} gives
\begin{gather*}
 \frac{\MomentT}{\sigma_r}\le 2n(2n-1)
 \left(\frac{r!}{(2n)^r}+\frac{r}{2n}\sum_{j=2}^{r-1}\frac{(r)_j}{(2n)^j}\right),\\
 \sum_{j=2}^{r-1}\frac{(r)_j}{(2n)^j}
 \le \frac{r(r-1)}{(2n)^2}\left(1+\frac{r-2}{2n-r}\right).
\end{gather*}
Also $\frac{r!}{(2n)^r}\le\frac{2r^2(r-1)}{(2n)^4}$: for $r=4$ this follows from $2!\le8$;
for $r\ge5$, it follows from
$(r-2)!\le r^{r-2}\le r^{2r-7}\le2r (2n)^{r-4}$.
Consequently, using $2n\ge r^2$,
\begin{align*}
 \frac{\MomentT}{\sigma_r}&\le\frac{r^2(r-1)(2n-1)}{(2n)^2}
       \left(1+\frac{r-2}{2n-r}+\frac{2}{2n}\right)\\
 &\le\frac{r^2}{2n}(r-1)\frac{2n-1}{2n-r}\le r-\frac{1}{r}.
\end{align*}
Here we used
$1+(r-2)/(2n-r)+2/(2n)\le2n/(2n-r)$ and
$(r-1)/(2n-r)\le1/r$.

The bounds $\Phi_k,\Psi_k\le (2n)^2\sigma_k$ and
$\sigma_{r-j}/\sigma_r\le \frac{(r)_j}{(2n)^j}$, together with the support identities, give
\begin{align*}
 \frac{\MomentU}{\sigma_r}
 &\le r-\frac{(r-1)\sigma_{r-1}}{\sigma_r}
       +\frac{r(r-1)(r-2)}{2n}+r^2\\
 &\le r^2+2r-3+\frac{2}{r}-\frac{(r-1)\sigma_{r-1}}{\sigma_r},\\
 &\frac{(r-1)\sigma_{r-1}}{\sigma_r}+(r-1)(r-2)\frac{\sigma_{r-2}}{\sigma_r}
 \le \frac{(r-1)\sigma_{r-1}}{\sigma_r}+\frac{(r-1)^2(r-2)}{r^3}<\frac{(r-1)\sigma_{r-1}}{\sigma_r}+1,\\
 &(2r+1)\frac{(r-1)\sigma_{r-1}}{\sigma_r}
 \le\frac{(2r+1)(r-1)}{r}=2r-1-\frac{1}{r}.
\end{align*}
Substitution into the lower bound yields
$2r^2+2r<r^2+6r-3-1/r$, contrary to
\[
 2r^2+2r-\left(r^2+6r-3-\frac{1}{r}\right)
 =(r-1)(r-3)+\frac{1}{r}>0.
\]
\end{proof}
It remains to force the scalar image to vanish. The next comparison
propagates a coefficient inequality from one dimension to all larger
ones. Fix an even integer $r\ge4$, $n_0\ge r$, and $B>0$.
Put $M=2n_0+B+1$ and define
\[
 H_0=1,\qquad jH_j=BH_{j-1}+M\sum_{k=0}^{j-2}H_k\quad(j\ge1),
\]
\[
 \mathscr U_0=(B+1)H_{r-1}+M\sum_{k=0}^{r-2}H_k,
 \qquad \mathscr L(n)=\sigma_r(n)-(r-1)\sigma_{r-1}(n).
\]
\begin{lemma}\label{lem:terminal-scaling}
Suppose every hypothetical lattice tiling in dimension $n\ge n_0$
satisfies $\Lambda\le B\sqrt{n/n_0}$. If $\mathscr U_0<\mathscr L(n_0)$, then
$2n+1=(r+1)(q+1)$ forces $(q+1)A=\{0\}$ for every $n\ge n_0$.
\end{lemma}
\begin{proof}
Suppose $(q+1)A\ne\{0\}$. For $y\in (q+1)A$, let
\[
 s_j(y)=|\{\mathbf x\in S_j:(q+1)\pi(\mathbf x)=y\}|,
 \qquad F(y)=s_{r+1}(y)-(q+1) s_r(y).
\]
Every exterior point decodes to a boundary point. A minimum boundary
fiber has size $q$ and scalar image zero. Thus every boundary
fiber contributing at $y\ne0$ has size at least $q+1$, and $F(y)\ge0$.
The symmetry $F(-y)=F(y)$ and the shell recurrence give, for any
nonprincipal character $\chi$ of $(q+1)A$,
\begin{align*}
 \widehat F(\chi)=F(0)+\sum_{y\ne0}F(y)\operatorname{Re}\chi(y)
 \le\sum_yF(y)
 =-\frac{\sigma_r-(r-1)\sigma_{r-1}}{r+1}.
\end{align*}
The support formulas give
\begin{align*}
 \mathscr L(n)
 &=\sum_{j=1}^r2^j\binom{n}{j}
       \left(\binom{r-1}{j-1}-(r-1)\binom{r-2}{j-1}\right)\\
 &=2^r\binom{n}{r}-
   \sum_{j=1}^{r-2}(r-j-1)2^j\binom{n}{j}\binom{r-1}{j-1}.
\end{align*}
For $j<r$,
\[
 \frac{\binom nj}{\binom nr}
 =\frac{r!}{j!}\prod_{k=j}^{r-1}\frac{1}{n-k}
 \le\frac{\binom{n_0}j}{\binom{n_0}r}.
\]
Consequently, since $n/n_0\ge1$,
\begin{gather*}
 \frac{\mathscr L(n)}{\binom nr}\ge\frac{\mathscr L(n_0)}{\binom{n_0}r}>0,
 \qquad
 \frac{\binom nr}{\binom{n_0}r}
 =\prod_{k=0}^{r-1}\frac{n-k}{n_0-k}\ge \left(\frac{n}{n_0}\right)^r,\\
 \mathscr L(n)\ge \left(\frac{n}{n_0}\right)^r\mathscr L(n_0).
\end{gather*}

Pull $\chi$ back by multiplication by $q+1$ to a nonprincipal
character $\psi$ of $A$, and put $a_j=\psi(B_n(j))$.
The ball-character recurrence is
\[
 ja_j=\sum_{k=1}^jc_ka_{j-k},\qquad
 c_k=\lambda_{\psi^k}-2n\mathbf1_{2\mid k},\qquad a_r=0.
\]
Since the group order is odd, $\psi^2$ is nonprincipal.
The identities $\widehat F(\chi)=a_{r+1}+(q+1) a_{r-1}$ and
$(r+1)(q+1)=2n+1$ therefore give the exact cancellation
\[
 (r+1)\widehat F(\chi)
 =(\lambda_{\psi^2}+1)a_{r-1}
                       +\sum_{k=3}^{r+1}c_ka_{r+1-k}.
\]
Moreover $|c_1|\le B\sqrt{n/n_0}$ and, for $k\ge2$, $|c_k|\le2n+\Lambda+1\le2n+B\sqrt{n/n_0}+1\le M\frac{n}{n_0}$.
Induction on $j$ now gives $|a_j|\le \left(\frac{n}{n_0}\right)^{j/2}H_j$: it is true at
$j=0$, and the induction step is
\begin{align*}
 j|a_j|&\le B\sqrt{n/n_0}\,\left(\frac{n}{n_0}\right)^{(j-1)/2}H_{j-1}
       +M\frac{n}{n_0}\sum_{k=0}^{j-2}\left(\frac{n}{n_0}\right)^{k/2}H_k\\
 &\le \left(\frac{n}{n_0}\right)^{j/2}\left(BH_{j-1}+M\sum_{k=0}^{j-2}H_k\right)
 =j\left(\frac{n}{n_0}\right)^{j/2}H_j.
\end{align*}
It follows that
\begin{align*}
 (r+1)|\widehat F(\chi)|&\le(B+1)\left(\frac{n}{n_0}\right)^{r/2}H_{r-1}
                     +M\frac{n}{n_0}\sum_{k=0}^{r-2}\left(\frac{n}{n_0}\right)^{k/2}H_k\\
 &\le \left(\frac{n}{n_0}\right)^{r/2}\mathscr U_0<\mathscr L(n).
\end{align*}
This contradicts the preceding negative Fourier bound.
\end{proof}

The required base inequalities follow from positive recurrences
and exact integer comparisons. We use the following parameters:
\[
\begin{array}{c|rrrrrrrr}
 r&4&6&8&10&12&14&16&18\\\hline
 n_0&17&31&49&60&84&112&144&180\\
 B&351/20&35&55&361/5&98&123&184&215.
\end{array}
\]
\begin{lemma}\label{lem:terminal-certificates}
For each displayed choice of $(r,n_0,B)$, the hypotheses of
Lemma~\ref{lem:terminal-scaling} hold.
\end{lemma}
\begin{proof}
For the first four spectral bounds, apply
Lemma~\ref{lem:spectral-ratio-certificate} with the following data.
For each radius, the two lines give $(a_1,\ldots,a_r)$ and
$(\delta_2,\ldots,\delta_r)$, respectively.
\begingroup
\small
\renewcommand{\arraystretch}{1}
\[
\begin{array}{@{}rr@{\quad}l@{}}
 \toprule
 r&D&(a_1,\ldots,a_r)\text{ and }(\delta_2,\ldots,\delta_r)\\
 \midrule
 4&400&(7020,5845,4036,1)\\
 &&(500,5480,3884)\\[2pt]
 6&2&(70,64,57,48,34,2)\\
 &&(32,32,12,0,12)\\[2pt]
 8&1&(55,52,49,45,40,34,24,1)\\
 &&(12,0,13,23,0,10,8)\\[2pt]
 10&20&(1444,1390,1331,1265,1189,1099,986,833,588,5)\\
 &&(1096,1150,829,735,305,662,166,248,132)\\
 \bottomrule
\end{array}
\]
\endgroup
In every case, $a_1=DB$, all $a_k$ are positive, and all $\delta_k$
are nonnegative, as required.
For $r=12,14$, the certificates $(n_r,\beta_r)$ in Lemma~\ref{lem:explicit-spectral-certificates}
are $(78,94)$ and $(105,119)$, respectively. The comparisons
\[
 94^2\cdot84=742224<749112=98^2\cdot78,
 \qquad119^2\cdot112=1586032<1588545=123^2\cdot105
\]
prove $\Lambda<98\sqrt{n/84}$ and
$\Lambda<123\sqrt{n/112}$.
For $r=16,18$, Corollary~\ref{thm:uspec} gives
\begin{align*}
 \Lambda&\le3(r-1)+2\sqrt{(r-1)(2n+2(r-1))}\\
 &\le\sqrt{n/n_0}\left(3(r-1)+2\sqrt{(r-1)(2n_0+2(r-1))}\right).
\end{align*}
The required rational comparisons are $(184-45)^2=19321>19080=4\cdot15(288+30)$,
and $(215-51)^2=26896>26792=4\cdot17(360+34)$.

For the coefficient inequality, write $B=b/D_B$ in lowest terms and put
$\widetilde H_j=j!D_B^jH_j$. The generating function $\sum_{j\ge0}H_jz^j=e^{-(2n_0+1)z}(1-z)^{-(2n_0+B+1)}$
gives the positive two-term recurrence
\[
 \widetilde H_0=1,\quad \widetilde H_1=b,\qquad
 \widetilde H_j=(b+(j-1)D_B)\widetilde H_{j-1}+(j-1)D_B^2(2n_0+1)\widetilde H_{j-2}.
\]
The defining recurrence also gives $\mathscr U_0=rH_r+H_{r-1}$, hence $c_B \mathscr U_0=\widetilde H_r+D_B\widetilde H_{r-1}$, where $c_B=(r-1)!D_B^r$.
The following exact values satisfy
\[
 0<c_B\lceil\mathscr U_0\rceil-\widetilde H_r-D_B\widetilde H_{r-1}<c_B.
\]
Thus the third column is the ceiling of $\mathscr U_0$, and
$\mathscr U_0\le\lceil\mathscr U_0\rceil<\mathscr L(n_0)$ in every row.
\begingroup
\small
\setlength{\arraycolsep}{4pt}
\[
\begin{array}{@{}rrrr@{}}
 \toprule
 r&\mathscr L(n_0)&\lceil\mathscr U_0\rceil&
 \substack{c_B\lceil\mathscr U_0\rceil-\widetilde H_r\\-D_B\widetilde H_{r-1}}\\
 \midrule
4&36380&36039&22859\\
6&41340236&40924734&18\\
8&92016218196&53855608389&4158\\
10&44015676259368&42438483471697&1699680267444\\
12&232935285302796440&91646339777643353&8118891\\
14&1770633918622790495520&168539155358233413385&4402188000\\
16&18314638017576693832691424&5780754631413943435589880&1219967426225\\
18&247610342854834858169341551800&14436124416839419669126707421&219533982832728
\\
 \bottomrule
\end{array}
\]
\endgroup
Each $\mathscr L(n_0)$ is given explicitly by
\[
 \mathscr L(n_0)=\sum_{j=1}^r2^j\binom{n_0}{j}
 \left(\binom{r-1}{j-1}-(r-1)\binom{r-2}{j-1}\right),
\]
so the displayed recurrences and identities involve only integers.
\end{proof}
For the larger even radii, the same coefficient criterion follows from a uniform generating-function estimate.

\begin{proposition}\label{thm:terminal-r20}
Let $r\ge20$ be even. If $2n+1=(r+1)(q+1)$ and $2n>r(r+1)$,
then $B_n(r)$ admits no lattice tiling.
\end{proposition}
\begin{proof}
The integer $q+1$ is odd and
$q+1>r+1/(r+1)$; hence $q+1\ge r+1$ and $n\ge n_0:=\frac{r(r+2)}{2}$ and $2n_0=r(r+2)$.
Put $B=3(r-1)+2\sqrt{(r-1)(2n_0+2(r-1))}$ and $M=2n_0+B+1$.
Corollary~\ref{thm:uspec} gives, for $n\ge n_0$,
\begin{align*}
 \Lambda&\le3(r-1)+2\sqrt{(r-1)(2n+2(r-1))}\\
 &\le\sqrt{n/n_0}\left(3(r-1)+2\sqrt{(r-1)(2n_0+2(r-1))}\right)
 =B\sqrt{n/n_0}.
\end{align*}
We verify the same coefficient criterion used at the smaller radii.
Let $H_j,\mathscr U_0,\mathscr L(n_0)$ be those of
Lemma~\ref{lem:terminal-scaling}, and set $t=1/(4\sqrt r)$. Since
$H_j\ge0$,
\[
 H_jt^j\le\mathscr H(t):=\sum_{k\ge0}H_kt^k
 =\exp\left(Bt+M\sum_{k\ge2}\frac{t^k}{k}\right).
\]
Therefore
\begin{align*}
 \mathscr U_0
 \le\mathscr H(t)\left((B+1)t^{1-r}
                           +M\sum_{j=0}^{r-2}t^{-j}\right)
 \le\frac{\mathscr H(t)}{t^r}
             \left((B+1)t+\frac{Mt^2}{1-t}\right).
\end{align*}
To bound these two factors, use
\[
 B\le3r+2r\sqrt{r+4}\le3r+2r\sqrt r+4\sqrt r,
 \qquad
 \left(\sqrt r+\frac{2}{\sqrt r}\right)^2=r+4+\frac{4}{r}\ge r+4.
\]
As $r\ge20$, we have $t<1/16$, $ Bt\le\frac{r}{2}+\frac{3}{4}\sqrt r+1$ and 
\begin{align*}
 Mt^2\le\frac{r}{16}+\frac{1}{8}+\frac{3}{16}
                  +\frac{\sqrt r}{8}+\frac{1}{4\sqrt r}+\frac{1}{16r}
 <\frac{r}{16}+\frac{\sqrt r}{8}+\frac{2}{5}.
\end{align*}
The geometric sum consequently gives
\begin{align*}
 \log\mathscr H(t)
 &\le Bt+\frac{Mt^2}{2(1-t)}\\
 &<\frac{r}{2}+\frac{3}{4}\sqrt r+1+
       \frac{8}{15}\left(\frac{r}{16}+\frac{\sqrt r}{8}+\frac{2}{5}\right)
 =\frac{8r}{15}+\frac{49\sqrt r}{60}+\frac{91}{75},\\
 (B+1)t+\frac{Mt^2}{1-t}
 &<\frac{17r}{30}+\frac{53\sqrt r}{60}+\frac{1787}{1200}\\
 &<\left(\frac{17}{30}+\frac{53}{240}+\frac{3}{40}\right)r<r.
\end{align*}
Here $\sqrt r\le r/4$ and $1787/1200<3/2\le3r/40$.
The ball generating function satisfies coefficientwise
\[
 \sum_{j\ge0}V_j(n_0)z^j
 =\exp\left((2n_0+1)z+
       \sum_{k\ge2}\frac{2n_0\mathbf1_{2\nmid k}+1}{k}z^k\right)
 \succeq e^{2n_0z}.
\]
Taking the coefficient of $z^r$ and using the factorial estimate
$r!\le e\sqrt r(r/e)^r$ ($r\ge1$) gives
\[
 V_r(n_0)\ge\frac{(2n_0)^r}{r!},\qquad
 \frac{1}{V_r(n_0)t^r}
 \le e\sqrt r\left(\frac{4\sqrt r}{e(r+2)}\right)^r.
\]
Combining these inequalities yields
\[
 \frac{r\mathscr U_0}{V_r(n_0)}<r^{5/2}
 \exp\left(\frac{166}{75}-\frac{7r}{15}+\frac{49\sqrt r}{60}\right)
 \left(\frac{4\sqrt r}{r+2}\right)^r.
\]
The right-hand side decreases on $r\ge20$.
Indeed, $(r+2)^2-16r=r(r-12)+4>0$, and its logarithmic derivative is
\begin{align*}
 &\frac{d}{dr}\left[\frac52\log r+\frac{166}{75}-\frac{7r}{15}
       +\frac{49\sqrt r}{60}+r\log\frac{4\sqrt r}{r+2}\right]\\
 &=\frac{5}{2r}-\frac{7}{15}+\frac{49}{120\sqrt r}
       +\log\frac{4\sqrt r}{r+2}+\frac{1}{2}-\frac{r}{r+2}\\
 &<\frac{1}{8}-\frac{7}{15}+\frac{49}{480}-\frac{1}{6}=-\frac{13}{32}<0.
\end{align*}
At $r=20$, $\sqrt{20}<9/2$ gives
\[
 \frac{166}{75}-\frac{28}{3}+\frac{49\sqrt{20}}{60}
 <\frac{166}{75}-\frac{28}{3}+\frac{147}{40}
 =-\frac{689}{200}<-3.
\]
Moreover
\[
 e^3>\sum_{j=0}^5\frac{3^j}{j!}=\frac{92}{5}>18,
 \qquad
 \left(\frac{11}{9}\right)^5=\frac{161051}{59049}>\frac{8}{3}
 \quad(483153>472392),
\]
\[
 \left(\frac{11}{9}\right)^{20}>\frac{4096}{81}>50.
\]
Hence the right-hand side at $r=20$ is less than
\[
 \frac{1800}{e^3}\left(\frac{9}{11}\right)^{20}
 <\frac{1800}{18\cdot50}=2.
\]
Finally, the shell and volume ratios give
\begin{align*}
 \frac{\mathscr L(n_0)}{V_r(n_0)}
 &\ge\left(1-\frac{r}{2n_0+1}\right)
          \left(1-\frac{r(r-1)}{2n_0}\right)
 =\left(1-\frac{r}{(r+1)^2}\right)\frac{3}{r+2},\\
 \frac{r\mathscr L(n_0)}{V_r(n_0)}
 &>\frac{9}{10}\cdot\frac{30}{11}=\frac{27}{11}>2.
\end{align*}
The last line uses
$1-r/(r+1)^2>1-1/r\ge19/20>9/10$ and $r/(r+2)\ge10/11$.
Thus $\mathscr U_0/V_r(n_0)<2/r<\mathscr L(n_0)/V_r(n_0)$.
Lemma~\ref{lem:terminal-scaling} gives $(q+1)A=\{0\}$, and
Lemma~\ref{lem:annihilated-terminal} supplies the contradiction.
\end{proof}

Two terminal parameters fall below the coefficient thresholds.
The first is also covered by \cite[Theorem~5.1]{XZ}; scalar capacity
gives a direct proof for both.

\begin{lemma}\label{lem:terminal-two-points}
Neither $B_{24}(6)$ nor $B_{40}(8)$ admits a lattice tiling.
\end{lemma}
\begin{proof}
Both pairs satisfy $c=r$, so $2n+1=(r+1)(q+1)$.
For the number $E$ of minimum fibers,
Lemma~\ref{lem:minimal-fibers} gives
\[
 \frac{\sigma_r-(r-1)\sigma_{r-1}}{r+1}
 \le E\le |A[q+1]|-1.
\]
The shell recurrence and $V_r=\sum_{j=0}^r\sigma_j$ give
\[
\begin{array}{c|r|r|r}
 (r,n)&q+1&V_r&\bigl(\sigma_r-(r-1)\sigma_{r-1}\bigr)/(r+1)\\\hline
 (6,24)&7&19665841&939056\\
 (8,40)&9&46943003137&1432732720.
\end{array}
\]
Here $19665841\equiv6\pmod7$ and
$46943003137\equiv1\pmod3$. Hence $\gcd(q+1,V_r)=1$ in both
cases, so $|A[q+1]|=1$. The displayed inequalities then give
$0<E\le0$, a contradiction.
\end{proof}

The terminal congruence now combines the coefficient bounds and the two
exceptional parameters into one exclusion in the remaining geometric range.

\begin{proposition}\label{prop:terminal-complete}
Suppose $r\ge4$, $n>r(r+1)/2+2$, and $c=r$ in
$2n=(r+1)q+c$, where $0\le c\le r$.
Then $B_n(r)$ admits no lattice tiling of $\Z^n$.
\end{proposition}
\begin{proof}
Suppose that such a tiling exists, and let $A$ be its finite abelian
quotient. Since $c=r$, we have
\[
 2n+1=(r+1)(q+1).
\]
The left-hand side is odd, so both $r+1$ and $q+1$ are odd.
In particular, $r$ is even. Moreover, the dimension hypothesis gives
$2n>r(r+1)+4>r(r+1)$.
If $r\ge20$, Proposition~\ref{thm:terminal-r20} already gives a
contradiction. It remains to consider even $4\le r\le18$.

Because $r+1$ is odd, multiplication by two is invertible modulo
$r+1$. Hence
\[
 2n\equiv r\pmod{r+1}
 \quad\Longleftrightarrow\quad
 n\equiv r/2\pmod{r+1}.
\]
Thus $n=r/2+(r+1)k$ for an integer $k$.
For $r=4$, the dimension hypothesis is $n>12$, so
$2+5k>12$ gives $k\ge3$ and $n\ge17$.
For $r=6$, the hypothesis is $n>23$, so
$3+7k>23$ gives $k\ge3$ and $n\ge24$; if $n\ne24$, then
$k\ge4$ and $n\ge31$.
For even $r\ge8$, the weaker inequality $n>r(r+1)/2$ gives
\[
 k>\frac r2-\frac{r}{2(r+1)}>\frac r2-1.
\]
Since both $k$ and $r/2$ are integers, $k\ge r/2$, and therefore
\[
 n\ge\frac r2+(r+1)\frac r2=\frac{r(r+2)}2.
\]
In particular, for $r=8$ one has $n\ge40$; if $n\ne40$, the
congruence $n\equiv4\pmod9$ gives $n\ge49$.
For $r=10,12,14,16,18$, the lower bounds $r(r+2)/2$ are
$60,84,112,144,180$, respectively.

Lemma~\ref{lem:terminal-two-points} excludes $(r,n)=(6,24)$ and
$(8,40)$. For all remaining pairs with even $4\le r\le18$, the
bounds just established give $n\ge n_0$ for the corresponding row
of Lemma~\ref{lem:terminal-certificates}: the thresholds are
$17,31,49$ for $r=4,6,8$, and $r(r+2)/2$ for
$r=10,12,14,16,18$.
That lemma verifies the hypotheses of Lemma~\ref{lem:terminal-scaling},
which, together with $2n+1=(r+1)(q+1)$, implies
$(q+1)A=\{0\}$.
Finally,
\[
 2n>r(r+1)+4>r^2,
\]
so Lemma~\ref{lem:annihilated-terminal} applies and gives a
contradiction. This excludes every terminal remainder in the stated
range.
\end{proof}

\section{Radius three and completion of the proof}\label{sec:small-radii}
Only the zero remainder at radius three remains. We exclude even dimensions at least ten by a cubic character estimate and dimension eight by a known criterion, then complete the main proof.

\subsection{The zero remainder at radius three}
We begin with the character estimate.
\begin{lemma}\label{lem:r3-zero-even}
For $r=3$ and every even $n\ge10$, no lattice tiling exists.
\end{lemma}
\begin{proof}
Suppose such a tiling exists. Write
\[
 V_3(n)=\frac{4n^3+6n^2+8n+3}{3}.
\]
For every nonprincipal character, the radius-three identity
of Zhang and Zhou~\cite[Lemma~2.4]{ZZ} gives
$\lambda_\chi^3+3\lambda_\chi(\lambda_{\chi^2}-2n)+2\lambda_{\chi^3}=0$.
If $3\mid V_3(n)$, choose a character $\chi$ of order three and define
$k_\chi=|\{i:\chi(g_i)\ne1\}|$. For a coordinate with
$\chi(g_i)=1$, the sum $\chi(g_i)+\chi(-g_i)$ is $2$; otherwise it
is $-1$. Consequently
$\lambda_\chi=\lambda_{\chi^2}=2n+1-3k_\chi\in\Z$ and
$\lambda_{\chi^3}=2n+1$.
Putting $t=3\lambda_\chi-2$ yields
\[
 (6\lambda_\chi-4)n=\lambda_\chi^3+3\lambda_\chi^2+2,
 \qquad54nt=t^3+15t^2+48t+98,
 \qquad t\mid98,\quad t\equiv1\pmod3.
\]
All possible divisors and resulting values of $n$ are
\[
\begin{array}{c|rrrrrr}
t&1&7&49&-2&-14&-98\\\hline
n&3&4&59&-1/2&1/2&303/2.
\end{array}
\]
None is an even integer at least ten. Hence $3\nmid V_3(n)$.
For even $n\ge8$, Lemma~\ref{lem:zero-kernel} makes
$q+1=(n+2)/2$ invertible. Also $2q+1=n+1$ is invertible;  since
$
 3V_3(n)=4(2q+1)^3-6(2q+1)^2+8(2q+1)-3
$, we have $\gcd(2q+1,V_3(n))\mid3
$
and $3\nmid V_3(n)$. Choose a nonprincipal character $\chi$ with
$|\lambda_\chi|=\Lambda$. Since $\gcd(6,V_3(n))=1$, both $\chi^2$ and
$\chi^3$ are nonprincipal. The oddness of $X^3-6nX$ gives
$|\lambda_\chi^3-6n\lambda_\chi|=|\Lambda^3-6n\Lambda|$.
The cubic identity and $|\lambda_{\chi^2}|,|\lambda_{\chi^3}|\le\Lambda$ give
$
 |\Lambda^3-6n\Lambda|\le3\Lambda^2+2\Lambda.
$
If $\Lambda=0$, the desired upper bound is immediate. If $\Lambda>0$,
then $\Lambda^3-6n\Lambda\le3\Lambda^2+2\Lambda$, so division by
$\Lambda$ gives
\[
 \Lambda^2-3\Lambda-6n-2\le0,
 \qquad \Lambda\le\frac{3+\sqrt{24n+17}}{2}.
\]
Since $\sigma_2=2n^2$ and $V_2=2n^2+2n+1$, for $n\ge10$ we have
\begin{gather*}
 (4n-5)^2-(24n+17)=8(2n^2-8n+1)>0,\\
 1\le\Lambda+1\le \frac{5+\sqrt{24n+17}}2<2n,
 \qquad\frac{V_2-\sigma_2}{\sigma_2}=\frac{2n+1}{2n^2}<\frac{1}{2}.
\end{gather*}
Since $(2q+1)A=A$, the general scalar-image comparison
\eqref{eq:full-image-necessary} gives
\begin{align*}
 2n-2
 &\le\frac{7}{2}+\frac{2nV_2}{V_3}
   +\frac{5+\sqrt{24n+17}}2\left(\frac{\sigma_2+V_2}{2\sigma_2}
                  -\frac{\sigma_2^2+V_2^2}{2V_3\sigma_2}\right)
 <\frac{13}{2}+\frac{53}{100}(5+\sqrt{24n+17}).
\end{align*}
Indeed,
\[
 3V_3-2nV_2=2n^2+6n+3>0,
 \qquad
 \frac{\sigma_2+V_2}{2\sigma_2}
 =1+\frac{1}{2n}+\frac{1}{4n^2}
 \le\frac{421}{400}<\frac{53}{50}.
\]
At $n=10$ the last upper bound is less than
\[
 \frac{13}{2}+\frac{53}{100}\left(5+\frac{81}{5}\right)
 =\frac{2217}{125}<18=2\cdot10-2.
\]
Its derivative is
$159/(25\sqrt{24n+17})<159/400<1$, whereas
$(2n-2)'=2$. Thus the contradiction holds for every even $n\ge10$.
This proves the assertion.
\end{proof}

The missing dimension is covered by Zhang and Zhou's
criterion~\cite[Corollary~4.5]{ZZ}: here $8\equiv1\pmod7$ and
$24\cdot8+1=193$ is not a square.
\begin{proposition}[{\cite[Corollary~4.5]{ZZ}}]
\label{thm:r3-eight-theoretical}
The Lee ball $B_8(3)$ admits no lattice tiling of $\Z^8$.
\end{proposition}

\subsection{Completion of the proof}
The three remainder exclusions now give the main theorem directly.

\begin{proof}[Proof of Theorem~\ref{thm:main}]
Suppose that $B_n(r)$ tiles $\Z^n$ by a lattice, with $n\ge3$ and
$r\ge2$. Lemma~\ref{lem:known-radius-two} excludes $r=2$.
For $r\ge3$, Proposition~\ref{thm:uniform-triangular-orbits} leaves
exactly the range~\eqref{eq:geometric-range}. Recall that
$2n=(r+1)q+c$, where $0\le c\le r$.

If $1\le c\le r-1$, Proposition~\ref{prop:positive-complete}
gives a contradiction. If $c=r$, then
$2n+1=(r+1)(q+1)$ is odd, so $r$ is even and hence $r\ge4$;
Proposition~\ref{prop:terminal-complete} excludes this case.
If $c=0$ and $r\ge4$, Proposition~\ref{prop:zero-complete} applies.
Finally, if $c=0$ and $r=3$, then $2n=4q$, so $n\ge8$ is even.
Proposition~\ref{thm:r3-eight-theoretical} excludes $n=8$, and
Lemma~\ref{lem:r3-zero-even} excludes every even $n\ge10$.
These cases exhaust $0\le c\le r$, completing the proof.
\end{proof}

\begin{proof}[Proof of Corollary~\ref{cor:abelian-moore}]
Let $A$ be the vertex group of a simple undirected abelian Cayley
graph of degree $2n$ and diameter at most $r$. Its generating set
consists of $m$ inverse pairs and $2n-2m$ involutions. Every vertex
has a word of length at most $r$, so
\[
 |A|\le[z^r]\frac{(1+z)^{2n-2m}}{1-z}
                  \left(\frac{1+z}{1-z}\right)^m.
\]
Indeed, an inverse pair contributes an arbitrary integral exponent,
with length generating function $(1+z)/(1-z)$, whereas an involution
need only occur zero or one time and contributes $1+z$.
If $m<n$, replace each pair of involution factors $(1+z)^2$ by
$(1+z)/(1-z)$. This increases coefficients, since
\[
 \frac{1+z}{1-z}-(1+z)^2
 =z^2\frac{1+z}{1-z}.
\]
The common factor $(1-z)^{-1}$ makes the increase strict in every
degree $r\ge2$. Hence $|A|<V_r(n)$ whenever an involution occurs.

If $m=n$, choose one generator from each inverse pair. The resulting
homomorphism $\Z^n\to A$ maps $B_n(r)$ onto $A$. Equality
$|A|=V_r(n)$ would make this restriction bijective, and its kernel
would give a lattice tiling by $B_n(r)$, contrary to
Theorem~\ref{thm:main} for $n\ge3$ and $r\ge2$.
Thus $|A|\le V_r(n)-1$ in that range.
For $n\le2$ or $r=1$, the classical lattice tilings~\cite[Theorems~1 and~3]{GW}
give quotient Cayley graphs attaining $V_r(n)$ vertices.
Their signed coordinate images are distinct and nonzero,
since a tiling kernel has no nonzero vector of norm at most $2r$;
thus these graphs are simple and have degree $2n$.
\end{proof}

\section*{Acknowledgments}
The authors used ChatGPT to assist with exploring proof strategies,
checking intermediate calculations, and editing the exposition.
All results and proofs were written by the authors, who
take full responsibility for the
content of this work.

\end{document}